\documentclass{LMCS}

\def\dOi{11(2:9)2015}
\lmcsheading%
{\dOi}
{1--72}
{}
{}
{Jun.~27, 2012}
{Jun.~19, 2015}
{}

\ACMCCS{[{\bf Theory of computation}]: Theory and algorithms for
  application domains---Algorithmic game theory and mechanism
  design---Algorithmic game theory\,/\,Exact and approximate
  computation of equilibria; Logic; [{\bf Software and its
    engineering}]: Software organization and properties---Software
  functional properties---Formal methods---Model checking}


\def\figurename{Figure}
\usepackage[small]{complexity}
\usepackage{wrapfig,subfig,multirow}
\usepackage{nash,nash-draft}
\usepackage{booktabs}

\usepackage{microtype,hyperref}

\begin{document}
\title[Pure Nash equilibria in concurrent deterministic games]{Pure Nash
  equilibria\\ in concurrent deterministic games\rsuper*}
\author[P.~Bouyer]{Patricia~Bouyer\rsuper a}%
\address{{\lsuper{a,c}}LSV -- CNRS \& ENS Cachan -- France}%
\email{\{bouyer,markey\}@lsv.ens-cachan.fr}%

\author[R.~Brenguier]{Romain~Brenguier\rsuper b}%
\address{{\lsuper b}Universit\'e Libre de Bruxelles -- Belgium}%
\email{romain.brenguier@ulb.ac.be}%

\author[N.~Markey]{Nicolas~Markey\rsuper c}%
\address{\vspace{-18 pt}}

\author[M.~Ummels]{Michael~Ummels\rsuper d}%
\address{{\lsuper d}Institute of Transportation Systems, German Aerospace Center -- Germany}%
\email{michael.ummels@dlr.de}%
\titlecomment{{\lsuper*}This article is an extended version of several works that
  appeared as~\cite{BBM10a,BBMU11,BBMU12}. Most of the work reported here was
  done while the second and fourth authors were students at LSV.}
\thanks{This work has been partly supported by 
  ESF LogiCCC project GASICS, 
  ERC Starting Grant EQualIS (308087), 
  ERC Starting Grant inVEST (279499), and 
  EU FP7 project Cassting (601148)}%

\begin{abstract}
  We study pure-strategy Nash equilibria in multi-player concurrent
  deterministic games, for a variety of preference relations. We provide a
  novel construction, called the suspect game, which transforms a
  multi-player concurrent game into a two-player turn-based game which turns
  Nash equilibria into winning strategies (for some objective that depends on
  the preference relations of the players in the original game). We use that
  transformation to design algorithms for computing Nash equilibria in finite
  games, which in most cases have optimal worst-case complexity, for large
  classes of preference relations. This includes the purely qualitative
  framework, where each player has a single $\omega$-regular objective that
  she wants to satisfy, but also the larger class of semi-quantitative
  objectives, where each player has several $\omega$-regular objectives
  equipped with a preorder (for~instance, a~player may want to satisfy all her
  objectives, or to maximise the number of objectives that she achieves.)
\end{abstract}

\maketitle

\section{Introduction}

Games (and especially games played on graphs) have been intensively
used in computer science as a powerful way of modelling interactions
between several computerised
systems~\cite{Thomas-cav02,Henzinger-tark05}. Until recently, more
focus had been put on the study of purely antagonistic games
(a.k.a.~zero-sum games), which conveniently represent systems evolving
in a (hostile) environment. In~this zero-sum games setting, the
objectives of both players are opposite: the aim of one player is to
prevent the other player from achieving her own objective.

Over the last ten years, games with non-zero-sum objectives have come
into the picture: they~allow for conveniently modelling complex
infrastructures where each individual system tries to fulfil its own
objectives, while still being subject to uncontrollable actions of the
surrounding systems. As an example, consider a wireless network in
which several devices try to send data: each device can modulate its
transmit power, in order to maximise its bandwidth and reduce energy
consumption as much as possible. In~that setting, focusing only on
optimal strategies for one single agent may be too narrow.
Game-theoreticians have defined and studied many other solution
concepts for such settings, of which Nash equilibrium~\cite{nash50} is
the most prominent.  A~Nash equilibrium is a strategy profile where no
player can improve the outcome of the game by unilaterally changing
her strategy. In~other terms, in a~Nash equilibrium, each individual
player has a satisfactory strategy. Notice that Nash equilibria need
not exist or be unique, and are not necessarily optimal: Nash
equilibria where all players lose may coexist with more interesting
Nash equilibria. Finding \emph{constrained} Nash equilibria (\eg,
equilibria in which some players are required to win) is thus an
interesting problem for our setting.

In this paper, we report on our recent contributions on the
computation of Nash equilibria in concurrent games (preliminary works
appeared as~\cite{BBM10a,BBMU11,BBMU12}).  Concurrent games played on
graphs are a general model for interactive systems, where the agents
take their decision simultaneously. Therefore concurrent games subsume
turn-based games, where in each state, only one player has the
decision for the next move. One motivation for concurrent games is the
study of \emph{timed games} (which are games played on timed
automata~\cite{AMPS98,AFH+03}): the~semantics of a timed game is
naturally given as a concurrent game (the~players all choose
simultaneously a delay and an action to play, and the player with the
shortest delay decides for the next move---this~mechanism cannot be
made turn-based since we cannot fix \textit{a~priori} the player who
will choose the smallest delay); the region-based game
abstraction which preserves Nash equilibria also requires the formalism of
concurrent games~\cite{BBMU11,brenguier12}.  Multi-agent
infrastructures can be viewed as distributed systems, which can
naturally be modelled as concurrent games.

\subsection*{Our contributions}  
The paper focuses on concurrent deterministic games and on \textit{pure} Nash
equilibria, that is, strategy profiles which are deterministic (as
opposed to randomised). In this work we assume strategies only depend
on the set of states which is visited, and not on the actions that
have been played. This is a partial-information hypothesis which we
believe is relevant in the context of distributed systems, where only
the effect of the actions can be seen by the players. We will discuss
in more detail all these choices in the conclusion.

In the context exposed above, we develop a complete methodology for
computing pure Nash equilibria in (finite) games. First, in
Section~\ref{sec:suspect}, we propose a novel transformation of the
multi-player concurrent game (with a preference relation for each
player) into a two-player zero-sum turn-based game, which we call the
\emph{suspect game}. Intuitively, in the suspect game, one of the
players suggests a global move (one action per player of the original
game), with the aim to progressively build a Nash equilibrium; while
the second player aims at proving that what the first player proposes
is \emph{not} a Nash equilibrium. This transformation can be applied
to arbitrary concurrent games (even those with infinitely many states)
and preference relations for the players, and it has the property that
there is a correspondence between Nash equilibria in the original game
and winning strategies in the transformed two-player turn-based
game. The winning condition in the suspect game of course depends on
the preference relations of the various players in the original game.

Then, using that construction we develop (worst-case)
optimal-complexity algorithms for deciding the existence of
(constrained) Nash equilibria in \textit{finite} games for various
classes of preference relations. In Section~\ref{sec:single}, we focus
on qualitative $\omega$-regular objectives, \ie, preference relations
are given by single objectives (which can be reachability, B\"uchi,
parity, etc), and it is better for a player to satisfy her objective
than to not satisfy her objective. We prove the whole set of results
which are summarised in the second column of Table~\ref{table-single}
(the first column summarises the complexity in the zero-sum two-player
setting~--~called the value problem). Among the results obtained this
way, the constrained Nash equilibrium existence problem is
\NP-complete in finite games with single reachability or safety
objectives, while it is \PTIME-complete for single B\"uchi objectives.

\begin{table}[t]
  \centering
\def\arraystretch{1.1}
  \begin{tabular}{@{}r||c|c@{}}
  \hline
    Objective & Value 
    & (Constrained) Existence of Nash Eq.\\
    \hline
    Reachability & \P-c. \cite{McNaughton93}  &
    \NP-c. (Sect.~\ref{subsec:reachability}) \\
    Safety & \P-c. \cite{McNaughton93} & \NP-c. (Sect.~\ref{subsec:safety}) \\
    B\"uchi  & \P-c. \cite{McNaughton93} & \P-c. (Sect.~\ref{subsec:buchi}) \\
    co-B\"uchi  & \P-c. \cite{McNaughton93}& \NP-c. (Sect.~\ref{subsec:cobuchi}) \\
    Parity & \UP $\cap$ \co-\UP  \cite{Jurdzinski98} &
    $\P^\NP_\parallel$-c.\footnotemark (Sect.~\ref{subsec:rabin}) \\
    Streett & \co-\NP-c. \cite{emerson1988complexity}& $\P^\NP_\parallel$-h. and
      in \PSPACE \\
    Rabin & \NP-c. \cite{emerson1988complexity} & $\P^\NP_\parallel$-c.
    (Sect.~\ref{subsec:rabin}) \\ 
    Muller & \PSPACE-c. \cite{Hunter07}
    & \PSPACE-c. \\ 
    Circuit & \PSPACE-c. \cite{Hunter07}
    & \PSPACE-c. (Sect.~\ref{subsec:circuits}) \\
    Det. B\"uchi Automata & \P-c. & \PSPACE-h.
    (Sect.~\ref{sec:rabin-auto}) and in \EXPTIME \\
    Det. Rabin Automata & \NP-c. & \PSPACE-h. and in \EXPTIME
    (Sect.~\ref{sec:rabin-auto}) \\
    \hline
  \end{tabular}
  \caption{Summary of the complexities for single
    objectives}\label{table-single}

  \begin{tabular}{@{}r|c|c|c@{}}
    \hline
    Preorder & Value & Existence of NE & Constr.~Exist. of~NE \\
    \hline
    Maximise, Disj. & 
      \P-c. (Sect.\ref{subsec:reducible})
    & \P-c. (Sect.\ref{subsec:reducible})
    & \P-c. (Sect.\ref{subsec:reducible})
      \\
     Subset  &  
     \P-c. (Sect.~\ref{sec:reduc-buchi-auto}) & 
     \P-c. (Sect.\ref{subsec:reducible}) & 
     \P-c. (Sect.\ref{subsec:reducible})
     \\
    Conj., Lexicogr. & 
    \P-c. (Sect.~\ref{sec:reduc-buchi-auto}) &
    \P-h., in \NP
    (Sect.~\ref{subsec:monotonic})  & 
    \NP-c. (Sect.~\ref{subsec:monotonic}) 
    \\
    Counting & 
    \coNP-c.  (Sect.~\ref{subsec:monotonic})  & 
    \NP-c.  (Sect.~\ref{subsec:monotonic})  & 
    \NP-c.   (Sect.~\ref{subsec:monotonic}) 
    \\
    Mon.~Bool.~Circuit & 
    \coNP-c.  (Sect.~\ref{subsec:monotonic})  &
    \NP-c.  (Sect.~\ref{subsec:monotonic})  & 
    \NP-c.    (Sect.~\ref{subsec:monotonic}) 
    \\
    Boolean~Circuit & 
    \PSPACE-c. (Sect.~\ref{subsec:general}) &
    \PSPACE-c. (Sect.~\ref{subsec:general}) & 
    \PSPACE-c. (Sect.~\ref{subsec:general}) 
    \\
    \hline
  \end{tabular}
  \caption{Summary of the results for ordered B\"uchi objectives}\label{table-buchi}

  \begin{tabular}{@{}r|c|c@{}}
  \hline
    Preorder & Value & (Constrained) Exist. of NE \\
    \hline
    Disjunction, Maximise & \P-c. (Sect.~\ref{reach-simple}) &  \NP-c.
    (Sect.~\ref{reach-simple}) \\
    Subset  & \PSPACE-c. (Sect.~\ref{ssec-generalcase}) & \NP-c.
    (Sect.~\ref{reach-simple}) \\
    Conjunction, Counting, Lexicogr. & \PSPACE-c.
    (Sect.~\ref{ssec-generalcase}) & \PSPACE-c. (Sect.~\ref{ssec-generalcase})
    \\
    (Monotonic) Boolean Circuit & \PSPACE-c. (Sect.~\ref{ssec-generalcase})
    & \PSPACE-c. (Sect.~\ref{ssec-generalcase})   \\
    \hline
  \end{tabular}
  \caption{Summary of the results for ordered reachability objectives}\label{table-reach}

\end{table}
\footnotetext{\label{fn-pnp||}%
  The complexity class $\P^\NP_\parallel$ is defined in terms of
  Turing machine having access to an oracle; oracle are artificial devices
  that can solve a problem in constant time, thus hiding part of the
  complexity of the overall problem. The class $\P^\NP$ is the class of
  problems that can be solved in polynomial time by a deterministic Turing
  machine which has access to an oracle for solving \NP problems. The class
  $\P^\NP_\parallel$ is the subclass where, instead of asking a sequence of
  (dependent) queries to the oracle, the Turing machine is only allowed to ask
  one set of queries. We~refer to~\cite{Pap94,Wag88} for more details.}

In Sections~\ref{sec:buchi} and~\ref{sec:reach}, we~extend the
previous qualitative setting to the \textit{semi-quantitative} setting
of ordered objectives. An ordered objective is a set of B\"uchi (or
reachability) objectives and a preorder on this set. The preference
relation given by such an ordered objective is then given by the value
of the plays (w.r.t. the objectives) in that preorder. Preorders of
interest are for instance conjunction, disjunction, lexicographic
order, counting preorder, maximise preorder, subset preorder, or more
generally preorders given as Boolean circuits. We provide algorithms
for deciding the existence of Nash equilibria for ordered objectives,
with (in most cases) optimal worst-case complexity. These algorithms
make use of the suspect-game construction.
The results are listed in Table~\ref{table-buchi} for B\"uchi
objectives and in Table~\ref{table-reach} for reachability objectives.

\subsection*{Examples}
Back to the earlier wireless network example, we can model a~simple
discretised version~of~it as follows. From a state, each device can
increase (action~$1$) or keep unchanged (action~$0$) its power: the
arena of the game is represented for two devices and two levels of
power on Figure~\ref{fig-network} (labels of states are power levels).
This yields a new bandwidth allocation (which depends on the
degradation due to the other devices) and a new energy
consumption. The satisfaction of each device is measured as a
compromise between energy consumption and bandwidth allocated, and it
is given by a quantitative payoff function.\footnote{The
  (quantitative) payoff for player $i$ can be expressed by $\payoff_i
  = \frac{R}{\mathsf{power}_i} \Big(1- e^{-0.5 \gamma_i}\Big)^L$ where
  $\gamma_i$ is the signal-to-interference-and-noise ratio for player
  $i$, $R$ is the rate at which the wireless system transmits the
  information in bits per seconds and $L$ is the size of the packets
  in bits (\cite{SMG99}).} This can be transformed into B\"uchi
conditions and a preorder on them.  There are basically two families
of pure Nash equilibria in this system: the one where the two players
choose to go and stay forever in state $(1,1)$; and the one where the
two players go to state $(2,2)$ and stay there forever.

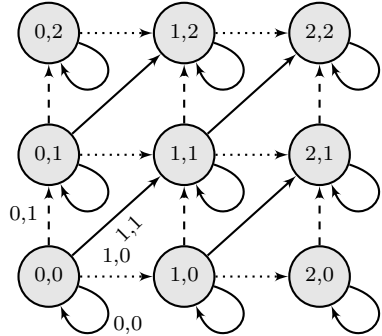
\begin{figure}[t]
 \centering
 \begin{tikzpicture}[scale=.9,yscale=.9]
   \def\plus{\ensuremath{1}\xspace}
   \def\egal{\ensuremath{0}\xspace}
    \everymath{\scriptstyle}
    \tikzset{noeud/.style={circle,draw=black,thick,fill=black!10,minimum height=8mm}}
    \draw (0,0) node [noeud] (00) {$0,0$};
    \draw (0,2) node [noeud] (01) {$0,1$};
    \draw (2,0) node [noeud] (10) {$1,0$};
    \draw (0,4) node [noeud] (02) {$0,2$};
    \draw (2,2) node [noeud] (11) {$1,1$};
    \draw (4,0) node [noeud] (20) {$2,0$};
    \draw (2,4) node [noeud] (12) {$1,2$};
    \draw (4,2) node [noeud] (21) {$2,1$};
    \draw (4,4) node [noeud] (22) {$2,2$};

    \draw [-latex',thick,dashed] (00) -- (01) node [pos=.5,left]
          {$\egal,\plus$};

    \draw [-latex',thick,dashed] (01) -- (02);
    \draw [-latex',thick,dashed] (10) -- (11);
    \draw [-latex',thick,dashed] (11) -- (12);
    \draw [-latex',thick,dashed] (21) -- (22);
    \draw [-latex',thick,dashed] (20) -- (21);
    
    \draw [-latex',thick,dotted] (00) -- (10) node [pos=.5,above,sloped]
          {$\plus,\egal$}; 
    \draw [-latex',thick,dotted] (01) -- (11);
    \draw [-latex',thick,dotted] (02) -- (12);
    \draw [-latex',thick,dotted] (10) -- (20);
    \draw [-latex',thick,dotted] (11) -- (21);
    \draw [-latex',thick,dotted] (12) -- (22);

    \draw [-latex',thick] (00) -- (11) node [midway,sloped,below] {$\plus,\plus$};
    \draw [-latex',thick] (11) -- (22);
    \draw [-latex',thick] (01) -- (12);
    \draw [-latex',thick] (10) -- (21);

    \draw [-latex',thick] (00) .. controls +(-20:1.5) and +(-70:1.5) .. (00) node
          [pos=.5,right] {$\egal,\egal$}; 

    \draw [-latex',thick] (01) .. controls +(-20:1.5) and +(-70:1.5)
    .. (01); 
    \draw [-latex',thick] (02) .. controls +(-20:1.5) and +(-70:1.5)
    .. (02); 
    \draw [-latex',thick] (10) .. controls +(-20:1.5) and +(-70:1.5)
    .. (10); 
    \draw [-latex',thick] (11) .. controls +(-20:1.5) and +(-70:1.5)
    .. (11); 
    \draw [-latex',thick] (12) .. controls +(-20:1.5) and +(-70:1.5)
    .. (12); 
    \draw [-latex',thick] (20) .. controls +(-20:1.5) and +(-70:1.5)
    .. (20); 
    \draw [-latex',thick] (21) .. controls +(-20:1.5) and +(-70:1.5)
    .. (21); 
    \draw [-latex',thick] (22) .. controls +(-20:1.5) and +(-70:1.5)
    .. (22); 

 \end{tikzpicture}
  \caption{A simple game-model for the wireless network}
  \label{fig-network}
\end{figure}

We describe another example, the medium access control, that involves
qualitative objectives. It was first given a game-theoretic model
in~\cite{MW03}. Several users share the access to a wireless
channel. During each slot, they can choose to either transmit or wait
for the next slot. If too many users are emitting in the same slot,
then they fail to send data. Each attempt to transmit costs energy to
the players. They have to maximise the number of successful attempts
using the energy available to them. We give in Figure~\ref{fig-mac} a
possible model for that protocol for two players and at most one
attempt per player and a congestion of $2$ (that is, the two players
should not transmit at the same time): each state is labelled with the
energy level of the two players, and the number of successful attempts
of each of the player.  There is several Nash equilibria, and they
give payoff $1$ to every player: it consists in going to state
$(0,1,0,1)$ by not simultaneously transmitting.

\begin{figure}[t]
 \centering
 \begin{tikzpicture}[xscale=1.6,yscale=.9]
   \def\transmit{1\xspace}
   \def\wait{0\xspace}

   \tikzset{round/.style={rounded
       corners=2mm,draw=black,thick,fill=black!10,minimum height=8mm}}
   \everymath{\scriptstyle}
    \draw (0,0.2) node [round] (00) {$1,0,1,0$};
    \draw (0,2) node [round] (01) {$1,0,0,1$};
    \draw (2,0.2) node [round] (10) {$0,1,1,0$};
    \draw (2,2) node [round] (11) {$0,1,0,1$};
    \draw (-45:1.5) node [round] (20) {$0,0,0,0$};

    \draw [-latex',thick] (00) -- node [left] {$\wait,\transmit$} (01) ;
    \draw [-latex',thick] (00) -- node [above] {$\transmit,\wait$}(10) ; 
    \draw [-latex',thick] (00) -- node [above,sloped] {$\transmit,\transmit$} (20) ; 

    \draw [-latex',thick] (10) -- node [right] {$\wait,\transmit$} (11) ;
    \draw [-latex',thick] (01) -- node [above] {$\transmit,\wait$} (11) ;

    \draw [-latex',thick] (00) .. controls +(-150:1.2) and +(150:1.2) .. (00) node
          [pos=.5,below left] {$\wait,\wait$}; 

    \draw [-latex',thick] (01) .. controls +(-150:1.2) and +(150:1.2)
    .. (01) node
          [pos=.5,left] {$\wait,\wait$}; 
    \draw [-latex',thick] (10) .. controls +(30:1.2) and +(-30:1.2)
    .. (10) node
          [pos=.5,below right] {$\wait,\wait$}; 
    \draw [-latex',thick] (11) .. controls +(30:1.2) and +(-30:1.2)
    .. (11) node
          [pos=.5,right] {$\wait,\wait$}; 
    \draw [-latex',thick] (20) .. controls +(30:1.2) and +(-30:1.2)
    .. (20) node
          [pos=.5,right] {$\wait,\wait$}; 

 \end{tikzpicture}
  \caption{A simple game-model for the medium access control}
  \label{fig-mac}
\end{figure}
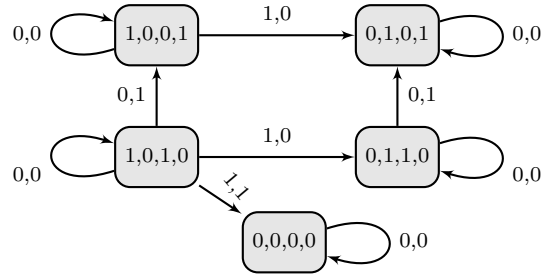

\subsection*{Related work}

Game theory has been a very active area since the 1940's, with the
pioneering works of Von~Neumann, Morgenstern~\cite{MvN47},
Nash~\cite{nash50} and Shapley~\cite{shapley1952value}. It~has had
numerous uses in various domains, ranging from economics to human
sciences and logic.  Equilibria are a central concept in
(non-zero-sum) games, as they are meant to represent rational
behaviours of the players. Many important results about existence of
various kinds of equilibria in different kinds of games have been
established~\cite{MvN47,nash50,fink64}.

For applications in logic and computer science, games played on graphs
have received more focus; also, computer scientists have been mostly
looking for algorithmic solutions for deciding the existence and
effectively computing equilibria and
$\epsilon$-equilibria~\cite{CMJ04,chatterjee05,ummels08}.

For two-player concurrent games with B\"uchi objectives, the existence of
$\epsilon$-equilibria (in~randomised strategies) was proved by
Chatterjee~\cite{chatterjee05}. However, exact Nash equilibria need not exist;
turn-based games with B\"uchi objectives are an~important subclass where Nash
equilibria (even in pure strategies) always exist~\cite{CMJ04}. When they
exist, Nash equilibria need not be unique; equilibria where all the players
lose can coexist with equilibria where some (or~all) of them win. Ummels
introduced \emph{constrained} Nash equilibria, \ie, Nash equilibria where some
players are required to win. In~particular, he~showed that the existence of
constrained Nash equilibria can be decided in polynomial time for turn-based
games with B\"uchi objectives \cite{ummels08}. In~this paper, we extend this
result to concurrent games, and to various classes of $\omega$-regular winning
objectives. For concurrent games with $\omega$-regular objectives, the
decidability of the constrained Nash equilibrium existence problem \wrt pure
strategies was established by Fisman~\ea~\cite{FKL10}, but their algorithm
runs in doubly exponential time, whereas our algorithm runs in exponential
time for objectives given as B\"uchi automata. Finally, Ummels and
Wojtczak~\cite{UW11} proved that the existence of a Nash equilibrium in pure
or randomised strategies is undecidable for \emph{stochastic} games with
reachability or B\"uchi objectives, which justifies our restriction to
concurrent games without probabilistic transitions. They also proved a similar
undecidability result for randomised Nash equilibria in non-stochastic
games~\cite{UW11a}, hence we consider only pure-strategy Nash equilibria.

Several solution concepts have been defined and studied for games on graphs.
In~particular, \emph{secure equilibria}~\cite{CHJ05b,DFKSV14} are Nash
equilibria where besides satisfying their primary objectives, the players try
to prevent the other players from achieving their own (primary) objectives.
Notice that our results in Sect.~\ref{subsec:monotonic} and
Sect.~\ref{ssec-generalcase} do apply to such kinds of lexicographic
combination of several objectives.

Temporal logics can also be used to express properties of games. While
\ATL~\cite{AHK02} can mainly express only zero-sum properties, other logics
such as \ATL with strategy contexts~(\ATLsc)~\cite{DLM10} and \textsf{Strategy
  Logic}~(\SL)~\cite{CHP10,MMV10} can be used to express rich properties in a
non-zero-sum setting.
In~terms of complexity however, model checking for such logics has high
complexity: Nash equilibria can be expressed using one quantifier alternation
(an existential quantification over strategy profiles followed with a
universal quantification over deviations); model checking this fragment of
\ATLsc or~\SL is \EXPTIME[2]-complete.

\section{Definitions}

\subsection{General definitions}\label{ssec-gendef}
In this section, we fix some definitions and notations.
\paragraph{Preorders.}
We~fix a non-empty set~$P$. A~\newdef{preorder} over~$P$ is a binary
relation~$\mathord\preorder\subseteq P\times P$ that is reflexive and
transitive.  With~a preorder~$\preorder$, we~associate an
\newdef{equivalence relation}~$\sim$ defined so that $a \sim b$ if,
and only if, ${a \preorder b}$ and ${b \preorder a}$.
The~\newdef{equivalence class} of~$a$, written $[a]_\preorder$, is the
set $\{ b \in P \mid a \sim b\}$.  We~also associate with~$\preorder$
a~\newdef{strict partial order}~$\prec$ defined so that ${a \prec b}$
if, and only if, ${a \preorder b}$ and ${b \not\preorder a}$.
A~preorder~$\preorder$ is said \newdef{total} if, for all
elements~$a,b\in P$, either ${a \preorder b}$, or ${b\preorder a}$.
An element $a$ in a subset~$P'\subseteq P$ is said \emph{maximal
  in~$P'$} if there is no $b \in P'$ such that $a \prec b$; it~is said
\emph{minimal in~$P'$} if there is no $b\in P'$ such that ${b \prec
  a}$.
A~preorder is said \newdef{Noetherian} (or \emph{upwards
  well-founded}) if any subset~$P'\subseteq P$ has at least one
maximal element. It is said \newdef{almost-well-founded} if any
lower-bounded subset $P' \subseteq P$ has a minimal element.

\paragraph{Transition systems.}
A \newdef{transition system} is a pair $\calS = \tuple{\Stat,\Edg}$
where $\Stat$ is a set of states and $\Edg \subseteq \Stat \times
\Stat$ is the set of transitions.  A~\newdef{path}~$\pi$ in $\calS$ is
a sequence $(s_i)_{0\leq i < n}$ (where~$n\in\N^+ \cup\{\infty\}$) of
states such that $(s_i,s_{i+1})\in \Edg$ for all~$i\leq n$.
The~\newdef{length} of~$\pi$, denoted by~$\length\pi$, is~$n-1$.
The~set of finite paths (also called \newdef{histories}) of~$\calS$ is
denoted by~$\Hist_\calS$, the set of infinite paths (also called
\newdef{plays}) of~$\calS$ is denoted by~$\Play_\calS$, and
$\Path_\calS = \Hist_\calS\cup \Play_\calS$ is the set of all paths
of~$\calS$.
Given a path $\pi = (s_i)_{0 \le i < n}$ and an integer~$j<n$, the
\newdef{$j$-th prefix} (\resp \newdef{$j$-th suffix}, \newdef{$j$-th
  state}) of~$\pi$, denoted by~$\pref\pi j$ (\resp $\pi_{\ge j}$,
$\pi_{=j}$), is the finite path~$(s_i)_{0\leq i<j+1}$ (\resp the path
$(s_{j+i})_{0 \le i <n-j}$, the state~$s_j$). If $\pi = (s_i)_{0\leq
  i< n}$ is a history, we write $\last(\pi) = s_{\length\pi}$ for the
last state of~$\pi$.  If~$\pi'$ is a path such that
$(\last(\pi),\pi'_{=0})\in\Edg$, then the
\newdef{concatenation}~$\pi\cdot \pi'$ is the path~$\rho$ s.t.
$\rho_{=i}=\pi_{=i}$ for~$i\leq \length\pi$ and
$\rho_{=i}=\pi'_{=(i-1-\length\pi)}$ for~$i>\length\pi$.  In~the
sequel, we~write $\Hist_\calS(\stat)$, $\Play_{\calS}(\stat)$
and~$\Path_\calS(\stat)$ for the respective subsets of paths starting
in state~$\stat$.  If~$\pi$~is a~play, $\Occ(\pi) = \{ s \mid \exists
j.\ \pi_{=j} =s \}$ is the sets of states that appears at least once
along~$\pi$ and $\Inf(\pi) = \{ s \mid \forall i.\ \exists j \ge i.\
\pi_{=j} =s \}$ is the set of states that appears infinitely often
along~$\pi$.

\subsection{Concurrent games}

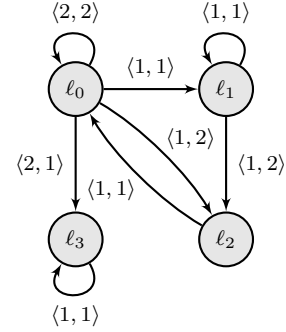
\begin{wrapfigure}r{5cm}
\hfill
\begin{minipage}{4.9cm}
\centering
\begin{tikzpicture}[scale=1,thick]
\tikzset{every node/.style={font=\scriptsize}}
\tikzset{rond/.style={circle,draw=black,thick,fill=black!10,minimum height=7mm}}
\draw (0,0) node[rond] (l0) {$\ell_0$};
\draw (2,0) node[rond] (l1) {$\ell_1$};
\draw (2,-2) node[rond] (l2) {$\ell_2$};
\draw (0,-2) node[rond] (l3) {$\ell_3$};
\path[use as bounding box] (0,.8) -- (0,-3);
\draw[-latex'] (l0) .. controls +(60:10mm) and +(120:10mm) .. (l0)
  node[midway,above] {$\langle 2,2\rangle$};
\draw[-latex'] (l0) -- (l1) node[midway,above] {$\langle 1,1\rangle$};
\draw[-latex'] (l0) -- (l3) node[midway,left] {$\langle
  2,1\rangle$};
\draw[-latex'] (l0) .. controls +(-30:10mm) and +(120:10mm) .. (l2)
  node[midway,above right=-1mm] {$\langle 1,2\rangle$};
\draw[-latex'] (l2) .. controls +(150:10mm) and +(-60:10mm) .. (l0)
  node[midway,below left=-1mm] {$\langle 1,1\rangle$};
\draw[-latex'] (l1) .. controls +(60:10mm) and +(120:10mm) .. (l1)
  node[midway,above] {$\langle 1,1\rangle$};
\draw[-latex'] (l3) .. controls +(-60:10mm) and +(-120:10mm) .. (l3)
  node[midway,below] {$\langle 1,1\rangle$};
\draw[-latex'] (l1) -- (l2) node[midway,right] {$\langle 1,2\rangle$};
\end{tikzpicture}
\caption{Representation of a two-player concurrent game}\label{fig-ex} 
\end{minipage}
\end{wrapfigure}
Our definition of concurrent games extends the definition
in~\cite{AHK02} by allowing for more than two players, each of them
having a preorder over plays.

\begin{definition}
  A \newdef{concurrent game} is a tuple
  $\calG=\tuple{\Stat,\Agt,\Act,\Allow,\Tab,(\mathord\prefrel_A)_{A\in\Agt}
  }$, where $\Stat$ is a finite non-empty set of states, $\Agt$~is a finite
  set of players, $\Act$~is a finite set of actions, and
  \begin{itemize}
  \item $\Allow\colon \Stat \times \Agt \to
    2^\Act\setminus\{\varnothing\}$ is a mapping indicating the
    actions available to a given player in a given state;
  \item $\Tab\colon \Stat\times \Act^{\Agt} \to \Stat$ associates,
    with a given state and a given move of the players (\ie, an
    element of $\Act^\Agt$), the state resulting from that move;
  \item for each~$A\in\Agt$, $\prefrel_A$ is a preorder over
    $\Stat^\omega$, called the preference relation of player~$A$.
  \end{itemize}
\end{definition}

\noindent Figure~\ref{fig-ex} displays an example of a finite concurrent game.
Transitions are labelled with the moves that trigger them. We~say that a
\newdef{move} $\act_\Agt=\indice\act A\Agt\in\Act^\Agt$ is \newdef{legal}
at~$\stat$ if $m_A\in\Allow(\stat,A)$ for all $A\in\Agt$. A~game
is~\newdef{turn-based} if for each state the set of allowed moves is a
singleton for all but at most one player.

In a concurrent game~$\calG$, whenever we arrive at a state~$\stat$,
the players simultaneously select an available action, which results
in a legal move~$\act_\Agt$; the next state of the game is then
$\Tab(\stat,\act_\Agt)$.  The same process repeats
\textit{ad~infinitum} to form an infinite sequence of states.

In the sequel, as no ambiguity will arise, we~may abusively write~$\calG$ for
its underlying transition system~$(\Stat, \Edg)$ where $\Edg = \{ (s,s') \in
\Stat \times \Stat \mid \exists m_{\Agt} \in \prod_{A\in\Agt}
\Allow(s,A)\allowbreak\text{ s.t. }\Tab(s,m_\Agt) = s' \}$. The notions of
paths and related concepts in concurrent games follow from this
identification.

\begin{remark}[Representation of finite games]
  \label{remark:encoding}
  In this paper, for finite games, we will assume an explicit encoding
  of the transition function $\Tab$. Hence, its size, denoted
  $|\Tab|$, is equal to $\sum_{s\in \Stat} \prod_{A\in\Agt}
  |\Mov(s,A)| \cdot \lceil\log(|\Stat|)\rceil$. Note that it can be
  exponential with respect to the number of players.  A~symbolic
  encoding of the transition table has been proposed in~\cite{LMO08},
  in the setting of \textsc{ATL} model checking. This makes the
  problem harder, as the input is more succinct (see
  Remark~\ref{remark:explosion} and
  Proposition~\ref{proposition:explosion} for a formal
  statement). We~would also have a blowup in our setting, and prefer
  to keep the explicit representation in order to be able to compare
  with existing results.  Notice that, as a matter of fact, there is
  no way to systematically avoid an explosion: as~there are
  $|\Stat|^{|\Stat|\cdot|\Act|^{|\Agt|}}$ possible transition
  functions, for any encoding there is one function whose encoding
  will have size at least $\lceil\log(|\Stat|)\rceil \cdot |\Stat|
  \cdot |\Act|^{|\Agt|}$.
  The total size of the game, is then
  \[ |\calG| = |\Stat| + |\Stat|\cdot|\Agt|\cdot|\Act| + \sum_{s\in
    \Stat} \prod_{A\in\Agt} |\Mov(s,A)| \cdot
  \lceil\log(|\Stat|)\rceil + \sum_{A\in\Agt} | \preorder_A |.\] The
  size of a preference relation~$\preorder_A$ will depend on how it is
  encoded, and we will make it precise when it is relevant. This is
  given in Section~\ref{sec:prefrel}.
\end{remark}

\begin{definition}
  Let~$\calG$ be a concurrent game, and~$A\in\Agt$.  A
  \newdef{strategy} for~$A$ is a mapping $\sigma_A\colon \Hist_\calG
  \to \Act$ such that $\sigma_A(\pi) \in \Allow(\last(\pi),A)$ for all
  $\pi\in\Hist_\calG$.
  A~strategy~$\sigma_P$ for a coalition~$P\subseteq \Agt$ is a tuple
  of strategies, one for each player in~$P$. We~write
  $\sigma_P=(\sigma_A)_{A\in P}$ for such a strategy.
  A~\newdef{strategy profile} is a strategy for~$\Agt$.  We~write
  $\Strat\calG P$ for the set of strategies of coalition~$P$, and
  $\Profile_\calG=\Strat\calG \Agt$.
\end{definition}

Note that, in this paper, we only consider \emph{pure} (\ie,
non-randomised) strategies. This is actually crucial in all the
constructions we give (lasso representation in
Subsection~\ref{sec:lasso} and suspect-game construction in
Section~\ref{sec:suspect}).
Notice also that our strategies are based on the sequences of visited states
(they map sequences of states to actions), which is realistic when considering
multi-agent systems. In some settings, it is more usual to base strategies on 
the sequences of actions played by all the players. When dealing with Nash
equilibria, this makes a big difference: strategies based on actions can
immediately detect which player(s) deviated from their strategy; strategies based
on states will only detect deviations because an unexpected state is visited,
without knowing which player(s) is responsible for the deviation. Our
construction precisely amounts to keeping track of a list of \emph{suspects}
for some deviation.

Let~$\calG$ be a game, $P$ a~coalition, and $\sigma_P$ a~strategy
for~$P$. A~path~$\pi$ is \newdef{compatible} with the
strategy~$\sigma_P$ if, for all~$k<\length\pi$, there exists a
move~$\indicebis\act A\Agt$ \st 
\begin{enumerate}
\item $\indicebis\act A\Agt$ is legal at~$\pi_{=k}$,
\item $\act_A = \sigma_A(\pi_{\le k})$ for all~$A\in P$, and
\item $\Tab(\pi_{=k}, \indicebis\act A\Agt) = \pi_{=k+1}$.
\end{enumerate}
We~write~$\Out_{\calG}(\sigma_P)$ for the set of paths (called
the~\emph{outcomes}) in~$\calG$ that are compatible with
strategy~$\sigma_P$ of~$P$.  We~write $\FOut_{\calG}$ (\resp
$\IOut_{\calG}$) for the finite (\resp infinite) outcomes, and
$\Out_\calG(\stat,\sigma_P)$, $\FOut_\calG(\stat,\sigma_P)$ and
$\IOut_\calG(\stat,\sigma_P)$ for the respective sets of outcomes
of~$\sigma_P$ with initial state~$\stat$.  Notice that any strategy
profile has a single infinite outcome from a given state. In the
sequel, when given a strategy profile~$\sigma_\Agt$, we~identify
$\Out(\stat, \sigma_\Agt)$ with the unique play it contains.

A~concurrent game involving only two players ($A$~and~$B$,~say) is
\newdef{zero-sum} if, for any two plays~$\pi$ and~$\pi'$, it~holds
$\pi\prefrel_A \pi'$ if, and only~if, $\pi'\prefrel_B \pi$. Such a
setting is purely antagonistic, as both players have opposite
objectives. The most relevant concept in such a setting is that of
\emph{winning strategies}, where the aim is for one player to achieve
her objectives \emph{whatever the other
  players~do}. In~\newdef{non-zero-sum} games, winning strategies are
usually too restricted, and the most relevant concepts are
\emph{equilibria}, which correspond to strategies that \emph{satisfy} (which
can be given several meanings) all the players. One of the most
studied notion of equilibria is \emph{Nash equilibria}~\cite{nash50},
which we now introduce.

\subsection{Nash equilibria}

We begin with introducing some vocabulary.  When $\pi \prefrel_A
\pi'$, we say that $\pi'$ is \newdef{at~least as good} as~$\pi$
for~$A$.
We say that a strategy~$\sigma_A$ for~$A$ \newdef{ensures}~$\pi$ if
every outcome of~$\sigma_A$ is at least as good as~$\pi$ for~$A$, and
that $A$ \newdef{can ensure}~$\pi$ when such a strategy exists.

Given a move~$\indicebis mA\Agt$ and an action~$m'$ for some
player~$A$, we~write ${\replaceter m A {m'}}$ for the move~$\indicebis
nA\Agt$ with $n_B=m_B$ when~$B\not=A$ and $n_A=m'$.  This~is extended
to strategies in the natural way.

\begin{definition}\label{def-NE}
  Let~$\calG$ be a concurrent game and let $\stat$ be a state
  of~$\calG$.  A~\newdef{Nash equilibrium} of~$\calG$ from~$\stat$ is
  a strategy profile $\sigma_\Agt \in \Profile_\calG$ \st
  $\Out(\stat,\replaceter \sigma A {\sigma'}) \prefrel_A
  \Out(\stat,\sigma_\Agt)$ for all players $A\in\Agt$ and all
  strategies $\sigma'\in\Strat{}A$.
\end{definition}

\begin{wrapfigure}r{5.4cm}
\centering
\begin{tikzpicture}[yscale=-1,scale=1.2,thick,minimum height=5mm]
\draw[dotted,rounded corners=4mm,fill=black!30!white] (0,.5) -- (1.5,.5) --
(1.5,2.4) -- (.4,3.5) -- (-.4,3.5) -- (-1.5,2.4) -- (-1.5,1.5) -- (-.5,1.5) --
(-.5,.5) -- (0,.5);
\draw[dotted,rounded corners=2mm,fill=black!10!white,opacity=.8] (-1,.7) -- (-.7,.7) --
(-.7,1.7) -- (.3,1.7) -- (.3,3.3) -- (-.2,3.3) -- (-1.3,2.2) -- (-1.3,.7) -- (-1,.7);
\draw (0,0) node[circle,draw,fill=white] (000) {};
\draw (-1,1) node[circle,draw,fill=black] (001) {};
\draw (0,1) node[circle,draw,fill=white] (010) {};
\draw (1,1) node[circle,draw,fill=white] (100) {};
\draw (-1,2) node[circle,draw,fill=white] (011) {};
\draw (0,2) node[circle,draw,fill=white] (101) {};
\draw (1,2) node[circle,draw,fill=white] (110) {};
\draw (0,3) node[circle,draw,fill=white] (111) {};
\draw[latex'-] (000) -- (001);
\draw[latex'-] (000) -- (010);
\draw[latex'-] (000) -- (100);
\draw[latex'-] (001) -- (011);
\draw[latex'-] (001) -- (101);
\draw[latex'-] (010) -- (011);
\draw[latex'-] (010) -- (110);
\draw[latex'-] (100) -- (101);
\draw[latex'-] (100) -- (110);
\draw[latex'-] (011) -- (111);
\draw[latex'-] (101) -- (111);
\draw[latex'-] (110) -- (111);
\end{tikzpicture}
\caption{Two different notions of \emph{improvements} for a non-total order.}\label{fig-improve}
\end{wrapfigure}
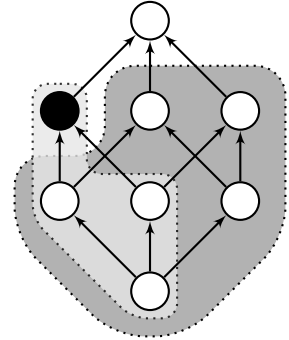
So, Nash equilibria are strategy profiles where no single player has an
incentive to unilaterally deviate from her strategy.

\begin{remark}
  Our definition of a Nash equilibrium requires any deviation to be
  worse or equivalent to the equilibrium. Another possible definition
  would have been to ask any deviation to be no better than the
  equilibrium. Those two definitions yield different notions of Nash
  equilibria (unless the preorders are total), as illustrated in
  Figure~\ref{fig-improve}:
  the black node~$n$ represents $\Out(\stat,\sigma_\Agt)$, the
  light-gray area contains the nodes~$n'$ such that $n'\prefrel n$,
  while the dark-gray area contains the nodes~$n'$ for which
  $n\not\prefrel n'$.

  This alternative definition would also be meaningful, and the
  techniques we develop in this paper could be adapted to handle such
  a variant.
\end{remark}

In this paper we will give a general construction that relates Nash
equilibria in a game (which can be infinite) and winning strategies in
a two-player turn-based game (called the \emph{suspect game}), it is
presented in Section~\ref{sec:suspect}.  We will then be mostly
interested in solving the decision problems that we define next, when
games are finite.

\subsection{Decision problems we will consider}\label{ssec-probs}

Given a~concurrent game
$\calG=\tuple{\Stat,\Agt,\Act,\Allow,\Tab,(\mathord\prefrel_A)_{A\in\Agt}}$
and a state~$\stat\in\Stat$, we consider the following problems:
\begin{itemize}
\item \emph{Value problem}: Given a player~$A$ and a play~$\pi$,
  is~there a strategy~$\sigma_A$ for player~$A$ such that for any
  outcome~$\rho$ in~$\calG$ from~$\stat$ of~$\sigma_A$, it~holds $\pi
  \prefrel_A \rho$?
\item \emph{NE Existence problem}: Does there exist a Nash equilibrium
  in~$\calG$ from~$\stat$?
\item \emph{Constrained NE existence problem}: Given two plays $\pi_A^-$
  and $\pi_A^+$ for each player~$A$, does there exist a Nash
  equilibrium in~$\calG$ from~$\stat$ whose outcome~$\pi$ satisfies
  $\pi_A^- \prefrel_A \pi \prefrel_A \pi_A^+$ for all~$A\in \Agt$?
\end{itemize}
We will focus on decidability and complexity results of these three
problems when games are finite, for various classes of preference
relations. Complexity results will heavily rely on what preorders we
allow for the preference relation and how they are represented. We
have already discussed the representation of the game structure in
Remark~\ref{remark:encoding}. We define and discuss now the various
preference relations we will study, and explain how we encode the
various inputs to the problems.

\subsection{Focus on the preference relations we will consider}
\label{sec:prefrel}

We define the various classes of preference relations we will focus on
in the rest of the paper. We begin with single-objective preference
relations, and we then define a more general class of ordered
objectives. We fix a game $\calG =
\tuple{\Stat,\Agt,\Act,\Allow,\Tab,(\mathord\prefrel_A)_{A\in\Agt} }$.

\subsubsection{Single-objective preference relations}

\begin{definition}
  An~\newdef{objective} (or \newdef{winning condition}) is an
  arbitrary set of plays. A preference relation $\prefrel_A$ is
  \newdef{single-objective} whenever there exists an objective
  $\Omega_A$ \st: $\rho \prefrel_A \rho'$ if, and only~if, $\rho' \in
  \Omega_A$ (we~then say that $\rho'$ is winning for~$A$) or
  $\rho\not\in\Omega_A$ (we~then say that $\rho$ is losing for~$A$).
\end{definition}

The setting of single-objective preference relations is purely
qualitative, since a player can only win (in case the outcome is in
her objective), or lose (otherwise). 

\medskip An~objective~$\Omega$ can be specified in various ways. Next
we will consider the following families of $\omega$-regular
objectives:
\begin{itemize}
\item A \textbf{reachability objective} is given by a target set $T
  \subseteq \Stat$ and the corresponding set of winning plays is
  defined by
  \[\Omega^{\text{Reach}}_T = \{\rho \in \Play
  \mid \Occ(\rho) \cap T \ne \varnothing \}. \]
\item A \textbf{safety objective} is given by a target set $T
  \subseteq \Stat$ and the corresponding set of winning plays is
  defined by
  \[\Omega^{\text{Safety}}_T = \{\rho \in \Play \mid
  \Occ(\rho) \cap T = \varnothing \}.\]
\item A \textbf{B\"uchi} objective is given by a target set $T
  \subseteq \Stat$ and the corresponding set of winning plays is
  defined by
  \[\Omega^{\text{B\"uchi}}_T = \{ \rho \in \Play
  \mid \Inf(\rho) \cap T \ne \varnothing \}.\]
\item A \textbf{co-B\"uchi objective} is given by a target set $T
  \subseteq \Stat$ and the corresponding set of winning plays is
  defined by
  \[\Omega^{\text{co-B\"uchi}}_T =
  \{ \rho \in \Play \mid \Inf(\rho) \cap T = \varnothing \}.\]
\item A \textbf{parity objective} is given by a priority function $p
  \colon \Stat \mapsto \lsem 0 , d \rsem$ (where $\lsem 0,d\rsem = [0,d]\cap
  \mathbb Z$) with $d\in \N$, and the
  corresponding set of winning plays is defined by
  \[\Omega^{\text{Parity}}_{p} = \{ \rho \in \Play \mid 
  \min(\Inf(p(\rho))) \ \text{is even} \}.\]
\item A \textbf{Streett objective} is given by a tuple
  $(Q_i,R_i)_{i\in\lsem 1,k\rsem}$ and the corresponding set of
  winning plays is defined by
  \[\Omega^{\text{Streett}}_{(Q_i,R_i)_{i\in\lsem
      1,k\rsem}} = \{ \rho \in \Play \mid \forall i.\ \Inf(\rho) \cap
  Q_i \ne \varnothing \Rightarrow \Inf(\rho) \cap R_i \ne \varnothing
  \}.\]
\item A \textbf{Rabin objective} is given by a tuple
  $(Q_i,R_i)_{i\in\lsem 1,k\rsem}$ and the corresponding set of
  winning plays is defined by
  \[\Omega^{\text{Rabin}}_{(Q_i,R_i)_{i\in\lsem
      1,k\rsem}} = \{ \rho \in \Play \mid \exists i.\ \Inf(\rho) \cap
  Q_i \ne \varnothing \land \Inf(\rho) \cap R_i = \varnothing \}.\]
\item A \textbf{Muller objective} is given by a finite set $C$, a
  coloring function $c \colon \Stat \mapsto C$, and a set $\mathcal{F}
  \subseteq 2^C$. The corresponding set of winning plays is then
  defined by
  \[\Omega^{\text{Muller}}_{c,\mathcal{F}} = \{ \rho \in \Play
  \mid \Inf(c(\rho)) \in\mathcal{F} \}.\]
\end{itemize}

\medskip\noindent We will also consider the following other types of objectives:
\begin{itemize}
\item A \textbf{circuit objective} is given by a boolean circuit~$C$ with the
  set $\Stat$ as input nodes and one output node. A~play~$\rho$ is winning if
  and only if $C$~evaluates to true when the input nodes corresponding to
  states in $\Inf(\rho)$ are set to \true, and all other input nodes are set to
  \false. We~write $\Omega^{\text{Circuit}}_{C}$ for the set of winning plays.
  
  Figure~\ref{fig:ex-circuit} displays an example of a circuit
  for the game of Figure~\ref{fig-ex}: this Boolean
  circuit defines the condition that either $\ell_3$ appears
  infinitely often, or if $\ell_1$ appears infinitely often then
  so does~$\ell_2$.
\begin{figure}[t]
  \centering
  \begin{tikzpicture}[thick]
    \draw (0,4) node [draw,circle] (S0) {$\ell_0$};
    \draw (1.5,4) node [draw,circle] (S1) {$\ell_1$};
    \draw (3,4) node [draw,circle] (S2) {$\ell_2$};
    \draw (4.5,4) node [draw,circle] (S3) {$\ell_3$};
    \draw (1.5,3) node [draw,circle] (Not) {$\lnot$};
    \draw (2.5,2.5) node [draw,circle] (Or) {$\lor$};
    \draw (3.5,2) node [draw,circle] (Or2) {$\lor$};
    \draw (3.5,1) node  (Out) {~~};
    \draw[-latex'] (S1) -- (Not);
    \draw[-latex'] (Not) -- (Or);
    \draw[-latex'] (S2) -- (Or);
    \draw[-latex'] (Or) -- (Or2);
    \draw[-latex'] (S3) -- (Or2);
    \draw[-latex'] (Or2) -- (Out);
  \end{tikzpicture}
  \caption{Boolean circuit defining
    the condition that either $\ell_3$ appears infinitely often, or if
    $\ell_1$ appears infinitely often then so does~$\ell_2$.}
  \label{fig:ex-circuit}
\end{figure}
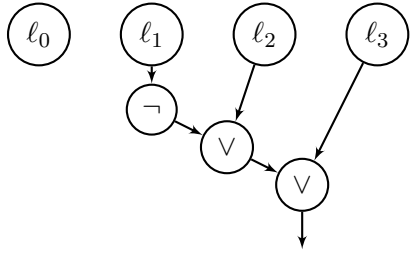
\item A \textbf{deterministic B\"uchi automaton objective} is given by
  a deterministic B\"uchi automaton~$\calA =
  \tuple{Q,\Sigma,\delta,q_0,R}$, with $\Sigma=\Stat$.  Then the
  corresponding set of winning plays is defined by
  \[\Omega^{\text{det-B\"uchi-aut}}_{\calA} = \Lang(\calA).\]
\item A \textbf{deterministic Rabin automaton objective} is given by a
  deterministic Rabin automaton~$\calA = \tuple{Q,\Sigma,\delta,q_0,
    (E_i,F_i)_{i\in \lsem 1 , k\rsem}}$, with $\Sigma=\Stat$. Then the
  corresponding set of winning plays is defined by
  \[\Omega^{\text{det-Rabin-aut}}_{\calA} = \Lang(\calA).\]
\item A \textbf{Presburger-definable objective} is given by a
  Presburger formula~$\phi$ with free variables $(X_s)_{s \in
    \Stat}$. The corresponding set of winning plays is defined
  by \[\Omega^{\text{Presb}}_\phi = \{\rho \in \Play \mid
  \phi(\#s(\rho))_{s \in \Stat})=0\}\] where $\#s(\rho)$ is the number
  of occurrences\footnote{By convention, if $s \in \Inf(\rho)$, and
    variable $X_s$ appears in~$\phi$, then $\rho \notin
    \Omega^{\text{Presb}}_\phi$.}  of state~$s$ along~$\rho$.
\end{itemize}

\paragraph{Encodings} For complexity issues we now make explicit how
the various objectives are encoded:
\begin{itemize}
\item Reachability, safety, B\"uchi and co-B\"uchi objectives are
  given by a set $T\subseteq \Stat$, they can therefore be encoded
  using $|\Stat|$ bits.
\item For parity objectives, we assume without loss of generality that
  $d \le 2 \cdot |\Stat|$.  The priority function has then size at
  most $|\Stat| \cdot \lceil\log(2\cdot |\Stat| + 1)\rceil$.
\item Street and Rabin objectives are given by tuples $(Q_i,R_i)_{i\in
    \lsem 1,k\rsem}$.  Their sizes are given by: $\sum_{i\in \lsem
    1,k\rsem} |Q_i| \lceil\log(|\Stat|)\rceil$.
\item Muller objectives are given by a coloring function and a set
  $\mathcal{F}$.  Its size is $|\Stat|\cdot \lceil\log(|C|)\rceil +
  |\mathcal{F}| \cdot \lceil\log(|C|)\rceil$.  Note that thanks to the
  coloring function, this encoding can be exponentially more succinct
  than an explicit representation such as the one considered
  in~\cite{horn2008explicit}.
\item The size of objectives given by circuits, deterministic automata
  or Presburger formulas is that of the corresponding circuits,
  deterministic automata or Presburger formulas.
\end{itemize}

\paragraph{Encodings of thresholds in inputs of the value and the
  constrained NE existence problems.}
For all the objectives except for those given by automata, whether a
play $\rho$ satisfies the objective or not only depends on the sets
$\Occ(\rho)$ and $\Inf(\rho)$. The various thresholds will therefore
be encoded as such pairs $(\Occ,\Inf)$.

For deterministic-automata objectives, the thresholds will be also
encoded as pairs of sets of states of the objectives, representing
respectively the set of states which are visited and the set of states
which are visited infinitely often.

For the Boolean circuit objectives, whether a play $\rho$ satisfies
the objective or not only depends on the set $\Inf(\rho)$. Therefore
we will use as encoding for the threshold a single set $\Inf$.

For the Presburger formulas objectives, we will use as encoding for
the thresholds the Parikh image of the play (\ie, the number of visits
to each of the states).

\subsubsection{Ordered objectives}
We now turn to a more general class of preference relations, allowing
for a \textit{semi-quantitative setting}.
\begin{definition}
  An \newdef{ordered objective} is a pair $\omega = \langle
  (\Omega_i)_{1 \le i \le n},\preorder\rangle$, where, for every $1
  \le i \le n$, $\Omega_i$~is an objective, and $\preorder$~is a
  preorder on $\{0,1\}^n$. A~play~$\rho$ is assigned a \newdef{payoff
    vector} w.r.t. that ordered objective, which is defined
  as~$\payoff_\omega(\rho) = \One_{\{i \mid
    \rho\in\Omega_i\}}\in\{0,1\}^{n}$ (where $\One_S$ is the vector
  $v$ such that $v_i = 1 \Leftrightarrow i \in S$). The corresponding
  preference relation $\prefrel_\omega$ is then defined by $\rho
  \prefrel_\omega \rho'$ \iff $\payoff_\omega(\rho) \preorder
  \payoff_\omega(\rho')$.
\end{definition}

There are many ways of specifying a preorder. We define below the
preorders on~$\{0,1\}^n$ that we consider in the
sequel. Figure~\ref{fig-preorder} displays four such preorders for
$n=3$. For the purpose of these definitions, we assume that
$\max\varnothing=-\infty$.

\begin{figure}[t]
\bgroup
\makeatletter
\def\@captype{subfigure}
\makeatother
\centering
  \begin{minipage}{.3\textwidth}
  \centering
  \begin{tikzpicture}
    \tikzset{every node/.style={font=\scriptsize}}
      \draw(0,0) node (A) {$(0,0,0)$};
      \draw (-1,1) node (B) {$(1,0,0)$};
      \draw (0,1) node (C) {$(0,1,0)$};
      \draw (1,1) node (D) {$(0,0,1)$};
      \draw (-1,2) node (E) {$(1,1,0)$};
      \draw (0,2) node (F) {$(1,0,1)$};
      \draw (1,2) node (G) {$(0,1,1)$};
      \draw (0,3) node (H) {$(1,1,1)$};
      \draw[-latex'] (A) -- (B);
      \draw[-latex'] (A) -- (C);
      \draw[-latex'] (A) -- (D);
      \draw[-latex'] (B) -- (E);
      \draw[-latex'] (B) -- (F);
      \draw[-latex'] (C) -- (E);
      \draw[-latex'] (C) -- (G);
      \draw[-latex'] (D) -- (F);
      \draw[-latex'] (D) -- (G);
      \draw[-latex'] (E) -- (H);
      \draw[-latex'] (F) -- (H);
      \draw[-latex'] (G) -- (H);
    \end{tikzpicture}
  \caption{Subset preorder}\label{fig-first}
  \end{minipage}%
  \begin{minipage}{.35\textwidth}
  \centering
    \begin{tikzpicture}[x=0.95cm]
      \tikzset{every node/.style={font=\scriptsize}}
      \draw(0,0) node (A) {$(0,0,0)$};
      \draw (0,1) node (B) {$(1,0,0)$};
      \draw (-1,2) node (C1) {$(0,1,0)$};
      \draw (1,2) node (C2) {$(1,1,0)$};
      \draw (-2,3) node (D1) {$(0,0,1)$};
      \draw (-0.7,3) node (D2) {$(1,0,1)$};
      \draw (0.7,3) node (D3) {$(0,1,1)$};
      \draw (2,3) node (D4) {$(1,1,1)$};
      \node[draw,densely dotted,rounded corners=2mm,fit=(A),inner sep=0mm] {};
      \node[draw,densely dotted,rounded corners=2mm,fit=(B),inner sep=0mm] {};
      \node[draw,densely dotted,rounded corners=2mm,fit=(C1)(C2),inner sep=0mm] {};
      \node[draw,densely dotted,rounded corners=2mm,fit=(D1)(D2)(D3)(D4),inner sep=0mm] {};
      \draw[-latex'] (0,.3) -- (0,.7);
      \draw[-latex'] (0,1.3) -- (0,1.7);
      \draw[-latex'] (0,2.3) -- (0,2.7);
    \end{tikzpicture}
    \caption{\mbox{Maximise preorder}}
  \end{minipage}%
  \begin{minipage}{.35\textwidth}
  \centering
    \begin{tikzpicture}
      \tikzset{every node/.style={font=\scriptsize}}
      \draw(0,0) node (A) {$(0,0,0)$};
      \draw (-1,1) node (B) {$(1,0,0)$};
      \draw (0,1) node (C) {$(0,1,0)$};
      \draw (1,1) node (D) {$(0,0,1)$};
      \draw (-1,2) node (E) {$(1,1,0)$};
      \draw (0,2) node (F) {$(1,0,1)$};
      \draw (1,2) node (G) {$(0,1,1)$};
      \draw (0,3) node (H) {$(1,1,1)$};
      \node[draw,densely dotted,rounded corners=2mm,fit=(A),inner sep=0mm] {};
      \node[draw,densely dotted,rounded corners=2mm,fit=(B)(C)(D),inner sep=0mm] {};
      \node[draw,densely dotted,rounded corners=2mm,fit=(E)(F)(G),inner sep=0mm] {};
      \node[draw,densely dotted,rounded corners=2mm,fit=(H),inner sep=0mm] {};
      \draw[-latex'] (0,.3) -- (0,.7);
      \draw[-latex'] (0,1.3) -- (0,1.7);
      \draw[-latex'] (0,2.3) -- (0,2.7);
    \end{tikzpicture}
  \caption{Counting preorder}
  \end{minipage}

\bigskip
  \begin{minipage}{.9\linewidth}
  \centering
  \begin{tikzpicture}[scale=1.1]
      \tikzset{every node/.style={font=\scriptsize}}
      \draw(0,0) node (A) {$(0,0,0)$};
      \draw (1.4,0) node (B) {$(0,0,1)$};
      \draw (2.8,0) node (C) {$(0,1,0)$};
      \draw (4.2,0) node (D) {$(0,1,1)$};
      \draw (5.6,0) node (E) {$(1,0,0)$};
      \draw (7,0) node (F) {$(1,0,1)$};
      \draw (8.4,0) node (G) {$(1,1,0)$};
      \draw (9.8,0) node (H) {$(1,1,1)$};
      \draw[-latex'] (A) -- (B);
      \draw[-latex'] (B) -- (C);
      \draw[-latex'] (C) -- (D);
      \draw[-latex'] (D) -- (E);
      \draw[-latex'] (E) -- (F);
      \draw[-latex'] (F) -- (G);
      \draw[-latex'] (G) -- (H);
  \end{tikzpicture}
  \caption{Lexicographic order}\label{fig-last}
\end{minipage}
\egroup
\caption{Examples of preorders (for $n=3$): dotted boxes represent
  equivalence classes for the relation~$\sim$, defined as $a\sim b
  \Leftrightarrow a\preorder b \land b\preorder a$; arrows represent
  the preorder relation~$\preorder$ quotiented by~$\sim$.}\label{fig-preorder}
\end{figure}
\begin{enumerate}
\item \newdef{Conjunction}: $v \preorder w$ \iff either $v_i=0$ for
  some~$1\leq i\leq n$, or $w_i=1$ for all~$1\leq i\leq n$.  This
  corresponds to the case where a player wants to achieve all her
  objectives.
\item \newdef{Disjunction}: $v \preorder w$ \iff either $v_i=0$ for
  all~$1\leq i\leq n$, or $w_i=1$ for some~$1\leq i\leq n$. The aim
  here is to satisfy at least one objective.
\item \newdef{Counting}: $v \preorder w$ \iff $|\{i\mid v_i=1\}| \leq
  |\{i\mid w_i=1\}|$.  The aim is to maximise the number of conditions
  that are satisfied;
\item \newdef{Subset}: $v \preorder w$ \iff $\{i\mid v_i=1\} \subseteq
  \{i\mid w_i = 1 \}$: in this setting, a player will always struggle
  to satisfy a larger (for inclusion) set of objectives.
\item \newdef{Maximise}: $v \preorder w$ \iff $\max \{i\mid v_i=1\}
  \leq \max \{i\mid w_i = 1 \}$.  The aim is to maximise the highest
  index of the objectives that are satisfied.
\item \newdef{Lexicographic}: $v \preorder w$ \iff either $v=w$, or
  there is~$1 \le i \le n$ such that $v_i=0$, $w_i=1$ and $v_j=w_j$
  for all $1\leq j<i$.
\item \newdef{Boolean Circuit}: given a Boolean circuit, with input
  from $\{0,1\}^{2 n}$, $v \preorder w$ \iff the circuit evaluates~$1$
  on input $v_1 \ldots v_n w_1 \ldots w_n$.
\item \newdef{Monotonic Boolean Circuit}: same as above, with the
  restriction that the input gates corresponding to~$v$ are negated,
  and no other negation appear in the circuit.
\end{enumerate}

\noindent In terms of expressiveness, any preorder over $\{0,1\}^n$ can be given
as a Boolean circuit: for each pair~$(v,w)$ with $v \preorder w$, it
is possible to construct a circuit whose output is~$1$ \iff the input
is $v_1 \ldots v_n w_1 \ldots w_n$; taking the disjunction of all
these circuits we obtain a Boolean circuit defining the preorder.  Its
size can be bounded by $2^{2 n+3} n$, which is exponential in
general. But all the above examples ((1)-(6)) can be specified with a
circuit of polynomial size. In Figure~\ref{fig:boolean-subset} we give
a polynomial-size Boolean circuit for the subset preorder. In the
following, for complexity issues, we will assume that the encoding of
all preorders (1)-(6) takes constant size, and that the size of the
preorder when it is given as a Boolean circuit is precisely the size
of the circuit for input size $n$, where $n$ is the number of
objectives.

A preorder~$\preorder$ is \newdef{monotonic} if it is compatible with
the subset ordering, \ie if $\{i\mid v_i=1\} \subseteq \{i\mid w_i = 1
\}$ implies $v \preorder w$. Hence, a preorder is monotonic if
fulfilling more objectives never results in a lower payoff.  All our
examples of preorders except for the Boolean circuit preorder are
monotonic.  Moreover, any monotonic preorder can be expressed as a
monotonic Boolean circuit: for a pair~$(v,w)$ with~$v\preorder w$, we
can build a circuit whose output is~$1$ \iff the input is~$v_1 \ldots
v_n w_1 \ldots w_n$.  We~can require this circuit to have negation at
the leaves.  Indeed, if the input~$w_j$ appears negated, and if
$w_j=0$, then by monotonicity, also the input $(v,\tilde w)$ is
accepted, with $\tilde w_i=w_i$ when $i\not=j$ and $\tilde w_j=1$.
Hence the negated input gate can be replaced with~$\texttt{true}$.
Similarly for positive occurrences of any~$v_j$.  Hence any monotonic
preorder can be written as a monotonic Boolean circuit.  Notice that
with Definition~\ref{def-NE}, any Nash equilibrium $\sigma_\Agt$ for
the subset preorder is also a Nash equilibrium for any monotonic
preorder.

\begin{figure}[htb]
  \centering{
  \begin{tikzpicture}[thick]
    \everymath{\scriptstyle}
    \draw[black!10!white,line width=4.5mm] (-.4,0) -- (3cm+3pt,0);
    \draw[black!10!white,line width=4.5mm] (4cm+3pt,0) -- +(3.4cm+3pt,0);
    \draw(0,0) node[draw,minimum width=8mm, minimum height=4.5mm,inner sep=0pt] (A) {$v_1$};
    \draw(A.0) node[draw,right,minimum width=8mm, minimum height=4.5mm,inner sep=0pt] (B) {$v_2$};
    \draw(B.0) node[right,minimum width=10mm, minimum height=4.5mm,inner sep=0pt]  (C) {$\dots$};
    \draw(C.0) node[draw,right,minimum width=8mm, minimum height=4.5mm,inner sep=0pt]  (D) {$v_n$};
    \draw(D.0) +(1,0) node[draw,right,minimum width=8mm, minimum height=4.5mm,inner sep=0pt]  (E) {$w_1$};
    \draw(E.0) node[draw,right,minimum width=8mm, minimum height=4.5mm,inner sep=0pt]  (F) {$w_2$};
    \draw(F.0) node[right,minimum width=10mm, minimum height=4.5mm,inner sep=0pt]  (G) {$\dots$};
    \draw(G.0) node[draw,right,minimum width=8mm, minimum height=4.5mm,inner sep=0pt]  (H) {$w_n$};
  
    \draw (A) + (0,-0.7) node[draw,rounded corners=2mm,inner sep=1mm] (G1) {$\textsf{NOT}$};
    \draw (B) + (0,-0.7) node[draw,rounded corners=2mm,inner sep=1mm] (G2)
          {$\textsf{NOT}$};
    \draw (C) + (0,-0.7) node  {$\dots$};
    \draw (D) + (0,-0.7) node[draw,rounded corners=2mm,inner sep=1mm] (G3) {$\textsf{NOT}$};

    \draw (E) + (0,-2.4) node[draw,rounded corners=2mm,inner sep=1mm] (G4) {$\textsf{OR}$};
    \draw (F) + (0,-1.9) node[draw,rounded corners=2mm,inner sep=1mm] (G5) {$\textsf{OR}$};
    \draw (G) + (0,-1.5) node  {$\dots$};
    \draw (H) + (0,-1.2) node[draw,rounded corners=2mm,inner sep=1mm] (G6) {$\textsf{OR}$};

    \draw (E) + (0,-3) node[draw,rounded corners=2mm,inner sep=1mm] (G7) {$\textsf{AND}$};

    \draw (A) -- (G1);
    \draw (B) -- (G2);
    \draw (D) -- (G3);

    \draw (G1) |- (G4);
    \draw (G2) |- (G5);
    \draw (G3) |- (G6);

    \draw (E) -- (G4);
    \draw (F) -- (G5);
    \draw (H) -- (G6);

    \draw (G4) -- (G7);
    \draw (G5) |- (G7.20);
    \draw (G6) |- (G7);

    \draw (G7) -- +(0,-0.6);
  \end{tikzpicture}
  }
  \caption{Boolean circuit defining the subset preorder}
  \label{fig:boolean-subset}
\end{figure}
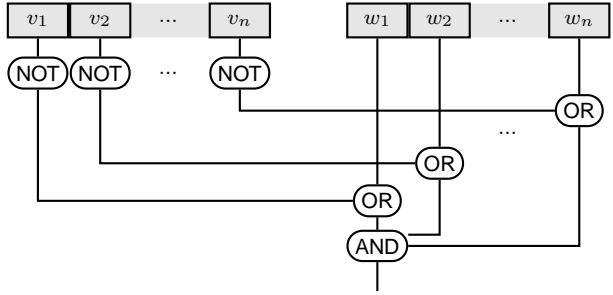

Next we will be be interested in two kinds of ordered objectives,
\emph{ordered reachability objectives}, where all objectives are supposed to
be reachability objectives, and \emph{ordered B\"uchi objectives}, where all
objectives are supposed to be B\"uchi objectives. Note that other classical
objectives (parity, Streett, Rabin, Muller, \etc) can be equivalently
described with a preorder given by a polynomial-size Boolean circuit over
B\"uchi objectives. For instance, each set of a~Muller condition can be
encoded as a conjunction of B\"uchi and co-B\"uchi conditions.

For ordered reachability (resp. B\"uchi) objectives, thresholds used
as inputs to the various decision problems will be given by the set of
states that are visited (resp. visited infinitely often).

\bigskip In Sections~\ref{sec:buchi} and~\ref{sec:reach}, we will be
interested in games where, for every player $A$, the preference
relation $\prefrel_A$ is given by an ordered objective $\omega_A =
\langle (\Omega^A_i)_{1 \le i \le n_A},\preorder_A\rangle$.  We will
then write~$\payoff_A$ instead of $\payoff_{\omega_A}$ for the
payoffs, and if $\rho$ is a play, $\payoff(\rho) =
(\payoff_A(\rho))_{A \in \Agt}$.

\subsection{Undecidability of all three problems for single Presburger-definable objectives}

We end this section with an undecidability result in the quite
general setting of Presburger-definable preference relations.

\begin{theorem}
  \label{theorem:undecidable}
  The value, NE existence and constrained NE existence problems are
  undecidable for finite games with preference relations given by
  Presburger-definable qualitative objectives.
\end{theorem}

\begin{proof}
  We first prove the result for the constrained NE existence problem, by
  encoding a two-counter machine. We fix a two-counter machine, and
  assume without loss of generality that the halting state is preceded
  by a non-zero test for the two counters (hence if the machine halts,
  the two counters have a positive value in the halting state).

  We begin with defining a family of preorders. Fix two sets of
  states~$S$ and~$T$; a~play is said $(S=T)$-winning if the
  number of visits to~$S$ equals the number of visits to~$T$, and both
  are finite. Formally, $\pi \prefrel_{S=T} \pi'$ whenever $\pi$ is
  not $(S=T)$-winning, or $\pi'$~is.

We use such preorders to encode the acceptance problem for two-counter
machines: the value of counter~$c_1$ is encoded as the difference
between the number of visits to~$S_1$ and~$T_1$, and similarly for
counter~$c_2$.  Incrementing counter~$c_i$ thus consists in visiting a
state in~$S_i$, and decrementing consists in visiting~$T_i$; in other
terms, if instruction~$q_k$ of the two-counter machine consists in
incrementing~$c_1$ and jumping to~$q_{k'}$, then the game will have a
transition from some state~$q_k$ to a state in~$S_1$, and a transition
from there to~$q_{k'}$. The game involves three players: $A_1$,
$A_2$ and~$B$. The~aim of player~$A_1$ (resp.~$A_2$) is to visit~$S_1$
and~$T_1$ (resp.~$S_2$ and~$T_2$) the same number of times: player
$A_i$'s preference is $\prefrel_{S_i=T_i}$.  The aim of player~$B$ is
to reach the state corresponding to the halting state of the
two-counter machine. Due to the assumption on the two-counter machine,
if $B$ wins, then both $A_1$ and $A_2$ lose.

\smallskip 
\begin{wrapfigure}{r}{6.3cm}
  \centering
  \begin{tikzpicture}[thick]
    \tikzstyle{square}=[draw,minimum height=5mm,minimum width=5mm, inner sep=1mm]
    \tikzstyle{smsquare}=[draw,minimum height=5mm,minimum width=5mm,fill=black!20!white,inner sep=1mm]
    \draw (0.2,0) node[draw,circle,minimum width=5mm] (B) {};
    \draw (B.-90) node[below] (BB) {$B$};    
    \draw (2,0.8) node[square] (A1) {$u_i^{{\scriptscriptstyle \ne 0}}$};
    \draw (A1.170) node[left] (A1B) {$A_i$};    
    \draw (2,-0.8) node[square] (A2) {$u_i^{{\scriptscriptstyle = 0}}$};
    \draw (A2.90) node[above] (A2A) {$A_i$};    
    \draw (2,1.8) node[smsquare] (U) {};
    \draw (1,-1.6) node[square] (S) {$s_i$};
    \draw (S.0) node[right] (A2A) {$A_i$};    
    \draw (3,-1.6) node[square] (T) {$t_i$};
    \draw (T.180) node[left] (A2A) {$A_i$};    
    \draw (1,-2.6) node[smsquare] (SB) {};
    \draw (3,-2.6) node[smsquare] (TB) {};

    \draw[-latex'] (B) -- (A1);
    \draw[-latex'] (A1) -- (U);
    \draw[-latex'] (A1) -- +(.8,0);
    \draw[-latex'] (U) .. controls +(.5,0.9) and +(-.5,0.9) .. (U);
    \draw[-latex'] (B) -- (A2);
    \draw[-latex'] (A2) -- +(.8,0);
    \draw[-latex'] (A2) -- (S);
    \draw[-latex'] (A2) -- (T);
    \draw[-latex'] (S) -- (SB);
    \draw[-latex'] (T) -- (TB);
    \draw[-latex'] (S) .. controls +(-1,-0.7) and +(-1,0.7) .. (S);
    \draw[-latex'] (T) .. controls +(1,-0.7) and +(1,0.7) .. (T);
    \draw[-latex'] (SB) .. controls +(0.5,-0.9) and +(-0.5,-0.9) .. (SB);
    \draw[-latex'] (TB) .. controls +(0.5,-0.9) and +(-0.5,-0.9) .. (TB);
    \draw[-latex'] (-0.5,0) -- (B);
  \end{tikzpicture}
  \caption{Testing whether $c_i=0$.}
  \label{fig:zero-test}
\end{wrapfigure}
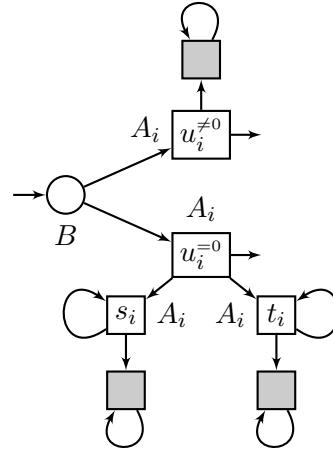
It~remains to encode the zero-test: this is achieved by the
module of Figure~\ref{fig:zero-test}. In~this module, player~$B$ 
tries to avoid the three sink states (marked in grey), since this would
prevent her from reaching her goal.
When entering the module, player~$B$ has to choose one of the
available branches: if she decides to go to~$u_i^{{\scriptscriptstyle
    \ne 0}}$, then $A_i$ could take the play into the self-loop, which
is winning for her if $S_i$ and~$T_i$ have been visited the same
number of times in the history of this path, which corresponds to
having $c_i=0$; hence player~$B$ should play
to~$u_i^{{\scriptscriptstyle \ne 0}}$ only if $c_i\not=0$, so that
$A_1$ has no interest in going to this self-loop.

Similarly, if player~$B$ decides to go to~$u_i^{{\scriptscriptstyle
    =0}}$, player~$A_i$ has the opportunity to ``leave'' the main
stream of the game, and go to~$s_i$ or~$t_i$ (obviously $s_i \in S_i$
and $t_i \in T_i$). If~the numbers of visits to~$S_i$ and~$T_i$ up to
that point are different, then player~$A_i$ has the opportunity to
make both numbers equal, and to win. Conversely, if both numbers are
equal (i.e., $c_i=0$), then going to~$s_i$ or~$t_i$ will be losing
for~$A_i$, whatever happens from there.  Hence, if~$c_i=0$ when
entering the module, then player~$B$ should go
to~$u_i^{{\scriptscriptstyle =0}}$.

\medskip One can then easily show that the two-counter machine stops
if, and only if, there is a Nash equilibrium in the resulting game~$\calG$, in
which player~$B$ wins and players~$A_1$ and~$A_2$ lose.
Indeed, assume that the machine stops, and consider the strategies
where player~$B$ plays (in the first state of the test modules)
according to the value of the corresponding counter, and where
players~$A_1$ and~$A_2$ always keep the play in the main stream of the
game. Since the machine stops, player~$B$ wins, while players~$A_1$
and~$A_2$ lose. Moreover, none of them has a way to improve their
payoff: since player~$B$ plays according to the values of the
counters, players~$A_1$ and~$A_2$ would not benefit from deviating
from their above strategies.
Conversely, if there is such a Nash equilibrium, then in any visited
test module, player~$B$ always plays according to the values of the
counters: otherwise, player~$A_1$ (or~$A_2$) would have the
opportunity to win the game.  By~construction, this means that the run
of the Nash equilibrium corresponds to the execution of the
two-counter machine. As~player~$B$ wins, this execution reaches the
halting state.

\bigskip Finally, it is not difficult to adapt this reduction to
involve only two players: players~$A_1$ and~$A_2$ would be replaced by
one single player~$A$, in charge of ensuring that both conditions
(for~$c_1$ and~$c_2$) are fulfilled.  This requires minor changes to
the module for testing~$c_i=0$: when leaving the main stream of the
game in a module for testing counter~$c_i$, player~$A$ should be given
the opportunity (after the grey state) to visit states~$S_{3-i}$
or~$T_{3-i}$ in order to adjust that part of her objective.

\smallskip
By changing the winning condition for Player~$B$, the game~$\calG$ can also be made
zero-sum: for this, $B$~must lose if the play remains in the main stream
forever without visiting the final state; otherwise, $B$~loses if the
number of visits to~$s_i$ and~$t_i$ are finite and equal for both~$i=1$
and~$i=2$; $B$~wins in any other case. The objective of player~$A$ is
opposite. It~is not difficult to modify the proof above for showing that the
two-counter machine halts if, and only if, player~$B$ has a winning strategy
in this game.

\smallskip 
\begin{wrapfigure}{r}{6.2cm}
  \centering
  \begin{tikzpicture}[thick]
    \tikzset{noeud/.style={circle,draw=black,thick,fill=black!10,minimum
        height=6mm,inner sep=0pt}}
    \draw (0,0) node [noeud] (A) {$s_0$};
    \draw (30:1.6cm) node [noeud] (B) {$s_1$};
    \draw (-30:1.6cm) node [noeud] (C) {$s$};
    \begin{scope}[xshift=1.1cm,yshift=-.75cm]
    \draw[dashed, rounded corners=2mm] (0,0) .. controls +(90:5mm) .. (1,.5) -- (2,.9) -- (2.5,.6) --
    (3,0) -- (2.5,-.8) -- (2,-.9) -- (1,-.7) .. controls +(150:5mm) and
    +(-90:5mm) .. (0,0);
    \draw (1.8,-0.1) node {\begin{minipage}{2.1cm}\footnotesize\centering
        Copy of~$\calG$
      \end{minipage}};
    \end{scope}
    \draw[-latex'] (A) -- (B) node[midway,above left=-2pt] {$\scriptstyle \tuple{1,1}, \tuple{2,2}$};
    \draw[-latex'] (A) -- (C) node[midway,below left=-2pt] {$\scriptstyle \tuple{1,2}, \tuple{2,1}$};
    \draw[-latex'] (B) .. controls +(30:10mm) and +(-30:10mm) .. (B);
    \draw[dashed] (C) -- +(40:5mm);
    \draw[dashed] (C) -- +(0:4mm);
    \draw[dashed] (C) -- +(-40:5mm);
  \end{tikzpicture}
  \caption{Extending the game with an initial concurrent module}
  \label{fig-init-module}
\end{wrapfigure}
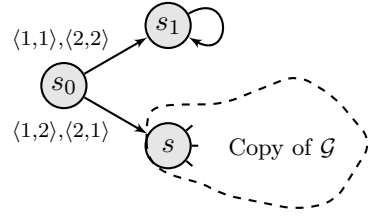
Finally, by adding a small initial module depicted on
Figure~\ref{fig-init-module} to this zero-sum version of the game~$\calG$, one
can encode the halting problem for two-counter machines to the
NE existence problem.  Indeed, in the zero-sum game, there is exactly one
Nash equilibrium, with only two possible payoffs (either~$A$ wins,
or~$B$~wins). 
Now, assuming that $A$ loses and $B$ wins in state~$s_1$, 
then there is a (pure) Nash equilibrium in
the game extended with the initial module if, and only if, player~$B$
wins in the zero-sum game above.
\end{proof}

\section{Preliminary results}\label{sec-prelim}

This section contains general results that will be applied later in
various settings. In each of the statements, we give the restrictions
on the games and on the preference relations that should be satisfied.

\subsection{Nash equilibria as lasso runs}
\label{sec:lasso}
We first characterise outcomes of Nash equilibria as ultimately
periodic runs, in the case where preference relations only depend on the set of
states that are visited, and on the set of states that are visited
infinitely often. 
Note that $\omega$-regular conditions satisfy this hypothesis, 
but Presburger relations such as the ones used for proving
Theorem~\ref{theorem:undecidable} do~not.

\begin{proposition}\label{lem:play-length}
  Let~$\calG=\tuple{\Stat,\Agt,\Act,\Allow,\Tab,(\mathord\prefrel_A)_{A\in\Agt} }$ be
  a \textbf{finite} concurrent game such that, for every player $A$,
  it~holds\footnote{We recall that $\rho \sim_A \rho'$ if, and only if, $\rho
    \prefrel_A \rho'$ and $\rho' \prefrel_A \rho$.} $\rho \sim_A \rho'$ as
  soon as $\Inf(\rho) = \Inf(\rho')$ and $\Occ(\rho) = \Occ(\rho')$. Let~$\rho
  \in \Play$. If there is a Nash equilibrium with outcome~$\rho$, then there
  is a Nash equilibrium with outcome~$\rho'$ of the form $\pi \cdot
  \tau^\omega$ such that $\rho \sim_A \rho'$, and where $|\pi|$ and~$|\tau|$
  are bounded by~$|\Stat|^2$.
\end{proposition}

\begin{proof}
  Let $\sigma_\Agt$ be a Nash equilibrium from some state~$\stat$, and
  $\rho$~be its outcome. We~define a new strategy
  profile~$\sigma'_{\Agt}$, whose outcome from~$\stat$ is ultimately
  periodic, and then show that $\sigma'_{\Agt}$ is a Nash equilibrium
  from~$\stat$.

  To begin with, we inductively construct a history $\pi = \pi_0 \pi_1
  \dots \pi_n$ that is not too long and visits precisely those states
  that are visited by~$\rho$ (that is, $\Occ(\pi) = \Occ(\rho)$).

  The initial state is $\pi_0 = \rho_0=\stat$. Then we assume we have
  constructed $\pi_{\le k} = \pi_0 \dots \pi_k$ which visits exactly
  the same states as $\rho_{\le k'}$ for some~$k'$. If all the states
  of~$\rho$ have been visited in $\pi_{\le k}$ then the construction
  is over. Otherwise there is an index~$i$ such that $\rho_{i}$ does
  not appear in~$\pi_{\le k}$. We~therefore define our next target as
  the smallest such~$i$: we~let $t(\pi_{\le k}) = \min \{ i \mid
  \forall j \le k.\ \pi_j \neq \rho_i\}$.  We~then look at the
  occurrence of the current state~$\pi_k$ that is the closest to the
  target in~$\rho$: we~let $c(\pi_{\le k}) = \max \{ j < t(\pi_{\le
    k}) \mid \pi_k = \rho_j\}$.  Then we~emulate what happens at that
  position by choosing $\pi_{j+1} = \rho_{c(\pi_{\le j})+1}$.  Then
  $\pi_{k+1}$ is either the target, or a state that has already been
  seen before in $\pi_{\le k}$, in which case the resulting $\pi_{\le
    k+1}$ visits exactly the same states as $\rho_{\le c(\pi_{\le
      k})+1}$.

  At each step, either the number of remaining targets strictly
  decreases, or the number of remaining targets is constant but the
  distance to the next target strictly decreases. Therefore the
  construction terminates.  Moreover, notice that between two targets
  we do not visit the same state twice, and we visit only states that
  have already been visited, plus the target.  As~the number of
  targets is bounded by~$|\Stat|$, we~get that the length of the path
  $\pi$ constructed thus far is bounded by $1+|\Stat|\cdot
  (|\Stat|-1)/2$.

  \smallskip Using similar ideas, we now inductively construct $\tau =
  \tau_0 \tau_1 \dots \tau_m$, which visits precisely those states
  which are seen infinitely often along~$\rho$, and which is not too
  long.  Let $l$ be the least index after which the states visited
  by~$\rho$ are visited infinitely often, \ie $l = \min\{ {i\in\N \mid
    \forall j\geq i.\ \rho_j \in \Inf(\rho)}\}$.  The~run~$\rho_{\ge
    l}$ is such that its set of visited states and its set of states
  visited infinitely often coincide. We~therefore define~$\tau$ in the
  same way we have defined~$\pi$ above, but for play~$\rho_{\ge
    l}$. As~a by-product, we also get $c(\tau_{\leq k})$, for~$k<m$.

  We now need to glue $\pi$ and~$\tau$ together, and to ensure
  that~$\tau$ can be glued to itself, so that $\pi \cdot \tau^\omega$
  is a real run.  We~therefore need to link the last state of~$\pi$
  with the first state of~$\tau$ (and similarly the last state of
  $\tau$ with its first state). This possibly requires appending some
  more states to~$\pi$ and~$\tau$: we~fix the target of~$\pi$ and
  $\tau$ to be~$\tau_0$, and apply the same construction as previously
  until the target is reached.  The total length of the resulting
  paths~$\pi'$ and~$\tau'$ is bounded by $1+(|\Stat| -
  1)\cdot(|\Stat|+2)/2$ which is less than~$|\Stat|^2$.

  \smallskip We~let~$\rho'=\pi'\cdot{\tau'}^\omega$, and abusively
  write $c(\rho'_{\leq k})$ for $c(\pi'_{\leq k})$ if~$k\leq
  \length{\pi'}$ and $c(\tau'_{\leq k'})$ with $k'=(k-1-\length{\pi'})
  \mod \length{\tau'}$ otherwise.  We now define our new strategy
  profile, having~$\rho'$ as outcome from~$\stat$.  Given a
  history~$h$:
  \begin{itemize}
  \item if $h$ followed the expected path, \ie, $h = \rho'_{\le k}$
    for some~$k$, we~mimic the strategy at~$c(h)$: $\sigma'_\Agt(h) =
    \sigma_\Agt(\rho'_{c(h)})$.
    This way, $\rho'$ is the outcome of~$\sigma'_\Agt$ from~$\stat$.
  \item otherwise we take the longest prefix $h_{\le k}$ that is a prefix
    of~$\rho'$,  and define $\sigma'_\Agt(h) =
    \sigma_\Agt(\rho'_{c(h_{\le k})} \cdot h_{\ge k+1})$.
  \end{itemize}

  \noindent We now show that $\sigma'_\Agt$ is a Nash equilibrium. Assume that
  one of the players changes her strategy while playing according
  to~$\sigma'_\Agt$: either the resulting outcome does not deviate
  from~$\pi \cdot \tau^\omega$, in which case the payoff of that
  player is not improved; or~it~deviates at some point, and from that
  point on, $\sigma'_\Agt$~follows the same strategies as
  in~$\sigma_\Agt$.  Assume that the resulting outcome is an
  improvement over~$\rho'$ for the player who deviated. The suffix of
  the play after the deviation is the suffix of a play
  of~$\sigma_\Agt$ after a deviation by the same player. By
  construction, both plays have the same sets of visited and
  infinitely-visited states. Hence we have found an advantageous
  deviation from~$\sigma_\Agt$ for one player, contradicting the fact
  that $\sigma_\Agt$ is a Nash equilibrium.  
\end{proof}

\subsection{Encoding the value problem as a constrained NE existence problem}
We now give a reduction that will be used to infer hardness results for the
constrained NE existence problem from the hardness of the value problem (as
defined in Section~\ref{ssec-probs}): this will be the case when the hardness
proof for the value problem involves the construction of a game satisfying the
hypotheses of the proposition.

\begin{proposition}
  \label{lem:link-value-constr}
  Let~$\calG=\tuple{\Stat,\Agt,\Act,\Allow,\Tab,(\mathord\prefrel_A)_{A\in\Agt}
  }$ be a two-player zero-sum game played between players $A$ and $B$,
  such that:
  \begin{itemize}
  \item the preference relation $\prefrel_{A}$ for player~$A$ is total,
    Noetherian and almost-well-founded (see~Section~\ref{ssec-gendef});
  \item $\calG$ is determined, \ie, for all play~$\pi$:
    \[ 
    [\exists \sigma_A.\ \forall \sigma_B.\ \pi \prefrel_A
    \Out(\sigma_A,\sigma_B)] \quad \Leftrightarrow \quad [\forall
    \sigma_B.\ \exists \sigma_A.\ \pi \prefrel_A
    \Out(\sigma_A,\sigma_B)].
    \]
  \end{itemize}
  Let $\calG'$ be the (non-zero-sum) game obtained from~$\calG$ by
  replacing the preference relation of player~$B$ by the one where all
  plays are equivalent. Then, for every state $s$, for every play
  $\pi$ from $s$, the two following properties are equivalent:
  \begin{enumerate}[label=(\roman*)]
  \item there is a Nash equilibrium in $\calG'$ from~$s$ with
    outcome~$\rho$ such that $\pi \not\prefrel_A \rho$;
  \item player $A$ cannot ensure $\pi$ from~$s$ in $\calG$.
  \end{enumerate}
\end{proposition}
\begin{proof}
  In this proof, $\sigma_A$ and $\sigma'_A$ (resp. $\sigma_B$ and
  $\sigma'_B$) refer to player-$A$ (resp. player-$B$) strategies.
  Furthermore we will write $\Out(\sigma_A,\sigma_B)$ instead of
  $\Out_{\calG}(s,(\sigma_A,\sigma_B))$.
 
  We first assume there is a Nash equilibrium $(\sigma_A,\sigma_B)$ in
  $\calG'$ from $s$ such that $\pi \not\prefrel_A
  \Out(\sigma_A,\sigma_B)$. Since $\prefrel_A$ is total,
  $\Out(\sigma_A,\sigma_B) \prec_A \pi$.  Consider a strategy
  $\sigma'_A$ of player $A$ in $\calG$.  As $(\sigma_A,\sigma_B)$ is a
  Nash equilibrium, it holds that $\Out(\sigma'_A,\sigma_B) \prefrel_A
  \Out(\sigma_A,\sigma_B)$, which implies $\Out(\sigma'_A,\sigma_B)
  \prec_A \pi$.  We conclude that condition $(ii)$ holds.

  Assume now property $(ii)$.  As the preference relation is
  Noetherian, we can select $\pi^+$ which is the largest element for
  $\prefrel_A$ which can be ensured by player $A$. Let $\sigma_A$ be a
  corresponding strategy: for every strategy $\sigma_B$, $\pi^+
  \prefrel_A \Out(\sigma_A,\sigma_B)$.  Towards a contradiction,
  assume now that for every strategy $\sigma'_B$, there exists a
  strategy $\sigma'_A$ such that $\pi^+ \prec_A
  \Out(\sigma'_A,\sigma'_B)$. Consider the set $S$ of such outcomes,
  and define $\pi'$ as its minimal element (this is possible since the
  order $\prefrel_A$ is almost-well-founded). Notice then that $\pi^+
  \prec_A \pi'$, and also that for every strategy $\sigma'_B$, there
  exists a strategy $\sigma'_A$ such that $\pi' \prefrel_A
  \Out(\sigma'_A,\sigma'_B)$.  Then, as the game is determined, we get
  that there exists some strategy $\sigma'_A$ such that for all
  strategy $\sigma'_B$, it holds that $\pi' \prefrel_A
  \Out(\sigma'_A,\sigma'_B)$. In particular, strategy $\sigma'_A$
  ensures $\pi'$, which contradicts the maximality of $\pi^+$.
  Therefore, there is some strategy $\sigma'_B$ for which for every
  strategy $\sigma'_A$, $\pi^+ \not\prec_A \Out(\sigma'_A,\sigma'_B)$,
  which means $\Out(\sigma'_A,\sigma'_B) \prefrel_A \pi^+$.  We show
  now that $(\sigma_A,\sigma'_B)$ is a witness for property $(i)$. We
  have seen on the one hand that $\pi^+ \prefrel_A
  \Out(\sigma_A,\sigma'_B)$, and on the other hand that
  $\Out(\sigma_A,\sigma'_B) \prefrel_A \pi^+$. By~hypothesis, $\pi^+
  \prec_A \pi$, which yields $\Out(\sigma_A,\sigma'_B) \prec_A \pi$.
  Pick another strategy $\sigma'_A$ for player $A$.  We have seen that
  $\Out(\sigma'_A,\sigma'_B) \prefrel_A \pi^+$, which implies
  $\Out(\sigma'_A,\sigma'_B) \prefrel_A
  \Out(\sigma_A,\sigma'_B)$. This concludes the proof of $(i)$.
\end{proof}

\begin{remark}
  Any finite total preorder is obviously Noetherian and
  almost-well-founded. Also, any total preorder isomorphic to the set
  of non-positive integers is Noetherian and almost-well-founded. 
  On the other hand, a total preorder isomorphic to $\{1/n \mid n \in
  \mathbb{N}^+\}$ is Noetherian but not almost-well-founded.
\end{remark}

\subsection{Encoding the value problem as a NE existence problem}
\label{sec:link-value-exist}

We prove a similar result for the NE existence problem. In this reduction
however, we have to modify the game by introducing a truly concurrent
move at the beginning of the game.  This is necessary since for
turn-based games with $\omega$-regular winning conditions, there
always exists a Nash equilibrium~\cite{CMJ04}, hence the NE existence
problem would be trivial.

Let
$\calG=\tuple{\Stat,\Agt,\Act,\Allow,\Tab,(\mathord\prefrel_A)_{A\in\Agt}}$ be
a two-player zero-sum game, with players $A$ and~$B$. Given a state~$s$ of
$\calG$and a play~$\pi$ from~$s$, we~define a game $\calG_\pi$ by adding two
states $s_0$ and~$s_1$, in the very same way as in
Figure~\ref{fig-init-module}, on page~\pageref{fig-init-module}. 
From~$s_0$, $A$~and $B$ play a matching-penny
game to either go to the sink state~$s_1$, or to the state~$s$ in the
game~$\calG$.
We~assume the same hypotheses than in Proposition~\ref{lem:link-value-constr}
for the preference relation~$\prefrel_A$. Let~$\pi^+$ be in the highest
equivalence class for $\prefrel_A$ smaller than~$\pi$ (it exists since
$\prefrel_A$ is Noetherian). In~$\calG_\pi$, player~$B$ prefers runs that end
in~$s_1$: formally, the preference relation~$\prefrel^\pi_B$ of player $B$ in
$\calG_\pi$ is given by $\pi' \prefrel^\pi_B \pi'' \Leftrightarrow \pi'' = s_0
\cdot s_1^\omega \lor \pi' \ne s_0 \cdot s_1^\omega$. On~the other hand,
player~$A$ prefers a path of $\calG$ over going to $s_1$, if and only~if,
it~is at least as good as~$\pi$: formally, the preference relation
$\prefrel^\pi_A$ for player~$A$ in~$\calG_\pi$ is given by $s_0 \cdot \pi'
\prefrel^\pi_A s_0 \cdot \pi'' \Leftrightarrow \pi' \prefrel_A \pi''$, and
$s_0 \cdot s_1^\omega \sim''_A s_0 \cdot \pi^+$.

\begin{proposition}
  \label{lem:link-value-exist}
  Let
  $\calG=\tuple{\Stat,\Agt,\Act,\Allow,\Tab,(\mathord\prefrel_A)_{A\in\Agt} }$
  be a two-player zero-sum game, with players $A$ and $B$, such that:
  \begin{itemize}
  \item the preference relation $\prefrel_{A}$ for player~$A$ is
    total, Noetherian and almost-well-founded;
  \item $\calG$ is determined.
  \end{itemize} 
  Let $s$ be a state and $\pi$ be a play in $\calG$ from $s$. Consider
  the game~$\calG_\pi$ defined above. Then the following two
  properties are equivalent:
  \begin{enumerate}[label=(\roman*)]
  \item there is a Nash equilibrium in $\calG_\pi$ from $s_0$;
  \item player $A$ cannot ensure $\pi$ from $s$ in $\calG$.
  \end{enumerate}
\end{proposition}
\noindent In particular, in a given class of games, if the hardness proof of the
value problem involves a game which satisfies the hypotheses of the
proposition, and if $\calG_\pi$ belongs to that class, then the
NE existence problem is at least as hard as the complement of the value
problem.

\begin{proof}
  Assume that player~$A$ cannot ensure at least $\pi$ from~$s$
  in~$\calG$, then according to
  Proposition~\ref{lem:link-value-constr}, there is a Nash equilibrium
  $(\sigma_A,\sigma_B)$ in the game~$\calG'$ of
  Proposition~\ref{lem:link-value-constr} with outcome $\rho$ such
  that $\pi \not\prefrel_A \rho$.  Consider the strategy profile
  $(\sigma^\pi_A,\sigma^\pi_B)$ in~$\calG_\pi$ that consists in
  playing the same action for both players in~$s_0$, and then if the
  path goes to~$s$, to play according to $(\sigma_A,\sigma_B)$.
  Player~$B$ gets her best possible payoff under that strategy
  profile. If $A$ could change her strategy to get a payoff better
  than $s_0 \cdot \pi^+$, then it would induce a strategy in $\calG'$
  giving her a payoff better than~$\rho$ (when played with strategy
  $\sigma_B$), which contradicts the fact that $(\sigma_A,\sigma_B)$
  is a Nash equilibrium in~$\calG'$. Therefore,
  $(\sigma^\pi_A,\sigma^\pi_B)$ is a Nash equilibrium in~$\calG_\pi$.

  Conversely, assume that $A$ can ensure $\pi$ from $s$ in $\calG$,
  and assume towards a contradiction that there is a Nash equilibrium
  $(\sigma^\pi_A,\sigma^\pi_B)$ in $\calG_\pi$ from $s_0$. Then
  $\Out_{\calG_\pi}(\sigma^\pi_A,\sigma^\pi_B)$ does not end in $s_1$,
  otherwise player $A$ could improve by switching to $s$ and then
  playing according to a strategy which ensures $\pi$. Also,
  $\Out_{\calG_\pi}(\sigma^\pi_A,\sigma^\pi_B)$ cannot end in $\calG$
  either, otherwise player $B$ would improve by switching to $s_1$. We
  get that there is no Nash equilibrium in $\calG_\pi$ from $s_0$,
  which concludes the proof.
\end{proof}

\subsection{Encoding the constrained NE existence problem as an NE 
  existence problem}
\label{sec:link-constr-exist}

The next proposition makes a link between the existence of a Nash
equilibrium where a player gets a payoff larger than some bound and
the (unconstrained) existence of a Nash equilibrium in a new game.
This will allow, in some specific cases, to infer hardness results
from the constrained NE existence problem to the NE existence problem.

The construction is inspired by the previous one, but it applies to a
game with at least two players, and it applies to any two selected
players as follows.  Let
$\calG=\tuple{\Stat,\Agt,\Act,\Allow,\Tab,(\mathord\prefrel_A)_{A\in\Agt} }$
be a concurrent game, $s$ be a~state of $\calG$, $\rho$ be a play from
$s$, and $A_i$ and~$A_j$ be two distinct players.  We define the new
game $E(\calG,A_i,A_j,\rho)$ again in the same way as on
Figure~\ref{fig-init-module}.  Now, in~$s_0$, the two players $A_i$~and
$A_j$~play a matching-penny game to either go to the sink state~$s_1$,
or to state~$s$ in game~$\calG$.

For player~$A_j$, the preference relation in $E(\calG,A_i,A_j,\rho)$
is given by $\prefrel'_{A_j}$ such that $s_0\cdot s_1^\omega
\prec'_{A_j} s_0 \cdot \pi$ and $s_0 \cdot \pi \prefrel'_{A_j} s_0
\cdot \pi' \Leftrightarrow \pi \prefrel_{A_j} \pi'$, for any path
$\pi$ and $\pi'$ from $s$ in $\calG$. For player $A_i$ the preference
relation is $s_0 \cdot \pi \prefrel'_{A_i} s_0 \cdot \pi'
\Leftrightarrow \pi \prefrel_{A_i} \pi'$, for any path $\pi$ and
$\pi'$ from $s$ in $\calG$, and $s_0\cdot s_1^\omega \sim_{A_i} s_0
\cdot \rho$.  For any other player~$A_k$, the preference relation
$E(\calG,A_i,A_j,\rho)$ is given by $s_0 \cdot \pi \prefrel'_{A_k} s_0
\cdot \pi' \Leftrightarrow \pi \prefrel_{A_k} \pi'$ for any path $\pi$
and $\pi'$ from $s$ in $\calG$, and $s_0\cdot s_1^\omega \sim_{A_k}
s_0 \cdot \rho$.

\begin{proposition}\label{lem:link-constr-exist}
  Let
  $\calG=\tuple{\Stat,\Agt,\Act,\Allow,\Tab,(\mathord\prefrel_A)_{A\in\Agt} }$
  be a concurrent game, let $s$ be a state of $\calG$, and $A_i$ and
  $A_j$ be two distinct players participating to $\calG$. Pick two
  plays $\pi$ and $\rho$ from $s$ such that $\rho \prefrel_{A_i}
  \pi$. If there is a Nash equilibrium in $\calG$ whose outcome
  is~$\pi$, then there is a Nash equilibrium in
  $E(\calG,A_i,A_j,\rho)$ whose outcome is ${s_0 \cdot \pi}$.
  Reciprocally, if there is a Nash equilibrium in
  $E(\calG,A_i,A_j,\rho)$ whose outcome is $s_0 \cdot \pi$, then there
  is a Nash equilibrium in $\calG$ whose outcome is~$\pi$.
\end{proposition}

\begin{proof}
  Assume that there is a Nash equilibrium $\sigma_\Agt$ in
  $\calG$ with outcome $\pi$ such that $\rho \prefrel_{A_i} \pi$.
  Then $s_0 \cdot s_1^\omega \prefrel'_{A_i} s_0 \cdot \pi$.  Consider
  the strategy profile in $E(\calG,A_i,A_j,\rho)$ that consists for
  $A_i$ and~$A_j$ in playing different actions in $s_0$ and when the
  path goes to~$s$, to play according to~$\sigma_\Agt$.
  Players $A_i$ and~$A_j$ have no interest in changing their
  strategies in~$s_0$, since for~$A_j$ all plays of~$\calG$ are better
  than $s_0 \cdot s_1^\omega$, and for~$A_i$ the play $s_0 \cdot \pi$
  is better than $s_0 \cdot s_1^\omega$. Hence, this is a Nash
  equilibrium in game $E(\calG,A_i,A_j,\rho)$.

  Reciprocally, if there is a Nash equilibrium in $E(\calG,A_i,A_j,\rho)$, its
  outcome cannot end in~$s_1$, since $A_j$ would have an interest in changing
  her strategy in~$s_0$ (all~plays of~$\calG$ are then better for her). The
  strategies followed from~$s$ thus defines a Nash equilibrium in~$\calG$.
\end{proof}

If we consider a class of games such that $E(\calG,A_i,A_j,\rho)$
belongs to that class when $\calG$ does, then the NE existence problem is
then at least as hard as the constrained NE existence problem. Note
however that the reduction assumes lower bounds on the payoffs, and we
do not have a similar result for upper bounds on the payoffs. For
instance, as we will see in Section~\ref{sec:buchi}, for a conjunction
of B\"uchi objectives, we do not know whether the NE existence problem is
in \P~(as the value problem) or \NP-hard (as is the existence of an
equilibrium where all the players are losing).

\section{The suspect game}
\label{sec:suspect}

In this section, we construct an abstraction of a multi-player
game~$\calG$ as a two-player zero-sum game~$\calH$, such that there is
a correspondence between Nash equilibria in~$\calG$ and winning
strategies in~$\calH$ (formalised in forthcoming
Theorem~\ref{thm:eq-win}).  This transformation does not require the
game to be finite and is conceptually much deeper than the reductions
given in the previous section; it will allow us to use algorithmic
techniques from zero-sum games to compute Nash equilibria and hence
solve the value and (constrained) NE existence problems in various
settings.

\subsection{Construction of the suspect game}

We fix a concurrent game
$\calG=\tuple{\Stat,\Agt,\Act,\Allow,\Tab,(\mathord\prefrel_A)_{A\in\Agt} }$
for the rest of the section, and begin with introducing a few extra
definitions.

\begin{definition}
  \label{def:trigger}
  A strategy profile~$\sigma_\Agt$ is a \emph{trigger profile} for a
  play~$\pi$ from some state~$s$ if, for every player~$A\in \Agt$, for
  every strategy~$\sigma'_A$ of player~$A$, the~path~$\pi$ is at
  least as good as the outcome of $\replaceter \sigma A {\sigma'_A}$
  from~$s$ (that~is, $\Out(s,\replaceter \sigma A {\sigma'_A})
  \prefrel_A \pi$).
\end{definition}

The following result is folklore and a direct consequence of the
definition:
\begin{lemma}
  A Nash equilibrium is a trigger profile for its outcome.
  Reciprocally, a strategy profile which is trigger profile for its
  outcome~is a Nash equilibrium.
\end{lemma}

\begin{definition}[\cite{BBM10a}]
  Given two states~$s$ and~$s'$, and a move~$m_{\Agt}$, the set of
  \newdef{suspect players} for~$(s,s')$ and~$m_{\Agt}$ is the set
  \[
  \Susp((s,s'),m_{\Agt}) = \{A \in\Agt \mid \exists\,m' \in
  \Allow(s,A).\ \Tab(s,\replaceter m A {m'}) = s'\}.
  \]
  Given a path~$\rho$ and a strategy profile~$\sigma_\Agt$, the set of
  suspect players for~$\rho$ and~$\sigma_\Agt$ is the set of players
  that are suspect along each transition of~$\rho$, \ie, it~is the set
  \[
  \Susp(\rho,\sigma_\Agt) = \Bigl\{A \in\Agt \Bigm| \forall i<
  \size\rho.\
  A\in\Susp\bigl((\rho_{=i},\rho_{=i+1}),\sigma_\Agt(\rho_{\leq
    i})\bigr)\Bigr\}.
  \]
\end{definition}
Intuitively, player~$A\in\Agt$ is a suspect for transition~$(s,s')$
and move~$\indicebis mA\Agt$ if she can unilaterally change her action
to activate the transition $(s,s')$: if~$s' \ne \Tab(s,
\indicebis mA\Agt)$, then this may be due to a deviation
from~$m_{\Agt}$ of any of the players in the set
$\Susp((s,s'),m_{\Agt})$, and no one else. If $s' = \Tab(s, \indicebis
mA\Agt)$, it~may simply be the case that no one has deviated, so
everyone is a potential suspect for the next moves.
Similarly, we easily infer that player~$A$ is in
$\Susp(\rho,\sigma_\Agt)$ if, and only if, there is a
strategy~$\sigma'_A$ such that $\Out(\stat,\replaceter \sigma A
{\sigma'_A})=\rho$.

Note that the notion of suspect players requires moves and arenas to
be deterministic, and therefore everything which follows assumes the
restriction to pure strategy profiles and to deterministic game
structures.

\bigskip 
We fix a play $\pi$ in~$\calG$. From game~$\calG$ and play~$\pi$, we build the
\emph{suspect game}~$\calH(\calG,\pi)$,
which is a two-player turn-based game defined as follows.
The players in~$\calH(\calG,\pi)$ are named~\Eve and \Adam.  Since~$\calH(\calG,\pi)$ is
turn-based, its state space can be written as the disjoint union of
the~set~$V_\shortEve$ controlled by~\Eve, which is (a~subset~of)
$\Stat \times 2^\Agt$, and the set~$V_\shortAdam$ controlled by~\Adam,
which~is (a~subset~of) $\Stat \times 2^\Agt \times \Act^\Agt$.  The
game is played in the following way: from a configuration~$(s,P)$
in~$V_\shortEve$, \Eve chooses a legal move~$m_\Agt$ from~$s$; the
next state is $(s,P,m_\Agt)$; then \Adam chooses some state~$s'$
in~$\Stat$, and the new configuration is~$(s',P
\cap\Susp((s,s'),m_\Agt))$.  In~particular, when the state~$s'$ chosen
by \Adam is such that $s'=\Tab(s,m_\Agt)$ (we~say that \Adam
\newdef{obeys}~\Eve when this is the case), then the new configuration
is~$(s',P)$.

We define projections $\Sproj_1$ and~$\Sproj_2$ from $V_\shortEve$ on
$\Stat$ and~$2^{\Agt}$, resp., by $\Sproj_1(s,P) = s$ and
$\Sproj_2(s,P)=P$. We~extend these projections to paths in a
natural~way (but only using \Eve's states in order to avoid
stuttering), letting $\Sproj_1((s_0,P_0) \cdot (s_0,P_0,m_0) \cdot
(s_1,P_1)\cdots ) = s_0 \cdot s_1 \cdots$.
For any play~$\rho$, $\Sproj_2(\rho)$ (seen as a sequence of sets of
players of~$\calG$) is non-increasing, therefore its
limit~$\limitpi2(\rho)$ is well defined.  We notice that if
$\limitpi2(\rho) \ne \emptyset$, then $\Sproj_1(\rho)$ is a play
in~$\calG$.  An~outcome~$\rho$ is \emph{winning for~\Eve}, if for
all~$A\in \limitpi2(\rho)$, it~holds ${\Sproj_1(\rho) \prefrel_A \pi}$.
The~\emph{winning region} $W(\calG, \pi)$ (later simply denoted
by~$W$ when $\calG$ and $\pi$ are clear from the context) is the set
of configurations of~$\calH(\calG, \pi)$ from which \Eve has a
winning strategy.
Intuitively \Eve tries to have the players play a Nash equilibrium,
and \Adam tries to disprove that it is a Nash equilibrium, by~finding
a possible deviation that improves the payoff of one of the players.

\subsection{Correctness of the suspect-game construction}

The next lemma establishes a correspondence between winning strategies
in $\calH(\calG,\pi)$ and trigger profiles (and therefore Nash equilibria) in
$\calG$. 
\begin{lemma}\label{lem:suspect-game}\label{lemma-suspectgame}
  Let $s$ be a state of $\calG$ and $\pi$ be a play from $s$
  in~$\calG$. The following two conditions are equivalent:
  \begin{itemize}
  \item \Eve has a winning strategy in~$\calH(\calG,\pi)$
    from~$(s,\Agt)$, and its outcome~$\rho'$ from~$(s,\Agt)$ when
    \Adam~obeys~\Eve is such that $\Sproj_1(\rho')=\rho$;
  \item there is a trigger profile for $\pi$ in~$\calG$ from
    state~$s$ whose outcome from~$s$ is~$\rho$.
  \end{itemize}
\end{lemma}

\begin{proof}
  Assume there is a winning strategy~$\sigma_\shortEve$ for \Eve
  in~$\calH(\calG,\pi)$ from~$(s,\Agt)$, whose outcome from~$(s,\Agt)$ when \Adam
  obeys \Eve is~$\rho'$ with $\Sproj_1(\rho') = \rho$. We~define the
  strategy profile~$\sigma_\Agt$ according to the actions played
  by~\Eve.
  Pick a history~$g=s_1 s_2\cdots s_{k+1}$, with $s_1=s$.  Let~$h$ be
  the outcome of~$\sigma_\shortEve$ from~$s$ ending in a state
  of~$V_\shortEve$ and such that $\Sproj_1(h)=s_1\cdots s_k$. This
  history is uniquely defined as follows: the first state of~$h$
  is~$(s_1,\Agt)$, and if its $(2i+1)$-st state is~$(s_i,P_i)$, then
  its~$(2i+2)$-nd state is $(s_i,P_i,\sigma_\shortEve(h_{\leq 2i+1}))$
  and its $(2i+3)$-rd state is $(s_{i+1}, P_i\cap
  \Susp((s_i,s_{i+1}),\sigma_\shortEve(h_{\leq 2i+1})))$.
  Now, write $(s_k,P_k)$ for the last state of~$h$, and let $h'=h\cdot
  (s_k,P_k,\sigma_\shortEve(h))\cdot (s_{k+1}, P_k\cap
  \Susp((s_k,s_{k+1}), \sigma_\shortEve(h)))$. Then we define
  $\sigma_\Agt(g)=\sigma_\shortEve(h')$.  Notice that when $g\cdot s$
  is a prefix of~$\Sproj_1(\rho')$, then $g\cdot s\cdot
  \sigma_\Agt(g\cdot s)$ is also a prefix of~$\Sproj_1(\rho')$. In
  particular, $\Out(s,\sigma_\Agt) = \Sproj_1(\rho') = \rho$.

  We~now prove that $\sigma_\Agt$ is a trigger profile for~$\pi$.
  Pick a player~$A\in \Agt$, a~strategy~$\sigma'_A$ for player~$A$,
  and let $g=\Out(\stat, \replaceter \sigma A {\sigma'_A})$.  With a
  play~$g$, we~associate a play~$h$ in~$\calH(\calG,\pi)$ in the same way as
  above. Then player~$A$ is a suspect along all the transitions
  of~$g$, so that she~belongs to~$\limitpi2(h)$. Now,
  as~$\sigma_\shortEve$~is winning, $\Sproj_1(h) \prefrel_A \pi$, which
  proves that $\sigma_\Agt$ is a trigger profile.

  \medskip Conversely, assume that $\sigma_\Agt$ is a trigger profile
  for~$\pi$ whose outcome is~$\rho$, and define the
  strategy~$\sigma_\shortEve$ by $\sigma_\shortEve(h) =
  \sigma_\Agt(\Sproj_1(h))$. Notice that the outcome~$\rho'$
  of~$\sigma_\shortEve$ when \Adam obeys~\Eve satisfies
  $\Sproj_1(\rho')=\rho$.

  Let~$\eta$ be an outcome of~$\sigma_\shortEve$ from~$\stat$,
  and~$A\in \limitpi2(\eta)$. Then $A$~is a suspect for each
  transition along~$\Sproj_1(\eta)$, which means that for all~$i$, there
  is a move~$m^A_i$ such that
  \[
  \Sproj_1(\eta_{=i+1}) = \Tab(\Sproj_1(\eta_{=i}),
  \sigma_\Agt(\Sproj_1(\eta_{\leq i}))[A\mapsto m^A_i]).
  \]
  Therefore there is a strategy $\sigma_A'$ such that
  $\Sproj_1(\eta)=\Out(s,\sigma_\Agt[A\mapsto\sigma_A'])$. Since
  $\sigma_\Agt$ is a trigger profile for~$\pi$, it holds that
  $\Sproj_1(\eta) \prefrel_A \pi$. As~this holds for any~$A\in
  \limitpi2(\eta)$, $\sigma_\shortEve$ is winning.
\end{proof}

We now state the correctness theorem for the suspect game construction.
\begin{theorem}\label{thm:eq-win}
  Let
  $\calG=\tuple{\Stat,\Agt,\Act,\Allow,\Tab,(\mathord\prefrel_A)_{A\in\Agt} }$
  be a concurrent game, $s$ be a state of~$\calG$, and $\pi$ be a play
  in~$\calG$.  The following two conditions are equivalent:
  \begin{itemize}
  \item there is a Nash equilibrium $\sigma_\Agt$ from~$s$ in~$\calG$
    whose outcome is~$\pi$.
  \item there is a play~$\rho$ from $(s,\Agt)$ in~$\calH(\calG,\pi)$, 
    \begin{enumerate}
    \item \label{cond:proj} such that $\Sproj_1(\rho)=\pi$;
    \item \label{cond:obey} along which \Adam always obeys~\Eve; and 
    \item\label{cond:win} such that for all indices~$i$, there is a
      strategy $\sigma^i_\shortEve$ for~\Eve, for which any play in
      $\rho_{\le i} \cdot \Out(\rho_{=i}, \sigma^i_\shortEve)$ is
      winning for~\Eve.
    \end{enumerate}
  \end{itemize}
\end{theorem}

\begin{proof}
  The Nash equilibrium is a trigger profile, and from
  Lemma~\ref{lemma-suspectgame}, we~get a winning
  strategy~$\sigma_\shortEve$ in~$\calH(\calG,\pi)$. The outcome~$\rho$
  of~$\sigma_\shortEve$ from~$s$ when \Adam obeys~\Eve is such that
  $\pi=\Sproj_1(\rho)$ is the outcome of the Nash equilibrium.  Now
  for all prefix $\rho_{\le i}$, the strategy
  $\sigma_\shortEve^i\colon h \mapsto \sigma_\shortEve(\rho_{\le
    i}\cdot h)$ is such that any play in $\rho_{\le i}\cdot
  \Out(\rho_{=i},\sigma^i_\shortEve)$ is winning for~\Eve.

  \medskip Conversely, let $\rho'$ be a path in~$\calH(\calG,\pi)$ and
  assume it satisfies all three conditions. We~define a
  strategy~$\lambda_\shortEve$ that follows~$\rho'$ when
  \Adam~obeys. Along~$\rho'$, this strategy is defined as follows:
  $\lambda_\shortEve(\rho'_{\leq 2i}) = m_\Agt$ such that
  $\Tab(\Sproj_1(\rho'_{=i}),m_\Agt) = \Sproj_1(\rho'_{=i+1})$. Such a
  legal move must exist since \Adam obeys~\Eve along~$\rho'$ by
  condition~\ref{cond:obey}. Now, if \Adam deviates from the obeying
  strategy (at step~$i$), we~make~$\lambda_\shortEve$ follow the
  strategy~$\sigma_\shortEve^i$ (given by condition~\ref{cond:win}),
  which will ensure that the outcome is winning for~\Eve.

  The outcomes of~$\lambda_\shortEve$ are then either the path~$\rho'$,
  or a path~$\rho''$ obtained by following a winning strategy after a
  prefix of~$\rho'$.  The path~$\rho''$ is losing for~\Adam, hence for
  all $A\in \limitpi2(\rho')$, $\rho''\prefrel_A \rho'$.  This proves
  that $\lambda_\shortEve$ is a winning strategy.  Applying
  Lemma~\ref{lemma-suspectgame}, we~obtain a strategy
  profile~$\sigma_\Agt$ in~$\calG$ that is a trigger profile
  for~$\pi$. Moreover, the~outcome of~$\sigma_\Agt$ from~$s$
  is~$\Sproj_1(\rho')$ (using condition~\ref{cond:proj}), so that
  $\sigma_\Agt$ is a Nash equilibrium.
\end{proof}

\begin{remark}\label{rem:prefix-independant}
  Assume the preference relations of each player~$A$ in $\calG$ are
  prefix-independent, \ie, for all plays $\rho$ and $\rho'$, $\rho
  \prefrel_A \rho'$ iff for all indices $i$ and $j$, $\rho_{\ge i}
  \prefrel_A \rho'_{\ge j}$.  Then the winning condition of \Eve is
  also prefix-independent, and condition~\ref{cond:win} just states
  that $\rho'$ has to stay within the winning region of~\Eve.  Note
  that, for prefix-dependent preference relations,
  condition~\ref{cond:win} does not reduce to stay within the winning
  region of~\Eve: for instance, for safety objectives, if the losing
  states of all the players have been visited then any prolongation
  will satisfy the condition, even though it might leave the winning
  region of~\Eve.
\end{remark}

\begin{example}
  We depict on Figure~\ref{fig-suspg} part of the suspect game for the game of
  Figure~\ref{fig-ex}. Note that the structure of $\calH(\calG,\pi)$
  does not depend on~$\pi$. Only the winning condition is affected by
  the choice of~$\pi$.
  \begin{figure}[t]
  \begin{center}
    \begin{tikzpicture}[thick]
      \everymath{\scriptstyle}
      \draw (0,1.5) node [draw,dashed] (A) {$\ell_0,\{\pl 1,\pl 2\}$};
      \draw (3,3.5) node [draw] (B1) {$\ell_0,\{\pl 1,\pl 2\},\langle 1,1 \rangle$};
      \draw (3,2) node [draw] (B2) {$\ell_0,\{\pl 1,\pl 2\},\langle 1,2 \rangle$};
      \draw (3,1) node [draw] (B3) {$\ell_0,\{\pl 1,\pl 2\},\langle 2,1 \rangle$};
      \draw (3,-.5) node [draw] (B4) {$\ell_0,\{\pl 1,\pl 2\},\langle 2,2 \rangle$};
      \draw [-latex'] (A) -- (B1);
      \draw [-latex'] (A) -- (B2);
      \draw [-latex'] (A) -- (B3);
      \draw [-latex'] (A) -- (B4);
      \begin{scope}[yshift=1mm]
      \draw (6,5.5) node [draw,dashed] (C10) {$\ell_0,\emptyset$};
      \draw (6,4.5) node [draw,dashed] (C11) {$\ell_1,\{\pl 1,\pl 2\}$};
      \draw (6,3.5) node [draw,dashed] (C12) {$\ell_2,\{\pl 2\}$};
      \draw (6,2.5) node [draw,dashed] (C13) {$\ell_3,\{\pl 1\}$};

      \draw (9,4) node [draw] (D1111) {$\ell_1,\{\pl 1,\pl 2\}, \langle 1,1\rangle$};
      \draw (9,5) node [draw] (D1112) {$\ell_1,\{\pl 1,\pl 2\}, \langle 1,2\rangle$};
      \draw (9,3) node [draw] (D1211) {$\ell_2,\{\pl 2\}, \langle 1,1\rangle$};
      \draw (9,2) node [draw] (D1311) {$\ell_3,\{\pl 1\}, \langle 1,1\rangle$};
      \end{scope}

      \begin{scope}[yshift=-3mm]
      \draw (6,-1) node [draw,dashed] (C41) {$\ell_1,\emptyset$};
      \draw (6,0) node [draw,dashed] (C42) {$\ell_2,\{\pl 1\}$};
      \draw (6,1) node [draw,dashed] (C43) {$\ell_3,\{\pl 2\}$};

      \draw (9,0) node [draw] (D4211) {$\ell_2,\{\pl 1\}, \langle 1,1\rangle$};
      \draw (9,1) node [draw] (D4311) {$\ell_3,\{\pl 2\}, \langle 1,1\rangle$};
      \end{scope}
      \draw [-latex'] (B1) -- (C10);
      \draw [-latex'] (B1) -- (C11);
      \draw [-latex'] (B1) -- (C12);
      \draw [-latex'] (B1) -- (C13);

      \draw [-latex'] (B4) .. controls +(180:2cm) and +(-90:2cm) .. (A);
      \draw [-latex'] (B4) -- (C41);
      \draw [-latex'] (B4) -- (C42);
      \draw [-latex'] (B4) -- (C43);

      \draw[-latex'] (C11.-20) .. controls +(-20:5mm) .. (D1111);
      \draw[-latex'] (C11) -- (D1112);
      \draw[-latex'] (C12) -- (D1211);
      \draw[-latex'] (C13.-20) .. controls +(-20:5mm) ..  (D1311);
      \draw[-latex'] (D1111.160) .. controls +(160:5mm) .. (C11);
      \draw[-latex'] (D1111) -- (C12);

      \draw[-latex'] (C42) -- (D4211);
      \draw[-latex'] (C43.-10) .. controls +(-10:5mm) .. (D4311.-170);
      \draw[-latex'] (D4311.170) .. controls +(170:5mm) .. (C43.10);

      \draw[-latex'] (D1311.160)  .. controls +(160:5mm) ..  (C13);

      \draw[dashed] (B2) -- +(15:18mm);
      \draw[dashed] (B2) -- +(5:18mm);
      \draw[dashed] (B2) -- +(-5:18mm);
      \draw[dashed] (B2) -- +(-15:18mm);
      \draw[dashed] (B3) -- +(15:18mm);
      \draw[dashed] (B3) -- +(5:18mm);
      \draw[dashed] (B3) -- +(-5:18mm);
      \draw[dashed] (B3) -- +(-15:18mm);
      \draw[dashed] (D1112) -- +(15:18mm);
      \draw[dashed] (D1112) -- +(5:18mm);
      \draw[dashed] (D1112) -- +(-5:18mm);
      \draw[dashed] (D1112) -- +(-15:18mm);
      \draw[dashed] (D1111) -- +(5:18mm);
      \draw[dashed] (D1111) -- +(-5:18mm);
      \draw[dashed] (D1211) -- +(15:18mm);
      \draw[dashed] (D1211) -- +(5:18mm);
      \draw[dashed] (D1211) -- +(-5:18mm);
      \draw[dashed] (D1211) -- +(-15:18mm);
      \draw[dashed] (D1311) -- +(10:18mm);
      \draw[dashed] (D1311) -- +(0:18mm);
      \draw[dashed] (D1311) -- +(-10:18mm);
      \draw[dashed] (D4311) -- +(10:18mm);
      \draw[dashed] (D4311) -- +(0:18mm);
      \draw[dashed] (D4311) -- +(-10:18mm);
      \draw[dashed] (D4211) -- +(15:18mm);
      \draw[dashed] (D4211) -- +(5:18mm);
      \draw[dashed] (D4211) -- +(-5:18mm);
      \draw[dashed] (D4211) -- +(-15:18mm);      
    \end{tikzpicture}
  \end{center}
  \caption{A small part of the suspect game for the game of
  Figure~\ref{fig-ex}}\label{fig-suspg}
\end{figure}
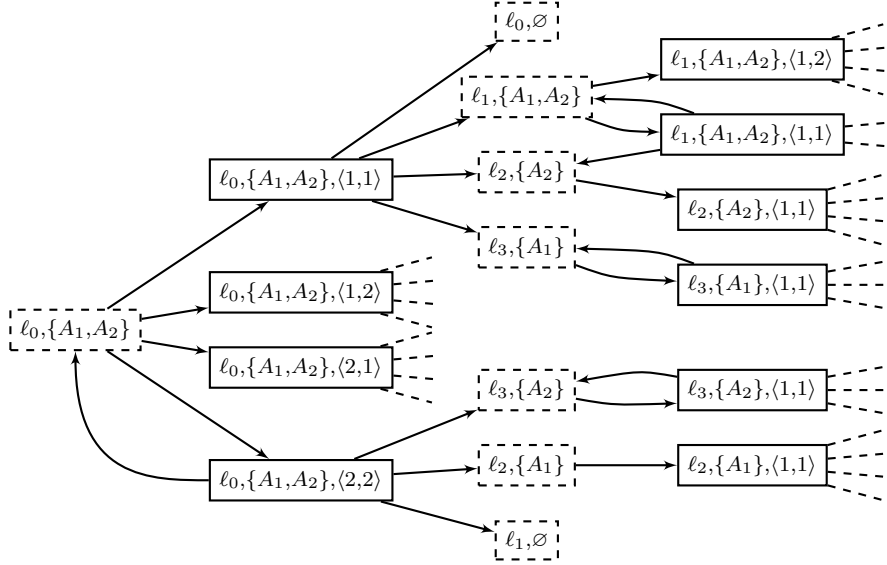
\end{example}

In the rest of the paper, we use the suspect-game construction to
algorithmically solve the NE existence problem and the constrained NE
existence problem in finite games for large classes of preference relations.
Before that we carefully analyse the size of the suspect game when the
original game is finite.

\subsection{Size of the suspect games when the original game is finite}
We suppose that $\calG$ is finite.
At first sight, the number of states in~$\calH(\calG,\pi)$ is exponential (in~the
number of players of~$\calG$). However, there are two cases for which
we easily see that the number of states of~$\calH(\calG,\pi)$ is actually only
polynomial:
\begin{itemize}
\item if there is a state in which all the players have several
  possible moves, then the transition table (which is part of the
  input, as discussed in Remark~\ref{remark:encoding}) 
  is also exponential in the number of players;
\item if the game is turn-based, then the transition table is ``small'', but
  there is always at most one suspect player (unless all of them are
  suspects), so that the number of reachable states in~$\calH(\calG,\pi)$ is also small.
\end{itemize}
We now prove that, due to the explicit encoding of the set of
transitions (recall Remark~\ref{remark:encoding},
page~\pageref{remark:encoding}), this can be generalised:
\begin{proposition}\label{lem:polynomial-size}
  Let $\calG=\tuple{\Stat,\Agt,\Act,\Allow,\Tab,(\mathord\prefrel_A)_{A\in\Agt}
  }$ be a finite concurrent game and $\pi$ be a play in $\calG$. The
  number of reachable configurations from $\Stat\times \{\Agt\}$
  in~$\calH(\calG,\pi)$ is polynomial in the size of~$\calG$.
\end{proposition}
\begin{proof}
  The game~$\calH(\calG,\pi)$ contains the state~$(s,\Agt)$ and the states
  $(s,\Agt,m_\Agt)$, where $m_\Agt$ is a legal move from~$s$; the
  number of these states is bounded by ${|\Stat| + |\Tab|}$.
  The successors of those states that are not of the same form, are
  the $(t,\Susp((s,t),m_\Agt))$ with $t \ne \Tab(s,m_\Agt)$. If~some
  player~$A\in\Agt$ is a suspect for transition~$(s,t)$, then
  besides~$m_A$, she~must have at least a second action~$m'$, for
  which $\Tab(s,m_\Agt [A\mapsto m']) = t$.  Thus the transition table
  from state~$s$ has size at least $2^{|\Susp((s,t),m_\Agt)|}$.
  The successors of $(t,\Susp((s,t),m_\Agt))$ are of the form $(t',P)$
  or $(t',P,m_\Agt)$ where $P$~is a subset of $\Susp((s,t),m_\Agt)$;
  there can be no more than $(|\Stat| + |\Tab|) \cdot
  2^{|\Susp((s,t),m_\Agt)|}$ of them, which is bounded by $(|\Stat| +
  |\Tab|)\cdot |\Tab|$. The total number of reachable states is then
  bounded by $(|\Stat| + |\Tab|) \cdot (1 + (|\Stat| + |\Tab|) \cdot
  |\Tab|)$. 
\end{proof}

\section{Single-objective preference  relations}\label{sec:single}

In this section we will be interested in finite games with
single-objective preference relations.

\medskip The value problem for finite concurrent games with $\omega$-regular
objectives has standard solutions in game theory; they are given in
Table~\ref{table-single} (page~\pageref{table-single}). Let us briefly give
some explanations. Most of the basic literature on two-player games focus on
turn-based games, and in particular algorithms for solving two-player games
with $\omega$-regular objectives only deal with turn-based games (see for
instance~\cite[Chapter~2]{GTW02}). In~particular, McNaughton developed an
algorithm to solve turn-based parity games in time $O(|\Stat| \cdot |
\Edg|^{p-1})$, where $p-1$ is the number of priorities~\cite{McNaughton93}.
B\"uchi games and co-B\"uchi games correspond to parity games with two
priorities, hence they are solvable in polynomial time. Similarly reachability
games and safety games can be transformed into B\"uchi games by making the
target states absorbing. Hence turn-based game with these types of objectives
can be solved in polynomial time.

Note however that we can reuse these algorithms in the concurrent case as follows. Any finite
concurrent zero-sum game with objective $\Omega$ for player $A_1$ can
be transformed into a turn-based zero-sum game with objective
$\widetilde\Omega$ for player~$A_1$: the idea is to replace any edge
labelled with pair of actions $\langle a_1,a_2\rangle$ into two
consecutive transitions labelled with $a_1$ (belonging to player
$A_1$) and with $a_2$ (belonging to player $A_2$). Furthermore
$\Omega$ is an $\omega$-regular condition, then so is
$\widetilde\Omega$, and the type of the objective (reachability,
B\"uchi, etc) is preserved (note however that this transformation only
preserves Player~$A_1$ objective). Hence the standard algorithm on the
resulting turn-based game can be applied. 
Lower bounds for reachability/safety
and B\"uchi/co-B\"uchi games are also folklore results, and can be
obtained by encoding the circuit-value problem (we recall the encoding
in Section~\ref{ptime-hard}).

\medskip We now focus on the NE existence problem and on the constrained
NE existence problem when each player has a single ($\omega$-regular)
objective using the suspect game construction. The results are
summarised in the second column of Table~\ref{table-single}.

Streett and Muller objectives are not explicitly mentioned in the
rest of the section. The complexity of their respective (constrained)
NE existence problems, which is given in Table~\ref{table-single}, can
easily be inferred from other ones.  The $\P^\NP_\parallel$-hardness
for the NE existence problem with Streett objectives follows from the
corresponding hardness for parity objectives (parity objectives can be
encoded efficiently as Streett objectives).  Hardness for the
NE existence problem in Muller games, is deduced from hardness of the
value problem (which holds for turn-based games), applying
Proposition~\ref{lem:link-value-exist}.  For both objectives,
membership in \PSPACE follows from \PSPACE membership for objectives
given as Boolean circuits, since they can efficiently be encoded as
Boolean circuits.

We fix for the rest of the section a multi-player finite game $\calG =
\tuple{\Stat,\Agt,\Act,\Allow,\Tab,(\mathord\prefrel_A)_{A\in\Agt} }$, and we
assume that each $\prefrel_A$ is single-objective, given by set
$\Omega_A$.

\begin{remark}
  \label{remark:explosion}
  Let us come back to Remark~\ref{remark:encoding} on our choice of an
  explicit encoding for the set of transitions. Assuming more compact
  encodings, the complexity of computing Nash equilibria for
  qualitative objectives does not allow to distinguish between the
  intrinsic complexity of the objectives. Indeed, in the formalism
  of~\cite{LMO08}, the transition function is given in each state by a
  finite sequence $((\phi_0 , s_0 ), ..., (\phi_h , s_h ))$, where
  $s_i \in \Stat$, and $\phi_i$ is a boolean combination of
  propositions $(A = \act)$ that evaluates to true iff agent $A$
  chooses action~$\act$.  The transition table is then defined as
  follows: $\Tab(s, m_\Agt) = s_j$ iff $j$ is the smallest index such
  that $\phi_j$ evaluates to true when, for every player $A\in\Agt$,
  $A$ chooses action $m_A$. It is required that the last boolean
  formula $\phi_h$ be $\top$, so that no agent can enforce a deadlock.
  
  We can actually state the following result, whose proof is postponed
  to the Appendix on page~\pageref{app}.
  \begin{proposition}\label{proposition:explosion}
    For finite concurrent games with compact encoding of transition
    functions and with reachability/B\"uchi/safety objectives, the
    constrained NE existence problems is \PSPACE-hard.
  \end{proposition}
\end{remark}

\begin{remark}
\label{simplification}
It is first interesting to notice that given two plays $\pi$
and~$\pi'$ the suspect games $\calH(\calG,\pi)$
and~$\calH(\calG,\pi')$ only differ in their winning conditions.
In~particular, the structure of the game only depends on~$\calG$, and
has polynomial size (see Proposition~\ref{lem:polynomial-size}).
We~denote it with~$\calJ(\calG)$. Moreover, as~each
relation~$\prefrel_A$ is given by a single objective~$\Omega_A$, the
winning condition for \Eve in $\calH(\calG,\pi)$ rewrites~as: for
every $A \in \limitpi2(\rho) \cap \Losers(\pi)$, $\Sproj_1(\rho)$~is
losing (in~$\calG$) for player~$A$, where $\Losers(\pi)$ is the set of
players losing along $\pi$ in $\calG$.
This winning condition only depends on $\Losers(\pi)$ (not on the
precise value of play~$\pi$). Therefore in this section, the suspect
game is denoted with~$\calH(\calG,L)$, where $L \subseteq \Agt$, and
\Eve wins play~$\rho$ if, for~every $A \in \limitpi2(\rho) \cap L$,
$A$~loses along $\Sproj_1(\rho)$ in~$\calG$. In~many cases we will be
able to simplify this winning condition, and to obtain simple
algorithms to the corresponding problems.
\end{remark}

We now distinguish between the winning objectives of the
players. There are some similarities in some of the cases (for
instance safety and co-B\"uchi objectives), but they nevertheless all
require specific techniques and proofs.

\subsection{Reachability objectives}
\label{subsec:reachability}

The value problem for a reachability winning condition is \P-complete. Below,
we design a non-deterministic algorithm that runs in polynomial time for
solving the constrained NE existence problem. We~then end this subsection with a
\NP-hardness proof of the constrained NE existence problem and NE existence
problem. In the end, we prove the following result:

\begin{theorem}
  For finite concurrent games with single reachability objectives, the NE
  existence problem and the constrained NE existence problem are \NP-complete.
\end{theorem}

\subsubsection{Reduction to a safety game}
We assume that for every player $A$, $\Omega_A$ is a single
reachability objective given by target set~$T_A$.  Given $L \subseteq
\Agt$, in the suspect game~$\calH(\calG,L)$, we show that the
objective of \Eve reduces to a safety objective.  We define the safety
objective~$\Omega_L$ in $\calH(\calG,L)$ by the set $T_L = \{(s,P)
\mid \exists A\in P \cap L.\ s \in T_A\}$ of target states.

\begin{lemma}
  \label{lemma:reach-to-safety}
  \Eve has a winning strategy in game $\calH(\calG,L)$ iff \Eve has a
  winning strategy in game $\calJ(\calG)$ with safety
  objective~$\Omega_L$.
\end{lemma}
\begin{proof}
  We first show that any play in $\Omega_L$ is winning in
  $\calH(\calG,L)$. Let $\rho \in \Omega_L$, and let $A \in
  \limitpi2(\rho) \cap L$. Toward a contradiction assume that
  $\Occ(\Sproj_1(\rho)) \cap T_A \ne \emptyset$: there is a state
  $(s,P)$ along $\rho$ with $s \in T_A$. Obviously $\limitpi2(\rho)
  \subseteq P$, which implies that $A \in P \cap L$. This contradicts
  the fact that $\rho \notin \Omega_L$.  We have shown so far that any
  winning strategy for \Eve in $\calJ(\calG)$ with safety objective
  $\Omega_L$ is a winning strategy for \Eve in $\calH(\calG,L)$.

  Now assume that \Eve has no winning strategy in game $\calJ(\calG)$
  with safety objective~$\Omega_L$. Turn-based games with safety
  objectives being determined, \Adam has a
  strategy~$\sigma_\shortAdam$ which ensures that no outcome of
  $\sigma_\shortAdam$ is in $\Omega_L$. If $\rho \in
  \Out(\sigma_\shortAdam)$, there is a state~$(s,P)$ along $\rho$ such
  that there is $A\in P \cap L$ with $s\in T_A$.  We now modify the
  strategy of \Adam such that as soon as such a state is reached we
  switch from $\sigma_\shortAdam$ to the strategy that always obeys
  \Eve. This ensures that in every outcome~$\rho'$ of the new
  strategy, we reach a state $(s,P)$ such that there is $A\in P \cap
  L$ with $s\in T_A$, and $\limitpi2(\rho') = P$. This \Adam's
  strategy thus makes \Eve lose the game~$\calH(\calG,L)$, and \Eve
  has no winning strategy in game~$\calH(\calG,L)$.
\end{proof}

\subsubsection{Algorithm}
The algorithm for solving the constrained NE existence problem in a game
where each player has a single reachability objective relies on
Theorem~\ref{thm:eq-win} and Proposition~\ref{lem:play-length}, and on
the above analysis:
\begin{enumerate}[label=(\roman*)]
\item\label{step1} guess a lasso-shaped play $\rho = \tau_1 \cdot
  \tau_2^\omega$ (with $|\tau_i| \le 2 |\Stat|^2$) in $\calJ(\calG)$,
  such that \Adam obeys \Eve along $\rho$, and $\pi = \Sproj_1 (\rho)$
  satisfies the constraint on the payoff;
\item\label{step3} compute the set $W(\calG,\Losers(\pi))$ of states that are
  winning for~\Eve in the suspect game $\calH(\calG,\Losers(\pi))$, where
  $\Losers(\pi)$ is the set of losing players along~$\pi$;
\item\label{step4} check that $\rho$ stays in $W(\calG,\Losers(\pi))$.
\end{enumerate}
First notice that this algorithm is non-deterministic and runs in
polynomial time: the witness~$\rho$ guessed in step~\ref{step1} has
size polynomial; the suspect game $\calH(\calG,\Losers(\pi))$ has also
polynomial size (Proposition~\ref{lem:polynomial-size});
Step~\ref{step3} can be done in polynomial time using a standard
attractor computation~\cite[Sect.~2.5.1]{GTW02} as the game under
analysis is equivalent to a safety game
(Lemma~\ref{lemma:reach-to-safety}); finally step~\ref{step4} can
obviously be performed in polynomial time.

Step~\ref{step1} ensures that conditions~\ref{cond:obey}
and~\ref{cond:proj} of Theorem~\ref{thm:eq-win} hold for $\rho$ and
step~\ref{step4} ensures condition~\ref{cond:win}.  Correctness of the
algorithm then follows from Theorem~\ref{thm:eq-win} and
Proposition~\ref{lem:play-length}.

\subsubsection{Hardness}
We prove \NP-hardness of the constrained NE existence problem by encoding an
instance of \SAT as follows. We assume set of atomic propositions $\AP =
\{x_1,\dots,x_k\}$, and we let $\phi = \bigwedge_{i=1}^n c_i$ where $c_i =
\ell_{i,1} \lor \ell_{i,2} \lor \ell_{i,3}$ where $\ell_{i,j} \in \{ x_k, \lnot
x_k \mid 1\leq k\leq p\}$. We~build the turn-based game $\calG_\phi$ with $n+1$
players $\Agt = \{ A, C_1,\dots , C_n\}$ as follows: for every $1 \le k \le
p$, player~$A$ chooses to visit either location $x_k$ or location $\lnot x_k$.
Location~$x_k$ is winning for player~$C_i$ if, and only~if, $x_k$~is one of
the literals in~$c_i$, and similarly location $\lnot x_k$ is winning for~$C_i$
if, and only~if, $\lnot x_k$ is one of the literals of~$c_i$. The construction
is illustrated on Figure~\ref{fig-jeuNP}, with the reachability objectives
defined as $\Omega_{C_i} = \{\ell_{i,1}, \ell_{i,2}, \ell_{i,3}\}$ for $1 \le
i \le n$. Now, it is easy to check that this game has a Nash equilibrium with
payoff~1 for all players $(C_i)_{1 \le i \le n}$ if, and only~if, $\phi$ is
satisfiable.

We prove hardness for the NE existence problem by using the transformation
described in Section~\ref{sec:link-constr-exist} once for each player. We
define the game $\calG_0$ similar to $\calG$ but with an extra
player~$C_{n+1}$ who does not control any state for now. For $1\leq i \leq n$,
we define $\calG_i = E(\calG_{i-1},C_i, C_{n+1},\rho)$, where $\rho$ is a
winning path for~$C_i$. The preference relation can be expressed in any
$\calG_i$ by a reachability condition, by giving to $C_{n+1}$ a target which
is the initial state of~$\calG$. According to
Proposition~\ref{lem:link-constr-exist} there is a Nash equilibrium
in~$\calG_i$ if, and only~if, there is one in~$\calG_{i-1}$ where $C_i$~wins.
Therefore there is a Nash equilibrium in $\calG_n$ if, and only~if, $\phi$~is
satisfiable. This entails \NP-hardness of the NE existence problem.

\begin{figure*}[ht]
  \centering
  \begin{tikzpicture}[thick,xscale=.8,yscale=.8]
    \tikzset{rond/.style={circle,draw=black,thick,fill=black!10,minimum
        height=6.5mm,inner sep=0pt}}
    \tikzset{carre/.style={draw=black,thick,fill=black!10,minimum
        height=5mm,minimum width=5mm,inner sep=0pt}}
    \everymath{\scriptstyle}
    \path (-.3,1) node {$\displaystyle \calG_\phi$};
    \draw (0,0) node [carre](choix-p1) {$A$};
    \draw [latex'-] (choix-p1.180) -- ++(-.4,0);
    \draw (1.5,1) node [rond] (p1) {$x_1$};
    \draw (1.5,-1) node [rond] (nonp1) {$\neg x_1$};
    \draw [-latex'] (choix-p1) -- (p1);
    \draw [-latex'] (choix-p1) -- (nonp1);
    \draw (3,0) node [carre] (choix-p2) {$A$};
    \draw (4.5,1) node [rond] (p2) {$x_2$};
    \draw (4.5,-1) node [rond] (nonp2) {$\neg x_2$};
    \draw [-latex'] (p1) -- (choix-p2);
    \draw [-latex'] (nonp1) -- (choix-p2);
    \draw [-latex'] (choix-p2) -- (p2);
    \draw [-latex'] (choix-p2) -- (nonp2);
    \draw (6,0) node [carre] (choix-p3) {$A$};
    \draw [-latex'] (p2) -- (choix-p3);
    \draw [-latex'] (nonp2) -- (choix-p3);
    \draw[dashed] (choix-p3) -- +(.75,.5);
    \draw[dashed] (choix-p3) -- +(.75,-.5);
    \draw (7.75,0) node {\Large\dots};
    \draw (9.5,1) node [rond] (ph) {$x_p$};
    \draw (9.5,-1) node [rond] (nonph) {$\neg x_p$};
    \draw[dashed,latex'-] (ph) -- +(-.75,-.5);
    \draw[dashed,latex'-] (nonph) -- +(-.75,.5);
    \draw (11,0) node [rond] (final) {};
    \draw [-latex'] (final) .. controls +(120:36pt) and +(60:36pt) .. (final);
    \draw [-latex'] (ph) -- (final);
    \draw [-latex'] (nonph) -- (final);
  \end{tikzpicture}
  \caption{Reachability game for the reduction of \SAT}
  \label{fig-jeuNP}
\end{figure*}
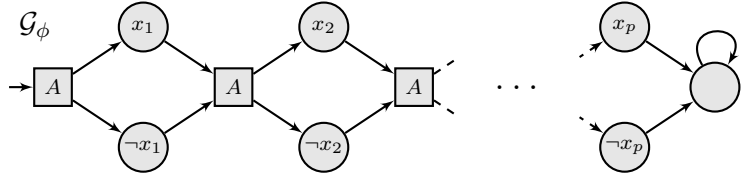

\subsection{Safety objectives}\label{subsec:safety}
The value problem for safety objectives is \P-complete. We next show
that the constrained NE existence problem can be solved in \NP, and
conclude with \NP-hardness of both the constrained NE existence problem and 
the NE existence problem. We hence prove:

\begin{theorem}
  For finite games with single safety objectives, the NE existence problem and
  the constrained NE existence problem are \NP-complete.
\end{theorem}

\subsubsection{Reduction to  a conjunction of reachability objectives}
We assume $\Omega_A$ is a single safety objective given by set~$T_A$.
In the corresponding suspect game, we show that the goal of \Eve is
equivalent to a conjunction of reachability objectives.  Let $L
\subseteq \Agt$. In suspect game $\calH(\calG,L)$, we define several
reachability objectives as follows: for each $A \in L$, we define
$T'_A = T_A \times \{P \mid P \subseteq \Agt\} \cup \Stat \times \{ P
\mid A \not\in P\}$, and we write $\Omega'_A$ for the corresponding
reachability objectives.
\begin{lemma}
  A play $\rho$ is winning for \Eve in $\calH(\calG,L)$ iff $\rho \in
  \bigcap_{A \in L} \Omega'_A$.
\end{lemma}

\begin{proof}
  Let $\rho$ be a play in $\calH(\calG,L)$, and assume it is winning
  for \Eve.  Then, for each $A \in \limitpi2(\rho)\cap L$, $\rho
  \notin \Omega_A$, which means that the target set $T_A$ is visited
  along $\Sproj_1(\rho)$, and therefore $T'_A$ is visited along $\rho$.
  If $A \notin \limitpi2(\rho)$, then a state $(s,P)$ with $A \notin
  P$ is visited by $\rho$: the target set $T'_A$ is visited. This
  implies that $\rho \in \bigcap_{A \in L} \Omega'_A$.

  Conversely let $\rho \in \bigcap_{A \in L} \Omega'_A$. For every $A
  \in L$, $T'_A$ is visited by~$\rho$. Then, either $T_A$ is visited
  by $\Sproj_1(\rho)$ (which means that $\rho \notin \Omega_A$) or
  $A\not\in \limitpi2(\rho)$.  In~particular, $\rho$~is a winning play
  for \Eve in~$\calH(\calG,L)$.
\end{proof}

\subsubsection{Algorithm for solving finite zero-sum turn-based games
  with a conjunction of reachability objectives} \label{conj-reach} We
now give a simple algorithm for solving zero-sum games with a
conjunction of reachability objectives. This algorithm works in
exponential time with respect to the size of the conjunction (we~will
see in Subsection~\ref{subsec:boo-reach} that the problem is
\PSPACE-complete). However for computing Nash equilibria in safety
games we will only use it for small (logarithmic size) conjunctions.

Let $\overline\calG$ be a two-player turn-based game with a winning
objective for \Eve given as a conjunction of $k$~reachability objectives
$\Omega_1,\dots,\Omega_k$. We assume vertices of \Eve and \Adam in
$\overline\calG$ are $V_\shortEve$ and $V_\shortAdam$ respectively,
and that the initial vertex is $v_0$. The idea is to construct a new
game $\overline\calG'$ that remembers the objectives that have been
visited so far.  The vertices of game $\overline\calG'$ controlled by
\Eve and \Adam are $V'_\shortEve = V_\shortEve \times 2^{\lsem 1 ,
  k\rsem}$ and $V'_\shortAdam = V_\shortAdam \times 2^{\lsem
  1,k\rsem}$ respectively.  There is a transition from~$(v,S)$ to
$(v',S')$ iff there is a transition from $v$ to $v'$ in the original
game and $S' = S\cup \{i \mid v'\in \Omega_i\}$.  The reachability
objective $\Omega$ for \Eve is given by target set $\Stat \times \lsem
1,k\rsem$.  It is clear that there is a winning strategy in
$\overline\calG$ from $v_0$ for the conjunction of reachability
objectives $\Omega_1,\dots,\Omega_k$ iff there is a winning strategy
in game $\overline\calG'$ from $(v_0, \{i \mid v_0 \in \Omega_i\})$
for the reachability objective $\Omega$.  The number of vertices of
this new game is $|V'_\shortEve \cup V'_\shortAdam| = |V_\shortEve
\cup V_\shortAdam| \cdot 2^k$, and the size of the new transition
table $\Tab'$ is bounded by $|\Tab| \cdot 2^k$, where $\Tab$ is the
transition table of $\overline\calG$.  An attractor computation on
$\overline\calG'$ is then done in time $\mathcal{O}(|V'_\shortEve \cup
V'_\shortAdam| \cdot |\Tab'|)$, we obtain an algorithm for solving
zero-sum games with a conjunction of reachability objectives, running
in time $\mathcal{O}(2^{2k}\cdot (|V_\shortEve \cup V_\shortAdam|
\cdot |\Tab|))$.

\subsubsection{Algorithm}
The algorithm for solving the constrained NE existence problem for single
reachability objectives could be copied and would then be correct. It~would
however not yield an \NP upper bound. We therefore propose a refined
algorithm:
\begin{enumerate}[label=(\roman*)]
\item guess a lasso-shaped play $\rho = \tau_1 \cdot \tau_2^\omega$
  (with $|\tau_i| \le |\Stat|^2$) in $\calJ(\calG)$ such that \Adam
  obeys \Eve along $\rho$, and $\pi = \Sproj_1(\rho)$ satisfies the
  constraint on the payoff.
  Note that if $\Losers(\pi)$ is the set of players
  losing in~$\pi$, computing $W(\calG,\Losers(\pi))$ would require
  exponential time. We will avoid this expensive computation.
\item check that any \Adam-deviation along~$\rho$, say at position~$i$
  (for~any~$i$), leads to a state from which \Eve has a
  strategy~$\sigma^i_\shortEve$ to ensure that any play in $\rho_{\le
    i}\cdot \Out(\sigma^i_\shortEve)$ is winning for her.
\end{enumerate}
Step~$(ii)$ can be done as follows: pick an \Adam-state $(s,\Agt,m_\Agt)$
along $\rho$ and a successor $(t,P)$ such that $t \ne \Tab(s,m_\Agt)$; we only
need to show that $(t,P) \in W(\calG,(\Losers(\pi) \setminus \Losers(\rho_{\le
  i})) \cap P)$. We~can compute this set efficiently (in polynomial time)
using the algorithm of the previous paragraph since $2^{|P|} \le |\Tab|$
(using the same argument as in Proposition~\ref{lem:polynomial-size}).

This non-deterministic algorithm, which runs in polynomial time,
precisely implements Theorem~\ref{thm:eq-win}, and therefore correctly
decides the constrained NE existence problem.

\subsubsection{Hardness}
The \NP-hardness for the constrained NE existence problem can be proven
by encoding an instance of \SAT using a game similar to that for
reachability objectives, see Section~\ref{subsec:reachability}. We
only change the constraint which is now that all players~$C_i$ should
be losing, and we get the same equivalence.

The reduction of Lemma~\ref{sec:link-constr-exist} cannot be used to
deduce the hardness of the NE existence problem, since it assumes a lower
bound on the payoff. Here the constraint is an upper bound (``each
player should be losing''). We therefore provide an ad-hoc reduction
in this special case, which is illustrated on
Figure~\ref{fig-safety-hardnes}. We add some module at the end of the
game to enforce that in an equilibrium, all players are losing.  We
add concurrent states between $A$ and each~$C_i$ (named~$A/C_i$).  All players~$C_i$
are trying to avoid~$t$, and $A$ is trying to avoid~$u$.

Since $A$ has no target in $\calG_\phi$ she cannot lose before seeing $u$, and
then she can always change her strategy in the concurrent states in order to
go to $t$. Therefore an equilibrium always ends in $t$. A~player~$C_i$ whose
target was not seen during game~$\calG_\phi$, can change her strategy in order
to go to~$u$ instead of~$t$. That means that if there is an equilibrium, there
was one in $\calG_\phi$ where all $C_i$ are losing. Conversely, if there was
such an equilibrium in~$\calG_\phi$, we can extend this strategy profile by
one whose outcome goes to $t$ and it is an equilibrium in the new game. This
concludes the \NP-hardness of the NE existence problem.

\begin{figure*}[t]
  \centering
  \begin{tikzpicture}[thick]
    \draw[dashed, rounded corners=2mm] (-.3,0) .. controls +(90:5mm)
    .. (1,.5) -- (2,.9) -- (3,.6) -- (3.5,0) -- (3,-.8) -- (2,-.9)
    -- (1,-.7) .. controls +(150:5mm) and +(-90:5mm) .. (-.3,0);
    \draw (1.85,0) node {\begin{minipage}{2.2cm}\small\raggedright
        Copy of $\calG_\phi$
    \end{minipage}};
    \tikzstyle{rond}=[draw,circle,minimum size=6mm,inner sep=0mm]
    \tikzstyle{oval}=[draw,minimum height=6mm,inner sep=0mm,rounded corners=2mm]
    \draw (0.2,0) node [rond] (I) {$s$};
    \draw (3,0) node [rond] (A) { };
    \draw (4.3,0) node [oval] (B) {$A/C_1$};
    \draw (6.5,0) node [oval] (C) {$A/C_2$};
    \draw (8.7,0) node [oval] (D) {$A/C_3$};
    \draw (11,0) node [rond] (T) {$t$};

    \draw (7,-1.7) node [rond] (U) {$u$};

    \draw[-latex'] (A) -- (B);
    \draw[-latex'] (B) -- (C) node[midway,above] {$\scriptstyle \tuple{1,1}, \tuple{2,2}$};
    \draw[-latex'] (B) -- (U) node[midway,left] {$\scriptstyle \tuple{1,2}, \tuple{2,1}$};
    \draw[-latex'] (C) -- (D) node[midway,above] {$\scriptstyle \tuple{1,1}, \tuple{2,2}$};
    \draw[-latex'] (C) -- (U) node[midway] {$\scriptstyle \tuple{1,2}, \tuple{2,1}$};
    \draw[-latex'] (D) -- (T) node[midway,above] {$\scriptstyle \tuple{1,1}, \tuple{2,2}$};
    \draw[-latex'] (D) -- (U) node[midway,right] {$\scriptstyle \tuple{1,2}, \tuple{2,1}$};

    \draw[-latex'] (T) .. controls +(30:10mm) and +(-30:10mm) .. (T);
    \draw[-latex'] (U) .. controls +(-50:10mm) and +(-120:10mm) .. (U);
    \draw[dashed] (I) -- +(40:5mm);
    \draw[dashed] (I) -- +(0:4mm);
    \draw[dashed] (I) -- +(-40:5mm);
    \draw[dashed] (A) -- +(140:5mm);
    \draw[dashed] (A) -- +(180:4mm);
    \draw[dashed] (A) -- +(220:5mm);

  \end{tikzpicture}
  \caption{Extending game~$\calG_\phi$ with final concurrent modules}
  \label{fig-safety-hardnes}
\end{figure*}
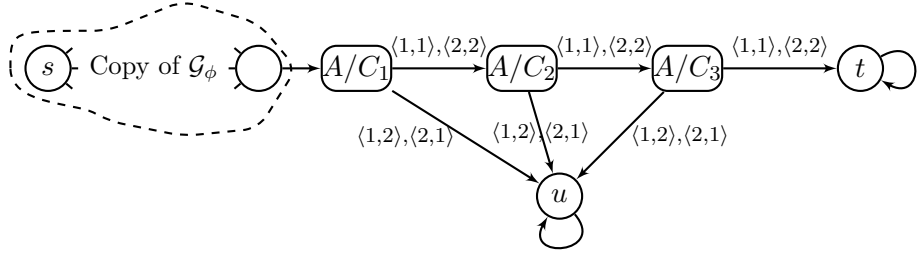

\subsection{B\"uchi objectives}
\label{subsec:buchi}
The value problem for B\"uchi objectives is \P-complete. In this
subsection we design a polynomial-time algorithm for solving the
constrained NE existence problem for B\"uchi objectives. The \P-hardness
of the NE existence problem can then be inferred from the \P-hardness of
the value problem, applying Propositions~\ref{lem:link-value-constr}
and~\ref{lem:link-value-exist}. Globally we prove the following
result:

\begin{theorem}
  For finite games with single B\"uchi objectives, the NE existence problem and
  the constrained NE existence problem are \P-complete.
\end{theorem}

\subsubsection{Reduction to a co-B\"uchi game}
We assume that for every player $A$, $\Omega_A$ is a B\"uchi objective
given by target set $T_A$.  Given $L \subseteq \Agt$, in the suspect
game~$\calH(\calG,L)$, we show that the objective of \Eve is
equivalent to a single co-B\"uchi objective.  We define the co-B\"uchi
objective~$\Omega_L$ in $\calH(\calG,L)$ given by the target set $T_L
= \{(s,P) \mid \exists A\in P \cap L.\ s \in T_A\}$. Notice that the
target set is defined in the same way as for reachability objectives.

\begin{lemma}\label{lem:red-co-buchi}
  A play $\rho$ is winning for \Eve in $\calH(\calG,L)$ iff $\rho \in
  \Omega_L$.
\end{lemma}
\begin{proof}
  Assume that $\rho$ is winning for \Eve in $\calH(\calG,L)$. Then for every $A
  \in \limitpi2(\rho) \cap L$, it~holds ${\Inf(\Sproj_1(\rho)) \cap T_A =
    \varnothing}$.  Toward a contradiction, assume that ${\Inf(\rho)
    \cap T_L \ne \varnothing}$. There exists $(s,P)$ such that there
  is $A \in P \cap L$ with $s \in T_A$, which appears infinitely often
  along $\rho$. In particular, $P = \limitpi2(\rho)$ (otherwise it
  would not appear infinitely often along $\rho$). Hence, we have
  found $A \in \limitpi2(\rho) \cap L$ such that $\Inf(\Sproj_1(\rho))
  \cap T_A \ne \varnothing$, which is a contradiction. Therefore,
  $\rho \in \Omega_L$.

  Assume $\rho \in \Omega_L$: for every $(s,P)$ such that there exists
  $A \in P \cap L$ with $s \in T_A$, $(s,P)$ appears finitely often
  along $\rho$. Let $A \in \limitpi2(\rho) \cap L$, and assume towards
  a contradiction that there is $s \in T_A$ such that $s$ appears
  infinitely often along $\Sproj_1(\rho)$. This means that
  $(s,\limitpi2(\rho))$ appears infinitely often along $\rho$, which
  contradicts the above condition. Therefore, $\rho$ is winning for
  \Eve in $\calH(\calG,L)$.
\end{proof}

\subsubsection{Algorithm}
\label{subsubsec:algo2}
As for reachability objectives, the winning region for \Eve in
$\calH(\calG,L)$ can be computed in polynomial time (since this is the
winning region of a co-B\"uchi game, see Lemma~\ref{lem:red-co-buchi}
above). A non-deterministic algorithm running in polynomial time
similar to the one for reachability objectives can therefore be
inferred. However we can do better than guessing an appropriate
lasso-shaped play $\rho$ by looking at the strongly connected
components of the game: a strongly connected component of the game
uniquely defines a payoff, which is that of all plays that visit
infinitely often all the states of that strongly connected
component. Using a clever partitioning of the set of strongly
connected components of the game, we obtain a polynomial-time
algorithm.

\medskip From now on and until the end of
Subsection~\ref{subsubsec:algo2} we relax the hypotheses on the
preference relations (that they are all single-objective with a
B\"uchi condition). We present an algorithm in a more general context,
since the same techniques will be used in
Subsection~\ref{subsec:co-reducible} (and we chose to only present
once the construction).
For the rest of this subsection we therefore make the following
assumptions on the preference relations $(\mathord\prefrel_A)_{A \in
  \Agt}$. For every player $A \in \Agt$: \label{hyp:star}
\begin{enumerate}[label=(\alph*)]
\item $\prefrel_A$ only depends on the set of states which is visited
  infinitely often: if $\rho$ and $\rho'$ are two plays such that
  $\Inf(\rho) = \Inf(\rho')$ then $\rho \prefrel_A \rho'$ and $\rho'
  \prefrel_A \rho$;
\item $\prefrel_A$ is given by an ordered objective $\omega_A$ with
  preorder $\preorder_A$, and $\preorder_A$ is supposed to be
  monotonic;
\item for every threshold $w^A$, we can compute in polynomial time
  $S^A \subseteq \Stat$ such that $\Inf(\rho) \subseteq S^A
  \Leftrightarrow \rho \prefrel_A w^A$.
\end{enumerate}
Obviously preferences given by single B\"uchi objectives do satisfy
those hypotheses. At every place where it is relevant, we will explain
how the particular case of single B\"uchi objectives is handled.
Next we write $(\star)$ for the above assumptions, and $(\star)_{a}$
(resp. $(\star)_b$, $(\star)_c$) for only the first (resp. second,
third) assumption.

\medskip We first characterise the `good' plays in $\calJ(\calG)$ in
terms of the strongly connected components they define: the strongly
connected component defined by a play is the set of states that are
visited infinitely often by the play.
We fix for each player $A$, equivalence classes of plays $u^A$ and
$w^A$, that represent lower- and upper-bounds for the constrained
NE existence problem. Both can be represented as finite sets,
representing the set of states which are visited infinitely often.
For each $K \subseteq \Stat$, we write $v^A(K)$ for the equivalence
class of all paths~$\pi$ that visits infinitely often exactly~$K$,
\ie: $\Inf(\pi) = K$.  We also write $v(K)=(v^A(K))_{A \in \Agt}$.  We
look for a transition system~$\langle K,E\rangle$, with $K\subseteq
\Stat$ and $E \subseteq K \times K$, for which the following
properties hold:
\begin{enumerate}
\item $u^A \preorder_A v^A(K) \preorder_A w^A$ for
  all~$A\in\Agt$; \label{cond:1}
\item $\langle K,E\rangle$ is strongly connected;\label{cond:2}
\item $\forall k\in K.\ (k,\Agt) \in W(\calG,v(K))$;\label{cond:3} 
\item $\forall (k,k') \in E.\ \exists (k,\Agt,m_\Agt)\in
  W(\calG,v(K)).\ \Tab(k,m_\Agt) = k'$;\label{cond:4}
\item $(K\times\{\Agt\})$ is reachable from $(s,\Agt)$ in
  $W(\calG,v(K))$;\label{cond:5}
\end{enumerate}
where $W(\calG,v(K))$ is the winning region of \Eve in suspect
game~$\calH(\calG,v(K))$.\footnote{Formally the suspect game has been
  defined with a play as reference, and not a equivalence
  class. However, in this subsection, if $\pi$ and $\pi'$ are
  equivalent, the games $\calH(\calG,\pi)$ and $\calH(\calG,\pi')$ are
  identical.}

If one can find one such transition system $\langle K,E \rangle$, then
we will be able to build a lasso-play $\rho$ from $(s,\Agt)$ in the
suspect-game that will satisfy the conditions of
Theorem~\ref{thm:eq-win}. Formally, we have the following lemma:

\begin{lemma}\label{lem:characterization}
  Under hypothesis $(\star)_{a}$, there is a transition system
  $\langle K,E\rangle$ satisfying
  conditions~\ref{cond:1}--\ref{cond:5} if, and only if, there is a
  path $\rho$ from $(s, \Agt)$ in $\calH(\calG,v(K))$ that never gets
  out of $W(\calG,v(K))$, along which \Adam always obeys \Eve, $u^A
  \preorder_A v^A(K) \preorder_A w^A$ for all~$A\in\Agt$, and
  $\Sproj_1(\Inf(\rho) \cap V_\shortEve)=K$ (which implies that
  $\rho\in v^A(K)$ for all $A$).
\end{lemma}

\begin{proof}
  The first implication is shown by building a path in $W(\calG,v(K))$
  that successively visits all the states in~$K\times\{\Agt\}$
  forever. Thanks to~\ref{cond:5}, \ref{cond:2} and~\ref{cond:4}
  (and the fact that \Adam obeys~\Eve), such a path exists, and
  from~\ref{cond:3} and~\ref{cond:4}, this path remains in the
  winning region.  From~\ref{cond:1}, we~have the condition on the
  preferences.
  Conversely, consider such a path~$\rho$, and let $K =
  \Sproj_1(\Inf(\rho)\cap V_\shortEve)$ and $E = \{ (k,k')\in K^2 \mid
  \exists (k,\Agt,m_\Agt) \in \Inf(\rho).\ \Tab(k,m_\Agt) =
  k'\}$. Condition~\ref{cond:5} clearly holds.
  Conditions~\ref{cond:1}, \ref{cond:3} and~\ref{cond:4} are
  easy consequences of the hypotheses and construction.  We~prove that
  $\langle K, E \rangle$ is strongly connected. First, since \Adam
  obeys~\Eve and $\rho$ starts in~$(k,\Agt)$, we~have
  $\limitpi2(\rho)=\Agt$. Now, take any two states~$k$ and~$k'$ in~$K$: then
  $\rho$ visits~$(k,\Agt)$ and~$(k',\Agt)$ infinitely often, and there
  is a subpath of~$\rho$ between those two states, all of which states
  appear infinitely often along~$\rho$. Such a subpath gives rise to a
  path between~$k$ and~$k'$, as required.
\end{proof}

As a consequence, if $\langle K,E \rangle$ satisfies the five previous
conditions, by~Theorem~\ref{thm:eq-win}, there is a Nash equilibrium
whose outcome lies between the bounds~$u^A$ and~$w^A$.  Our aim is to
compute efficiently all maximal pairs $\langle K, E \rangle$ that
satisfy the five conditions.

To that aim we define a recursive function \SSG (standing for ``solve
sub-game''), working on transition systems, that will decompose
efficiently any transition system that does not satisfy the five
conditions above into polynomially many disjoint sub-transition
systems \textit{via} a decomposition into strongly connected
components.
\begin{itemize}
\item if $K\times\{\Agt\} \subseteq W(\calG,v(K))$, and if for all
  $(k,k') \in E$ there is a $(k,\Agt,m_\Agt)$ in $W(\calG,v(K))$
  s.t. $\Tab(k,m_\Agt) = k'$, and finally if $\langle K,E\rangle$ is
  strongly connected, then we set $\SSG(\langle K,E\rangle) =
  \{\langle K,E\rangle\}$. This means that conditions (2)-(4) are
  satisfied by $\langle K,E\rangle$.
\item otherwise, we let
  \[ 
  \SSG(\langle K,E\rangle) = \bigcup_{\langle K',E'\rangle \in
    \SCC(\langle K,E\rangle)} \SSG ( T (\langle K',E'\rangle))
  \]
  where $\SCC(\langle K,E \rangle)$ is the set of strongly connected
  components of~$\langle K,E\rangle$ (which can be computed in linear
  time), and where $T(\langle K',E'\rangle)$ is the transition system
  whose set of states is $\{ k \in K' \mid (k,\Agt) \in
  W(\calG,v(K'))\}$ and whose set of edges is
  \[
  \{(k,k') \in E'\mid \exists (k,\Agt,m_\Agt) \in W(\calG,v(K')).\
  \Tab(k,m_\Agt) = k' \}.
  \]
  Notice that this set of edges is never empty, but $T(\langle
  K',E'\rangle)$ might not be strongly connected anymore, so that this
  is really a recursive definition.
\end{itemize}
The recursive function $\SSG$ decomposes any (sub-)transition system
of the game into a list of disjoint transition systems which all
satisfy conditions (2)-(4) above.

So far the computation does not take into account the bounds for the
payoffs of the players (lower bound $u^A$ and upper bound $w^A$ for
player $A$). For each upper bound $w^A$, we assume condition
$(\star)_c$ holds . In the particular case of a single B\"uchi
objective for each player define by target $T_A$, this is simply done
by setting $S^A = \Stat \setminus T_A$, if this player has to be
losing (that is, if $w^A$ does not satisfy the B\"uchi objective).
Now assuming we have found the appropriate set $S^A$, we define
\[ 
\Sol = \SSG\bigl(\langle \bigcap_{A\in\Agt} S^A , \Edg' \rangle \bigr)
\cap \bigl\{ \langle K,E\rangle \mid \forall A\in \Agt.\ u^A \preorder
v^A(K) \bigr\}
\]
where $\Edg'$ restricts $\Edg$ to $\bigcap_{A\in\Agt} S^A$.

We now show that the set $\Sol$ computes (in a sense that we make
clear) the transition systems that are mentioned in
Lemma~\ref{lem:characterization}).

\begin{lemma}\label{lem:character2}
  We suppose condition $(\star)$ holds.
  If $\langle K,E\rangle \in \Sol$ then it satisfies
  conditions~\ref{cond:1} to~\ref{cond:4}. Conversely, if $\langle
  K,E\rangle$ satisfies conditions~\ref{cond:1} to~\ref{cond:4}, then
  there exists $\langle K',E'\rangle \in \Sol$ such that $\langle K,E
  \rangle \subseteq \langle K',E'\rangle$.
\end{lemma}
\begin{proof}
  Let $\langle K,E\rangle \in \Sol$.  By~definition of~\SSG, all
  $(k,\Agt)$ for $k \in K$ are in~$W(\calG,v(K))$, and for all
  $(k,k')\in E$, there is a state~$(k,\Agt,m_\Agt)$ in $W(\calG,v(K))$
  such that $\Tab(k,m_\Agt) = k'$, and $\langle K,E\rangle$ is
  strongly connected.  Also, for all $A$, $u^A \preorder v^A(K)$
  because $\Sol \subseteq \{ \langle K,E\rangle \mid u^A \preorder
  v^A(K) \}$.  Finally, for any $A \in \Agt$, $v^A(K) \preorder w^A$
  because the set $K$ is included in $S^A$.

  Conversely, assume that $\langle K,E \rangle$ satisfies the
  conditions.  We~show that if $\langle K,E\rangle \subseteq \langle
  K',E'\rangle$ then there is $\langle K'',E''\rangle$ in
  $\SSG(\langle K',E'\rangle)$ such that $\langle K,E\rangle \subseteq
  \langle K'',E''\rangle$.  The proof is by induction on the size of
  $\langle K',E' \rangle $.

  The basic case is when $\langle K',E'\rangle$ satisfies
  the conditions~\ref{cond:2}, \ref{cond:3}, and~\ref{cond:4}: in that case, 
  $\SSG(\langle K' ,E'\rangle) = \{\langle K',E'\rangle\}$, and by
  letting $\langle K'',E'' \rangle= \langle K',E' \rangle$ we get the
  expected result.

  We now analyze the other case.  There is a strongly connected
  component of $\langle K',E'\rangle$, say $\langle K'',E''\rangle$,
  which contains $\langle K,E\rangle$, because $\langle K,E\rangle$
  satisfies condition~\ref{cond:2}.  We have $v^A(K) \preorder_A
  v^A(K'')$ (because $K\subseteq K''$ and $\preorder_A$ is monotonic)
  for every~$A$, and thus $W(\calG,v(K))\subseteq W(\calG,v(K''))$.
  This ensures that $T(\langle K'',E''\rangle)$ contains $\langle K,E
  \rangle$ as a subgraph.  Since $\langle K'',E'' \rangle$ is a
  subgraph of $\langle K',E' \rangle$, the graph $T(\langle K'',E''
  \rangle)$ also~is.  We show that they are not equal, so~that we can
  apply the induction hypothesis to~$T(\langle K'',E''\rangle)$.  For
  this, we~exploit the fact that $\langle K',E'\rangle$ does not
  satisfy one of conditions~\ref{cond:2} to~\ref{cond:4}:
 \begin{itemize}
 \item first, if $\langle K',E'\rangle$ is not strongly connected while
   $\langle K'',E''\rangle$~is, they cannot be equal;
 \item if there is some $k \in K'$ such that $(k,\Agt)$ is not in
   $W(\calG,v(K'))$, then $k$ is not a vertex of $T(\langle K'' , E''
   \rangle)$;
 \item if there some edge $(k,k')$ in $E'$ such that there is no state
   $(k,\Agt,m_\Agt)$ in $W(\calG,v(K'))$ such that $\Tab(k,m_\Agt) =
   k'$, then the edge $(k,k')$ is not in $T(\langle K'' , E''
   \rangle)$.
 \end{itemize}
  We then apply the induction hypothesis to $T(\langle
  K'',E''\rangle)$, and get the expected result. 
  Now, because of condition~\ref{cond:1}, 
  $u^A \preorder v^A(K) \preorder w^A$.
  Hence, due to the previous analysis, there exists $\langle
  K',E'\rangle \in \SSG\left(\langle \bigcap_{A \in \Agt} S^A
    , \Edg'  \rangle \right)$
    such that $\langle
  K,E\rangle \subseteq \langle K',E'\rangle$.  This concludes the
  proof of the lemma.
\end{proof}

\begin{lemma}\label{lem:compute-sol}
  Under assumptions $(\star)$, if for every $K$, the set
  $W(\calG,v(K))$ can be computed in polynomial time, then the set
  $\Sol$ can also be computed in polynomial time.
\end{lemma}

\begin{proof}
  Each recursive call to \SSG applies to a decomposition in strongly
  connected components of the current transition system under
  consideration.  Hence the number of recursive calls is bounded by
  $|\Stat|^2$.  Computing the decomposition in SCCs can be done in
  linear time.  By assumption, each set 
  $W(\calG,v(K))$ can be computed in polynomial time.  $S^A$ is
  obtained by removing the target of the losers (for $w^A$) from
  $\Stat$.  Hence globally we can compute $\Sol$ in polynomial time.
\end{proof}

To conclude the algorithm, we need to check that
condition~\ref{cond:5} holds for one of the solutions $\langle K, E
\rangle$ in $\Sol$. It can be done in polynomial time by looking for a
path in the winning region of \Eve in $\calH(\calG,v(K))$ that reaches
$K \times \{\Agt\}$ from $(s,\Agt)$. The correctness of the algorithm
is ensured by the fact that if some $\langle K, E \rangle$ satisfies
the five conditions, there is a $\langle K',E' \rangle$ in $\Sol$ with
$K \subseteq K'$ and $E \subseteq E'$. Since $K \subseteq K'$ implies
$v^A(K) \preorder_A v^A(K')$, the winning region of \Eve in
$\calH(\calG,v(K'))$ is larger than that $\calH(\calG,v(K'))$, which
implies that the path from $(s,\Agt)$ to $K \times \{\Agt\}$ is also a
path from $(s,\Agt)$ to $K' \times \{\Agt\}$. Hence, $\langle K', E'
\rangle$ also satisfies condition~\ref{cond:5}, and therefore the five
expected conditions.

\bigskip We have already mentioned that single B\"uchi objectives do
satisfy the hypotheses~$(\star)$. Furthermore,
Lemma~\ref{lem:red-co-buchi} shows that, given $v(K)$, one can compute
the set~$W(\calG,v(K))$ as the winning region of a co-B\"uchi
turn-based game, which can be done in polynomial time (this is argued
at the beginning of the section). Therefore
Lemma~\ref{lem:compute-sol} and the subsequent analysis apply: this
concludes the proof that the constrained NE existence problem for finite
games with single B\"uchi objectives is in \P.

\subsubsection{Hardness}
\label{ptime-hard}
We recall a possible proof of \P-hardness for the value problem, from
which we will infer the other lower bounds.  The circuit-value problem
can be easily encoded into a deterministic turn-based game with
B\"uchi objectives: a circuit (which we assume w.l.o.g. has only \AND-
and \OR-gates) is transformed into a two-player turn-based game, where
one player controls the \AND-gates and the other player controls the
\OR-gates. We add self-loops on the leaves.  Positive leaves of the
circuit are the (B\"uchi) objective of the \OR-player, and negative
leaves are the (B\"uchi) objective of the \AND-player. Then obviously,
the circuit evaluates to true iff the \OR-player has a winning
strategy for satisfying his B\"uchi condition, which in turn is
equivalent to the fact that there is an equilibrium with payoff~$0$
for the \AND-player, by Proposition~\ref{lem:link-value-constr}.  We
obtain \P-hardness for the NE existence problem, using
Proposition~\ref{lem:link-value-exist}: the preference relations in
the game constructed in Proposition~\ref{lem:link-value-exist} are
B\"uchi objectives.

\subsection{Co-B\"uchi objectives}\label{subsec:cobuchi}

The value problem for co-B\"uchi objectives is \P-complete. We now prove
that the constrained NE existence problem is in \NP, and that the
constrained NE existence problem and the 
NE existence problem are \NP-hard.
We therefore deduce:

\begin{theorem}
  For finite games with single co-B\"uchi objectives, the NE existence problem 
  and the constrained NE existence problem are \NP-complete.
\end{theorem}

The proof of this Theorem is very similar to that for safety
objectives: instead of conjunction of reachability objectives, we
need to deal with conjunction of B\"uchi objectives. Of course
constructions and algorithms need to be adapted. That is what we
present now.

\subsubsection{Reduction to a conjunction of B\"uchi conditions}
We assume that for every player $A$, $\Omega_A$ is a single co-B\"uchi
objective~$\Omega_A$ given by~$T_A$.  In the corresponding suspect
game, we show that the goal of player \Eve is equivalent to a
conjunction of B\"uchi objectives.  Let $L \subseteq \Agt$. In suspect
game $\calH(\calG,L)$, we define several B\"uchi objectives as
follows: for each $A \in L$, we define $T'_A = T_A \times \{P \mid P
\subseteq \Agt\} \cup \Stat \times \{ P \mid A \not\in P\}$, and we
write $\Omega'_A$ for the corresponding B\"uchi objective.
\begin{lemma}
  A play $\rho$ is winning for \Eve in $\calH(\calG,L)$ iff $\rho \in
  \bigcap_{A \in L} \Omega'_A$.
\end{lemma}

\begin{proof}
  Let $\rho$ be a play in $\calH(\calG,L)$, and assume it is winning
  for \Eve.  Then, for each $A \in \limitpi2(\rho)\cap L$, $\rho
  \notin \Omega_A$, which means that the target set $T_A$ is visited
  along $\Sproj_1(\rho)$, and therefore $T'_A$ is visited infinitely
  often along $\rho$.  If $A \notin \limitpi2(\rho)$, then a state
  $(s,P)$ with $A \notin P$ is visited infinitely often by $\rho$: the
  target set $T'_A$ is visited infinitely often. This implies that
  $\rho \in \bigcap_{A \in L} \Omega'_A$.

  Conversely let $\rho \in \bigcap_{A \in L} \Omega'_A$. For every $A
  \in L$, $T'_A$ is visited infinitely often by $\rho$. Then, either
  $T_A$ is visited infinitely often by $\Sproj_1(\rho)$ (which means
  that $\rho \notin \Omega_A$) or $A\not\in \limitpi2(\rho)$.  In
  particular, $\rho$ is a winning play for \Eve in $\calH(\calG,L)$.
\end{proof}

\subsubsection{Algorithm for solving zero-sum games with a conjunction of
  B\"uchi objectives}
We adapt the algorithm for conjunctions of reachability
objectives (page~\pageref{conj-reach}) to conjunctions of B\"uchi
objectives.  Let $\calG$ be a two-player turn-based game with a
winning objective for \Eve given as a conjunction of B\"uchi
objectives $\Omega_1,\dots,\Omega_k$.  The idea is to construct a new
game $\calG'$ which checks that each objective~$\Omega_i$ is visited
infinitely often.  The vertices of~$\calG'$ controlled by \Eve and
\Adam are $V'_\shortEve = V_\shortEve \times \lsem 0 , k\rsem$ and
$V'_\shortAdam = V_\shortAdam \times \lsem 0,k \rsem$ respectively.
There is a transition from $(v,k)$ to $(v',0)$ iff there is a
transition from $v$ to $v'$ in the original game and for $0 \le i <
k$, there is a transition from~$(v,i)$ to $(v',i+1)$ iff there is a
transition from $v$ to $v'$ in the original game and $v'\in
\Omega_{i+1}$. In~$\calG'$, the objective for \Eve is the B\"uchi
objective~$\Omega$ given by target set $\Stat \times \{k\}$, where
$\Stat=V_\shortEve \cup V_\shortAdam$ is the set of vertices of
$\calG$. It is clear that there is a winning strategy in $\calG$ from
$v_0$ for the conjunction of B\"uchi objectives
$\Omega_1,\dots,\Omega_k$ iff there is a winning strategy in $\calG'$
from $(v_0, 0)$ for the B\"uchi objective~$\Omega$.  The number of
states of game $\calG'$ is $|\Stat'|= |\Stat| \cdot k$, and the size
of the transition table $|\Tab'|=|\Tab| \cdot k$.  Using the standard
algorithm for turn-based B\"uchi objectives~\cite{CHP08}, which works in time
$\mathcal{O}(|\Stat'| \cdot |\Tab'|)$, we~obtain an algorithm for
solving zero-sum games with a conjunction of B\"uchi objectives
running in time $\mathcal{O}(k^2\cdot |\Stat|\cdot |\Tab|)$ (hence in
polynomial time).

\subsubsection{Algorithm}
The algorithm is the same as for reachability objectives. Only the
computation of the set of winning states in the suspect game is
different.  Since we just showed that this part can be done in
polynomial time, the global algorithm still runs in
(non-deterministic) polynomial time.

\subsubsection{Hardness}
The hardness result for the constrained NE existence problem with
co-B\"uchi objectives was already proven in~\cite{ummels08}. The idea
is to encode an instance of \SAT into a game with co-B\"uchi
objectives.  For completeness we describe the reduction below, and
explain how it can be modified for proving \NP-hardness of the
NE existence problem.

Let us consider an instance $\phi = c_1 \land \dots \land c_n$ of
\SSAT, where $c_i = \ell_{i,1} \lor \ell_{i,2} \lor
\ell_{i,3}$, and $\ell_{i,j} \in \{ x_k, \lnot x_k \mid 1 \le k \le
p\}$. 
The game~$\calG$ is obtained from module~$M(\phi)$ depicted on
Figure~\ref{fig-M}, by joining the outgoing edge of~$c_{n+1}$ to~$c_1$. Each
module~$M(\phi)$ involves a set of players~$B_{k}$, one for each
variable~$x_k$, and a player $A_1$. Player $A_1$ controls the clause states.
Player~$B_{k}$ control the literal states~$\ell_{i,j}$ when $\ell_{i,j} = \neg
x_k$, then having the opportunity to go to state~$\bot$. There is no
transition to~$\bot$ for literals of the form~$x_k$. In~$M(\phi)$, assuming
that the players~$B_k$ will not play to~$\bot$, then $A_1$ has a strategy that
does not visit both~$x_k$ and~$\neg x_k$ for every $k$ if, and only~if,
formula~$\phi$ is satisfiable.
Finally, the co-B\"uchi objective of~$B_k$ is given by
$\{x_k\}$. In~other terms, the aim of $B_k$ is to visit~$x_k$ only a
finite number of times. This way, in a Nash equilibrium, it~cannot be
the case that both $x_k$ and~$\neg x_k$ are visited infinitely often:
it~would imply that $B_k$ loses but could improve her payoff by going
to~$\bot$ (actually, $\neg x_k$ should not be visited at all if~$x_k$
is visited infinitely often).  Therefore setting the objective of
$A_1$ to $\{\bot\}$, there is a Nash equilibrium where she wins iff
$\phi$ is satisfiable.  This shows \NP-hardness for the constrained
NE existence problem.

For the NE existence problem, we use the transformation described in
Section~\ref{sec:link-constr-exist}. We add an extra player~$A_2$ to~$\calG$
and consider the game $\calG' = E(\calG,A_1,A_2,\rho)$, where $\rho$ is a
winning path for~$A_1$. The objective of the players in~$\calG'$ can be
described by co-B\"uchi objectives: $A_2$~has to avoid seeing $T = \{ s_1\}$
infinitely often and keep the same target for~$A_1$. Applying
Proposition~\ref{lem:link-constr-exist}, there is a Nash equilibrium
in~$\calG'$ if, and only~if, there is one in~$\calG$ where $A_1$~wins, this
shows \NP-hardness for the NE existence problem.
\begin{figure}[t]
  \centering
  \begin{tikzpicture}[thick]
    \tikzstyle{rond}=[draw,circle,minimum size=6mm,inner sep=0mm]
    \tikzstyle{oval}=[draw,minimum height=6mm,inner sep=0mm,rounded corners=2mm]
    \draw (0,0) node [rond] (C1) {$c_1$};
    \draw (3,0) node [rond] (C2) {$c_2$};
    \draw (7,0) node [rond] (CN) {$c_{n+1}$};
    \draw (C1.90) node [above] {$A_1$};
    \draw (C2.90) node [above] {$A_1$};

    \draw (1.5,1) node [rond] (L11) {$\ell_{1,1}$};
    \draw (1.5,0) node [rond] (L12) {$\ell_{1,2}$};
    \draw (1.5,-1) node [rond] (L13) {$\ell_{1,3}$};
    \draw (3.5,2) node [rond] (B) {$\bot$};

    \draw[-latex'] (C1) -- (L11);
    \draw[-latex'] (C1) -- (L12);
    \draw[-latex'] (C1) -- (L13);
    \draw[-latex'] (L11) -- (C2);
    \draw[-latex'] (L12) -- (C2);
    \draw[-latex'] (L13) -- (C2);

    \draw[-latex'] (C2) -- +(1,0.5);
    \draw[-latex'] (C2) -- +(1,0);
    \draw[-latex'] (C2) -- +(1,-0.5);
    \draw (5,0) node {\dots} ;

    \draw[-latex'] (CN)+(-1,-0.5) -- (CN);
    \draw[-latex'] (CN)+(-1,0) -- (CN);
    \draw[-latex'] (CN)+(-1,0.5) -- (CN);

    \draw[-latex'] (CN) -- + (1,0); 
    \draw[-latex'] (-1,0) -- (C1) ;
    \draw[-latex',dotted] (L11) -- (B);
    \draw[-latex',dotted] (L12) -- (B);
    \draw[-latex',dotted] (L13) -- (B);
    \draw[-latex'] (B) .. controls +(.5,1) and +(-.5,1) .. (B);
  \end{tikzpicture}
  \caption{Module $M(\phi)$, where $\phi = c_1 \land \dots \land c_n$  
    and $c_i = \ell_{i,1} \lor \ell_{i,2} \lor \ell_{i,3}$}\label{fig-M}
  \end{figure}

\subsection{Objectives given as circuits}\label{subsec:circuits}

The value problem is known to be \PSPACE-complete for turn-based games
and objectives given as circuits~\cite{Hunter07}. The transformation
presented in the beginning of the section can be used to decide the
value problem for finite concurrent games with a single
circuit-objective, yielding \PSPACE-completeness of the value problem
in the case of finite concurrent games as well.

We now show that the (constrained) NE existence problem is also
\PSPACE-complete in this framework:
\begin{theorem}
  For finite games with single objectives given as circuits, the NE existence
  problem and the constrained NE existence problem are \PSPACE-complete.
\end{theorem}

\subsubsection{Reduction to a circuit objective}
We assume the preference relation of each player $A \in \Agt$ is given
by a circuit~$C_A$.  Let $L \subseteq \Agt$.  We  define a Boolean
circuit defining the winning condition of \Eve in the suspect game
$\calH(\calG,L)$.

We define for each player $A \in \Agt$ and each set $P$ of players
(such that $\Stat \times P$ is reachable in $\calH(\calG,L)$), a
circuit $D_{A,P}$ which outputs \true for the plays $\rho$ with
$\limitpi2(\rho)=P$ (\ie whose states that are visited infinitely
often are in $\Stat \times \{P\}$), and whose value by $C_A$ is \true.
We do so by making a copy of the circuit $C_A$, adding $|\Stat|$ \OR
gates $g_1 \cdots g_{|\Stat|}$ and one \AND gate $h$. There is an edge
from $(s_i,P)$ to $g_i$ and from $g_{i-1}$ to $g_i$ if $i<|\Stat|$
then there is an edge from the output gate of $C_A$ to $h$ and from
$h$ to the output gate of the new circuit. Inputs of $C_A$ are now the
$(s,P)$'s (instead of the $s$'s).  The circuit $D_{A,P}$ is given on
Figure~\ref{fig:DAP}.

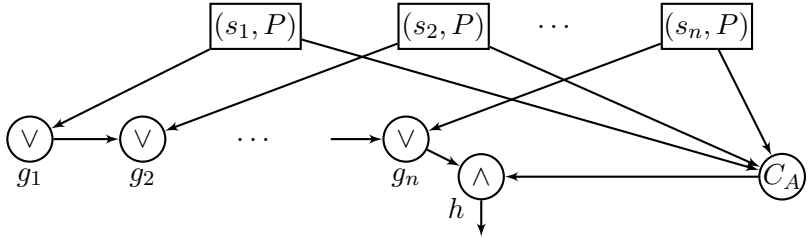
\begin{figure}[!ht]
  \centering
  \begin{tikzpicture}[thick]
    \tikzstyle{rond}=[draw,circle,minimum size=6mm,inner sep=0mm]
    \tikzstyle{oval}=[draw,minimum height=6mm,inner sep=0mm,rounded corners=2mm]
    \tikzstyle{carre}=[draw,minimum height=6mm,minimum width=12mm,inner sep=0mm]
    \draw (3,0.5) node [carre] (I1) {$(s_1,P)$};
    \draw (5.5,0.5) node [carre] (I2) {$(s_2,P)$};
    \draw (9,0.5) node [carre] (IN) {$(s_{n},P)$};
    \draw (7,0.5) node (ID) {$\dots$};
    \draw (10,-1.5) node[rond] (C) {$C_A$};
    \draw (0,-1) node [rond] (G1) {$\lor$} node [below=7pt] {$g_1$};
    \draw (1.5,-1) node [rond] (G2) {$\lor$} node [below=7pt] {$g_2$};
    \draw (3,-1) node {$\dots$};
    \draw (5,-1) node [rond] (GN) {$\lor$} node [below=7pt] {$g_n$};
    \draw (6,-1.5) node [rond] (H) {$\land$} node [label=-120:$h$] {};

    \draw[-latex'] (I1) -- (G1);
    \draw[-latex'] (I2) -- (G2);
    \draw[-latex'] (IN) -- (GN);
    \draw[-latex'] (G1) -- (G2);
    \draw[-latex'] (4,-1) -- (GN);
    \draw[-latex'] (GN) -- (H);
    \draw[-latex'] (C) -- (H);
    \draw[-latex'] (H) -- +(0,-.8);

    \draw[-latex'] (I1) -- (C);
    \draw[-latex'] (I2) -- (C);
    \draw[-latex'] (IN) -- (C);

  \end{tikzpicture}
  \caption{Circuit $D_{A,P}$}
  \label{fig:DAP}
\end{figure}

We then define a circuit $E_A$ which outputs \true for the plays
$\rho$ with $A \in \limitpi2(\rho)$ and whose output by $C_A$ is
\true.  We do so by taking the disjunction of the circuits $D_{A,P}$.
Formally, for each set of players $P$ such that $\Stat\times P$ is
reachable in the suspect game and $A\in P$, we include the circuit
$D_{A,P}$ and writing $o_{A,P}$ for its output gate, we add \OR gates
so that there is an edge from $o_{A,P}$ to $g_i$ and from $g_i$ to
$g_{i+1}$, and then from $g_{n+1}$ to the output gate.

Finally we define the circuit $F_L$, which outputs \true for the plays
$\rho$ such that there is no $A\in L$ such that $A\in \limitpi2(\rho)$
and the output of $\Sproj_1(\rho)$ by $C_A$ is \true.  This corresponds
exactly to the plays that are winning for \Eve in suspect game
$\calH(\calG,L)$.  We do so by negating the disjunction of all the
circuits $E_A$ for $A\in L$.

The next lemma follows from the construction:

\begin{lemma}
  \label{lem:circuit}
  A play $\rho$ is winning for \Eve in $\calH(\calG,L)$ iff $\rho$
  evaluates circuit $F_L$ to \true.
\end{lemma}

We should notice that circuit $F_L$ has size polynomial in the size of
$\calG$, thanks to Proposition~\ref{lem:polynomial-size}.

\subsubsection{Algorithm and complexity analysis}
To solve the constrained NE existence problem we apply the same algorithm
as for reachability objectives (see
section~\ref{subsec:reachability}). For complexity matters, the only
difference stands in the computation of the set of winning states in
the suspect game. Thanks to Lemma~\ref{lem:circuit}, we know it
reduces to the computation of the set of winning states in a
turn-based game with an objective given as a circuit (of
polynomial-size). This can be done in \PSPACE~\cite{Hunter07}, which
yields a \PSPACE upper bound for the constrained NE existence problem
(and therefore for the NE existence problem and the value problem~--~see
Proposition~\ref{lem:link-value-constr}).  \PSPACE-hardness of all
problems follows from that of the value problem in turn-based
games~\cite{Hunter07}, and from Propositions~\ref{lem:link-value-constr}
and~\ref{lem:link-value-exist} (we notice that the preference
relations in the new games are easily definable by circuits).

\subsection{Rabin and parity objectives}
\label{subsec:rabin}

The value problem is known to be \NP-complete for Rabin
conditions~\cite{emerson1988complexity} and in \UP $\cap$ \co-\UP\ for
parity conditions~\cite{Jurdzinski98}.

We then notice that a parity condition is a Rabin condition with half
as many pairs as the number of priorities: assume the parity condition
is given by $p \colon \Stat \mapsto \lsem 0 , d \rsem$ with $d\in \N$;
take for $i$ in $\lsem 0, \frac{d}{2}\rsem$, $Q_i = p^{-1}\{2 i\}$ and
$R_i = p^{-1} \{ 2j+1 \mid j \ge i\}$. Then the Rabin objective
$(Q_i,R_i)_{0 \le i \le \frac{d}{2}}$ is equivalent to the parity
condition given by $p$.

We design an algorithm that solves the constrained NE existence
problem in $\P^\NP_\parallel$ for Rabin objectives (see
footnote~\ref{fn-pnp||} on page~\pageref{fn-pnp||} for an informal definition
of $\P^{\NP}_\parallel$). 

Our algorithm heavily uses non-determinism (via the oracle). We
then propose a deterministic algorithm which runs in exponential
time, but will be useful in Section~\ref{subsec:rabin-auto}. This
subsection ends with proving $\P^\NP_\parallel$-hardness of the
constrained NE existence problem and NE existence problem for parity
objectives. 
In the end, we will have proven the following theorem:

\begin{theorem}
  For finite games with single objectives given as Rabin or parity conditions,
  the NE existence problem and the constrained NE existence problem are
  $\P^\NP_\parallel$-complete.
\end{theorem}

\subsubsection{Reduction to a Streett game}
We assume that the preference relation of each player $A \in \Agt$ is
given by the Rabin condition $(Q_{i,A},R_{i,A})_{i\in\lsem
  1,k_A\rsem}$. Let $L \subseteq \Agt$. In the suspect game
$\calH(\calG,L)$, we define the Streett objective
$(Q'_{i,A},R'_{i,A})_{i\in\lsem 1,k_A\rsem, A\in L}$, where $Q'_{i,A}
= (Q_{i,A} \times \{P \mid A \in P\}) \cup (\Stat \times \{ P \mid A
\not\in P\})$ and $R'_{i,A} = R_{i,A} \times \{P \mid A \in P\}$, and
we write $\Omega_L$ for the corresponding set of winning plays.

\begin{lemma}
  \label{lemma:red-streett}
  A play $\rho$ is winning for \Eve in $\calH(\calG,L)$ iff $\rho \in
  \Omega_L$.
\end{lemma}

\begin{proof}
  Assume $\rho$ is winning for \Eve in $\calH(\calG,L)$.  For all $A
  \in \limitpi2(\rho) \cap L$, $\Sproj_1(\rho)$ does not satisfy the Rabin
  condition given by $(Q_{i,A},R_{i,A})_{i\in\lsem 1,k_A\rsem}$.  For
  all $1 \le i \le k_A$, $\Inf(\Sproj_1(\rho))\cap Q_{i,A} = \varnothing$
  or $\Inf(\Sproj_1(\rho)) \cap R_{i,A}\ne \varnothing$.  We infer that
  for all $1 \le i \le k_A$, $\Inf(\rho) \cap Q'_{i,A} = \varnothing$
  or $\Inf(\rho) \cap R'_{i,A} \ne \varnothing$.  Now, if $A \notin
  \limitpi2(\rho)$ then all $Q'_{i,A}$ are seen infinitely often along
  $\rho$. Therefore for every $A \in L$, the Streett
  conditions~$(Q'_{i,A},R'_{i,A})$ is satisfied along $\rho$ (that is,
  $\rho \in \Omega_L$).

  Conversely, if the Streett condition $(Q'_{i,A},R'_{i,A})_{i\in\lsem
    1,k_A\rsem, A\in L}$ is satisfied along $\rho$, then either the
  Rabin condition $(Q_{i,A}, R_{i,A})$ is not satisfied along
  $\Sproj_1(\rho)$ or $A\not\in \limitpi2(\rho)$. This means that \Eve is winning
  in $\calH(\calG,L)$.
\end{proof}

\subsubsection{Algorithm}
We now describe a $\P^\NP_\parallel$ algorithm for solving the
constrained NE existence problem in games where each player has a single
Rabin objective. As in the previous cases, our algorithm relies on the
suspect game construction.

Write $\calP$ for the set of sets of players of~$\Agt$ that appear as
the second item of a state of~$\calJ(\calG)$:
\[
\calP = \{P\subseteq \Agt \mid \exists s\in\Stat.\ (s,P)\text{ is a
  state of }\calJ(\calG)\}. 
\]
Since~$\calJ(\calG)$ has size polynomial, so has~$\calP$.  Also, for
any path~$\rho$, $\limitpi2(\rho)$ is a set of~$\calP$. Hence, for a
fixed~$L$, the~number of sets~$\limitpi2(\rho)\cap L$ is
polynomial. Now, as recalled on page~\pageref{simplification}, the
winning condition for~\Eve is that the players in~$\limitpi2(\rho)\cap
L$ must be losing along $\Sproj_1(\rho)$ in $\calG$ for their Rabin
objective. We have seen that this can be seen as a Streett objective
(Lemma~\ref{lemma:red-streett}).

Now, deciding whether a state is winning in a turn-based game for a
Streett condition can be decided in
\coNP~\cite{emerson1988complexity}.  Hence, given a state~$s\in\Stat$
and a set~$L$, we can decide in \coNP~whether $s$ is winning for \Eve
in~$\calH(\calG,L)$. This will be used as an oracle in our algorithm
below.

\smallskip Now, pick a set~$P\subseteq\Agt$ of suspects, \ie, for
which there exists $(s,t)\in\Stat^2$ and $m_\Agt$
s.t. $P=\Susp((s,t),m_\Agt)$. Using the same arguments as in the proof
of Proposition~\ref{lem:polynomial-size}, it~can be shown that $2^{\size P}
\leq \size\Tab$, so that the number of subsets of~$P$ is polynomial.
Now, for each set~$P$ of suspects and each~$L\subseteq P$, write
$w(L)$ for the size of the winning region of~\Eve
in~$\calH(\calG,L)$. Then the sum $\sum_{P\in\calP\setminus\{\Agt\}}
\sum_{L\subseteq P} w(L)$ is at most $\size\Stat\times \size\Tab^2$.

Assume that the exact value~$M$ of this sum is known, and consider the
following algorithm: 
\begin{enumerate}
\item for each $P\subseteq \calP\setminus\{\Agt\}$ and
  each~$L\subseteq P$, guess a set $W(L)\subseteq \Stat$, which we
  intend to be the exact winning region for~\Eve in~$\calH(\calG,L)$.
\item check that the sizes of those sets sum up to~$M$;
\item for each $s\notin W(L)$, check that \Eve does not have a winning
  strategy from~$s$ in~$\calH(\calG,L)$. This can be checked in
  non-deterministic polynomial time, as explained above.
\item guess a lasso-shaped path~$\rho=\pi\cdot \tau^\omega$ in~$\calH(\calG,L)$
  starting from~$(s,\Agt)$, with $\size \pi$ and $\size \tau$ less
  than $\size\Stat^2$ (following Proposition~\ref{lem:play-length})
  visiting only states where the second item is~$\Agt$. This path can
  be seen as the outcome of some strategy of~\Eve when \Adam
  obeys. For this path, we~then check the following:
  \begin{itemize}
  \item along~$\rho$, the sets of winning and losing players satisfy
    the original constraint (remember that we aim at solving the
    constrained NE existence problem);
  \item any deviation along~$\rho$ leads to a state that is winning
    for~\Eve.  In~other terms, pick a state~$h=(s,\Agt,m_\Agt)$
    of~\Adam along~$\rho$, and pick a successor~$h'=(t,P)$ of~$h$ such
    that $t\not=\Tab(s,m_\Agt)$. Then the algorithm checks that $t\in
    W(L\cap P)$.
  \end{itemize}
\end{enumerate}
The algorithm accepts the input~$M$ if it succeeds in finding the
sets~$W$ and the path~$\rho$ such that all the checks are
successful. This algorithm is non-deterministic and runs in polynomial
time, and will be used as a second oracle.

\medskip We now show that if $M$ is exactly the sum of the $w(L)$,
then the algorithm accepts~$M$ if, and only~if, there is a Nash
equilibrium satisfying the constraint, \ie, if, and only~if, \Eve has
a winning strategy from~$(s,\Agt)$ in~$\calH(\calG,L)$.

First assume that the algorithm accepts~$M$. This means that it~is
able, for each~$L$, to find sets~$W(L)$ of states whose complement
does not intersect the winning region of~$\calH(\calG,L)$. Since~$M$
is assumed to be the exact sum of~$w(L)$ and the size of the
sets~$W(L)$ sum up to~$M$, we~deduce that~$W(L)$ is exactly the
winning region of~\Eve in~$\calH(\calG,L)$. Now, since the algorithm
accepts, it~is also able to find a (lasso-shaped) path~$\rho$ only
visiting states having~$\Agt$ as the second component. This path has
the additional property that any ``deviation'' from a state of \Adam
along this path ends up in a state that is winning for~\Eve for
players in~$L\cap P$, where $P$ is the set of suspects for the present
deviation. This way, if during~$\rho$, \Adam deviates to a
state~$(t,P)$, then \Eve~will have a strategy to ensure that along any
subsequent play, the objectives of players in~$L\cap P$ (in~$\calG$)
are not fulfilled, so that along any run~$\rho'$, the players
in~$L\cap \limitpi2(\rho')$ are losing for their objectives
in~$\calG$, so that \Eve wins in~$\calH(\calG,L)$.

Conversely, assume that there is a Nash equilibrium satisfying the
constraint.  Following Proposition~\ref{lem:play-length}, we assume
that the outcome of the corresponding strategy profile has the form
$\pi\cdot\tau^\omega$.  From Lemma~\ref{lemma-suspectgame}, there is a
winning strategy for~\Eve in~$\calH(\calG, L)$ whose outcome when
\Adam obeys follows the outcome of the Nash equilibrium. As~a
consequence, the outcome when \Adam obeys is a path~$\rho$ that the
algorithm can guess. Indeed, it~must satisfy the constraints, and any
deviation from~$\rho$ with set of suspects~$P$ ends in a state where
\Eve wins for the winning condition of~$\calH(\calG,L)$, hence also
for the winning condition of~$\calH(\calG,L\cap P)$, since any
path~$\rho'$ visiting~$(t,P)$ has $\limitpi2(\rho')\subseteq P$.

\medskip Finally, our global algorithm is as follows: we run the first
oracle for all the states and all the sets~$L$ that are subsets of a
set of suspects (we know that there are polynomially many such
inputs). We~also run the second algorithm on all the possible values
for~$M$, which are also polynomially many. Now, from the answers of
the first oracle, we~compute the exact value~$M$, and return the value
given by the second on that input. This algorithm runs
in~$\P^{\NP}_{\parallel}$ and decides the constrained NE existence
problem.

\subsubsection{Deterministic algorithm}
\label{rabin:det-algo}
In the next section we will need a deterministic algorithm to solve
games with objectives given as deterministic Rabin automata. We
therefore present it right now. The deterministic algorithm works by
successively trying all the possible payoffs, there are $2^{|\Agt|}$
of them.  Then it computes the winning strategies of the suspect game
for that payoff.  In \cite{horn2005streett} an algorithm for Streett
games is given, which works in time $\mathcal{O}(n^k\cdot k!)$, where
$n$ is the number of vertices in the game, and $k$ the size of the
Streett condition.  The algorithm has to find, in the winning region
of \Eve in $\calJ(\calG)$, a lasso that satisfies the Rabin winning
conditions of the winners and do not satisfy whose of the losers.  To
do so it tries all the possible choices of elementary Rabin condition
that are satisfied to make the players win, there are at most
$\prod_{A \in \Agt} k_A$ possible choices.  And for the losers, we try
the possible choices for whether $Q_{i,A}$ is visited of not, there
are $\prod_{A\in\Agt} 2^{k_A}$ such choices.  It then looks for a
lasso cycle that, when $A$ is a winner, does not visit $Q_{i_A,A}$ and
visits $R_{i_A,A}$, and when $A$ is a loser, visits $R_{i_A,A}$ when
it has to, or does not visit $Q_{i_A,A}$.  This is equivalent to
finding a path satisfying a conjunction of B\"uchi conditions and can
be done in polynomial time $\mathcal{O}(n\times \sum_{A\in \Agt}
k_A)$.  The global algorithm works in
time \[\mathcal{O}\left(2^{|\Agt|} \cdot \left(|\Tab|^{3 \sum_A k_A}
    \cdot (\sum_A k_A)! + \left(\prod_{A \in \Agt} k_A \cdot
      2^{k_A}\right) \cdot |\Tab|^3 \cdot \sum_A k_A\right)\right)\]
Notice that the exponential does not come from the size of the graph
but from the number of agents and the number of elementary Rabin
conditions, this will be important when in the next subsection we will
reuse the algorithm on a game structure whose size is exponential.

\subsubsection{$\P^{\NP}_{\parallel}$-hardness}
We now prove $\P^{\NP}_{\parallel}$-hardness of the (constrained)
NE existence problem in the case of parity objectives.  The main
reduction is an encoding of the \PARITYSAT problem, where the aim is
to decide whether the number of satisfiable instances among a set of
formulas is even. This problem is known to be complete for
$\P^{\NP}_{\parallel}$~\cite{Got95}.

Before tackling the whole reduction, we first develop some
preliminaries on single instances of \SSAT, inspired
from~\cite{CHP07}.  Let us consider an instance $\phi = c_1 \land
\dots \land c_n$ of \SSAT, where $c_i = \ell_{i,1} \lor \ell_{i,2}
\lor \ell_{i,3}$, and $\ell_{i,j} \in \{ x_k, \lnot x_k \mid 1 \le k
\le p\}$.  With~$\phi$, we associate a three-player game~$N(\phi)$,
depicted on Figure~\ref{fig-N1} (where the first state of~$N(\phi)$ is
controlled by~$A_1$, and the first state of each~$N'(c_j)$ is
concurrently controlled by~$A_2$ and~$A_3$).  For each variable~$x_j$,
players~$A_2$ and~$A_3$ have the following target sets:
\begin{xalignat*}4
  T^{A_2}_{2j} &= \{x_j\} &
  T^{A_2}_{2j +1} &=\{\lnot x_j\} &\qquad
  T^{A_3}_{2j +1} &= \{ x_j\} &
  T^{A_3}_{2j} &= \{\lnot x_j\}
\end{xalignat*}

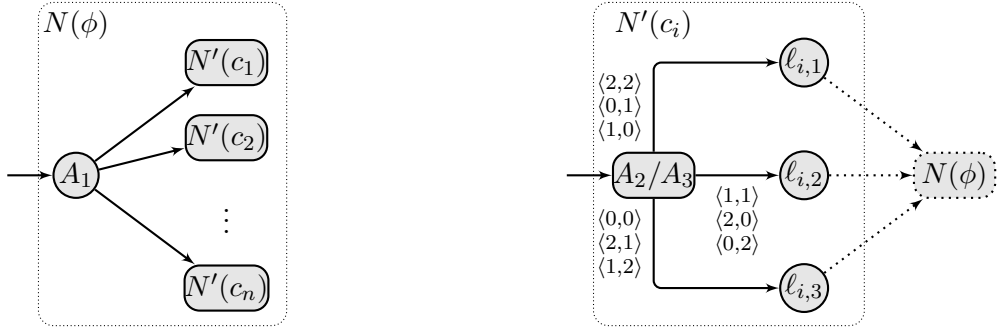
\begin{figure}[t]
  \begin{minipage}{0.45\textwidth}
    \centering
    \begin{tikzpicture}[thick]
    \tikzstyle{rond}=[draw,circle,minimum size=6mm,inner sep=0mm,fill=black!10]
    \tikzstyle{oval}=[draw,minimum height=6mm,inner sep=0mm,rounded corners=2mm,fill=black!10]
      \draw [thin,densely dotted,rounded corners=2mm] (-.5,2.3) --(2.8,2.3) -- (2.8,-2) -- (-.5,-2) --cycle;
      \path (0,2) node {$N(\phi)$};
      \draw (0,0) node[rond] (A) {$A_1$};
      \draw (2,1.5) node[oval] (P1) {$N'(c_1)$};
      \draw (2,.5) node[oval] (P2) {$N'(c_2)$};
      \draw (2,-.5) node {\vdots};
      \draw (2,-1.5) node[oval] (Pn) {$N'(c_n)$};
      \draw[latex'-] (A.-180) -- +(-.6,0);
      \draw[-latex'] (A) -- (P1); 
      \draw[-latex'] (A) -- (P2);
      \draw[-latex'] (A) -- (Pn);
    \end{tikzpicture}
  \end{minipage}
  \hfill
  \begin{minipage}{0.45\textwidth}
    \centering
    \begin{tikzpicture}[thick]
    \tikzstyle{rond}=[draw,circle,minimum size=6mm,inner sep=0mm,fill=black!10]
    \tikzstyle{oval}=[draw,minimum height=6mm,inner sep=0mm,rounded corners=2mm,fill=black!10]
      \draw [thin,densely dotted,rounded corners=2mm] (-.8,2.3) --(2.8,2.3) -- (2.8,-2) -- (-.8,-2) --cycle;
      \path (0,2) node {$N'(c_i)$};
      \draw (0,0) node[oval] (A) {$A_2/A_3$};
      \draw (2,1.5) node[rond] (P1) {$\ell_{i,1}$};
      \draw (2,0) node[rond] (P2) {$\ell_{i,2}$};
      \draw (2,-1.5) node[rond] (P3) {$\ell_{i,3}$};
      \draw (4,0) node[oval,dotted,inner sep=1mm] (MP) {$N(\phi)$};
      \draw[latex'-] (A.-180) -- +(-.6,0);
      \draw[-latex',rounded corners=1mm] (A) |- node[left,pos=.25]
           {\begin{minipage}{1cm}\baselineskip=7pt\flushright
               $\scriptstyle\langle 2,2\rangle$
               $\scriptstyle\langle 0,1\rangle$
               $\scriptstyle\langle 1,0\rangle$
           \end{minipage}}
           (P1);  
     \draw[-latex',rounded corners=1mm] (A) -- 
     node[below] {\begin{minipage}{1cm}\baselineskip=7pt\centering
         $\scriptstyle\langle 1,1\rangle$
         $\scriptstyle\langle 2,0\rangle$
         $\scriptstyle\langle 0,2\rangle$
     \end{minipage}}
     (P2); 
     \draw[-latex',rounded corners=1mm] (A) |- 
     node[left,pos=.25] {\begin{minipage}{1cm}\baselineskip=7pt\flushright
         $\scriptstyle\langle 0,0\rangle$
         $\scriptstyle\langle 2,1\rangle$
              $\scriptstyle\langle 1,2\rangle$
     \end{minipage}}
     (P3);
     \draw[-latex',dotted] (P1) -- (MP); 
     \draw[-latex',dotted] (P2) -- (MP); 
     \draw[-latex',dotted] (P3) -- (MP); 
    \end{tikzpicture}
  \end{minipage}
  \caption{The game~$N(\phi)$ (left), where $N'(c_i)$ is the module on the right.}\label{fig-N1}
\end{figure}

\noindent This construction enjoys interesting properties, given by the
following lemma:
\begin{lemma}\label{lemma-N}
  If the formula~$\phi$ is not satisfiable, then there is a strategy
  for player~$A_1$ in~$N(\phi)$ such that players~$A_2$ and~$A_3$
  lose.
  If the formula~$\phi$ is satisfiable, then for any strategy
  profile~$\sigma_\Agt$, one of~$A_2$ and~$A_3$ can change her
  strategy and win.
\end{lemma}
  
\begin{proof}
  We begin with the first statement, assuming that $\phi$~is not
  satisfiable and defining the strategy for~$A_1$.  With a history~$h$
  in~$N(\phi)$, we~associate a valuation $v^h \colon \{x_k \mid k \in
  [1,p]\} \to \{\top,\bot\}$ (where $p$ is the number of distinct
  variables in~$\phi$), defined as follows:
  \[
  v^h(x_k) = \top \ \Leftrightarrow\ \exists m.\ h_m = x_k \land
  \forall m'>m .\ h_{m'} \ne \lnot x_k \qquad \text{for all
    $k\in[1,p]$}
  \]
  We also define $v^h(\neg x_k) = \neg v^h(x_k)$.  Under this
  definition, $v^h(x_k)=\top$ if the last occurrence of~$x_k$ or~$\neg
  x_k$ along~$h$ was~$x_k$.  We~then define a strategy~$\sigma_1$ for
  player~$A_1$: after a history~$h$ ending in an $A_1$-state,
  we~require~$\sigma_1(h)$ to go to $N'(c_i)$ for some $c_i$ (with
  least index,~say) that evaluates to false under~$v^h$ (such a~$c_i$
  exists since $\phi$ is not satisfiable). This strategy enforces that
  if $h\cdot\sigma_1(h)\cdot \ell_{i,j}$ is a finite outcome
  of~$\sigma_1$, then $v^h(\ell_{i,j}) = \bot$, because $A_1$ has
  selected a clause~$c_i$ whose literals all evaluate to~$\bot$.
  Moreover, $v^{h\cdot\sigma_1(h)\cdot \ell_{i,j}}(\ell_{i,j}) =
  \top$, so that for each~$j$, any outcome of~$\sigma_1$ will either
  alternate between~$x_k$ and~$\neg x_k$ (hence visit both of them
  infinitely often), or no longer visit any of them after some
  point. Hence both~$A_2$ and~$A_3$ lose.

  \medskip We now prove the second statement.  Let $v$ be a valuation
  under which~$\phi$ evaluates to~true, and $\sigma_\Agt$ be a
  strategy profile.  From~$\sigma_{A_2}$ and~$\sigma_{A_3}$, we~define
  two strategies $\sigma_{A_2}'$ and~$\sigma_{A_3}'$.  Consider a
  finite history~$h$ ending in the first state of~$N'(c_i)$, for
  some~$i$. Pick a~literal~$\ell_{i,j}$ of~$c_i$ that is true
  under~$v$ (the one with least index,~say). We~set
  \begin{xalignat*}2
    \sigma'_{A_2}(h) &= [j - \sigma_{A_3}(h) \text{ (mod 3)}]
    & 
    \sigma'_{A_3}(h) &= [j - \sigma_{A_2}(h) \text{ (mod 3)}].
  \end{xalignat*}
  It~is easily checked that, when $\sigma_{A_2}$ and~$\sigma'_{A_3}$
  (or $\sigma'_{A_2}$ and $\sigma_{A_3}$) are played simultaneously in
  the first state of some~$N'(c_i)$, then the game goes
  to~$\ell_{i,j}$. Thus under those strategies, any visited literal
  evaluates to~true under~$v$, which means that at most one of~$x_k$
  and $\neg x_k$ is visited (infinitely often). Hence one of~$A_2$
  and~$A_3$ is winning, which proves our claim.
\end{proof}

\medskip We now proceed by encoding an instance
\begin{xalignat*}1
  \exists x^1_{1}, \dots x^1_{k}.\ &\phi^1(x^1_{1},\dots,x^1_{k}) \\
  &\dots \\
  \exists x^m_{1}, \dots x^m_{k}.\ &\phi^m(x^m_{1},\dots,x^m_{k})
\end{xalignat*}
of \PARITYSAT into a parity game.  The game involves the three
players~$A_1$, $A_2$ and~$A_3$ of the game~$N(\phi)$ defined above,
and it~will contain a copy of~$N(\phi^r)$ for each~$1\leq r\leq m$.
The~objectives of~$A_2$ and~$A_3$ are the unions of their objectives
in each~$N(\phi^r)$, e.g. $p^{A_2}(x^1_j) = p^{A_2}(x^2_j) = \cdots =
p^{A_m} (x^m_j) = 2j$.

For each such~$r$, the game will also contain a copy of the
game~$M(\phi^r)$ depicted on Figure~\ref{fig-M}.  Each
game~$M(\phi^r)$ involves an extra set of players~$B^r_{k}$, one for
each variable~$x^r_k$.  As we have seen in
Section~\ref{subsec:cobuchi}, in a Nash equilibrium, it~cannot be the
case that both $x^r_k$ and~$\neg x^r_k$ are visited infinitely often.

In order to test the parity of the number of satisfiable formulas, we
then define two families of modules, depicted on Figure~\ref{fig-mod1}
to~\ref{fig-modn}. Finally, the whole game~$\calG$ is depicted on
Figure~\ref{fig-g}.  In~that game, the objective of~$A_1$ is to visit
infinitely often the initial state~$\textrm{init}$.

\begin{figure}[!ht]
  \begin{minipage}{0.5\textwidth}
    \centering
    \begin{tikzpicture}[thick,xscale=1.2]
     \tikzstyle{rond}=[draw,circle,minimum size=7mm,inner sep=1mm,fill=black!10]
     \tikzstyle{oval}=[draw,minimum width=10mm,inner sep=1mm,minimum height=7mm,rounded corners=2mm,fill=black!10]
     \draw (0,0) node[rond] (A1) {$A_1$};
      \draw (1,1) node[oval] (MP) {$M(\phi^r)$};
      \draw (3,1) node[oval] (GP) {$G(\phi^{r-1})$};
      \draw (1,-1) node[oval] (A2A3) {$A_2/A_3$};
      \draw (3,-2) node[oval] (NP) {$N(\phi^r)$};
      \draw (3,0) node[oval] (HP) {$H(\phi^{r-1})$};

      \draw[-latex'] (-0.8,0) -- (A1);
      \draw[-latex'] (A1) -- (MP);
      \draw[-latex'] (A1) -- (A2A3);
      \draw[-latex'] (MP) -- (GP);
      \draw[-latex'] (A2A3) -- node[above,sloped] {$\scriptstyle\langle
        1,0\rangle $} node[below,sloped] {$\scriptstyle\langle 0,1 \rangle$}
      (NP);
      \draw[-latex'] (A2A3) -- node[below,sloped] {$\scriptstyle\langle
        1,1\rangle $} node[above,sloped] {$\scriptstyle\langle 0,0 \rangle$}
      (HP);
    \end{tikzpicture}
    \caption{Module $H(\phi^r)$ for $r\ge 2$}\label{fig-mod1}  
  \end{minipage}%
  \begin{minipage}{0.5\textwidth}
    \centering
    \begin{tikzpicture}[thick,xscale=1.2]
      \tikzstyle{rond}=[draw,circle,minimum size=7mm,inner sep=1mm,fill=black!10]
     \tikzstyle{oval}=[draw,minimum width=10mm,inner sep=1mm,minimum height=7mm,rounded corners=2mm,fill=black!10]
     \draw (0,0) node[rond] (A1) {$A_1$};
      \draw (1,1) node[oval] (MP) {$M(\phi^r)$};
      \draw (3,1) node[oval] (GP) {$H(\phi^{r-1})$};
      \draw (1,-1) node[oval] (A2A3) {$A_2/A_3$}; 
      \draw (3,-2) node[oval] (NP) {$N(\phi^r)$};
      \draw (3,0) node[oval] (HP) {$G(\phi^{r-1})$};

      \draw[-latex'] (-0.8,0) -- (A1);
      \draw[-latex'] (A1) -- (MP);
      \draw[-latex'] (A1) -- (A2A3);
      \draw[-latex'] (MP) -- (GP);
      \draw[-latex'] (A2A3) -- node[above,sloped] {$\scriptstyle\langle
        1,0\rangle $} node[below,sloped] {$\scriptstyle\langle 0,1 \rangle$}
      (NP); 
      \draw[-latex'] (A2A3) -- node[below,sloped] {$\scriptstyle\langle
        1,1\rangle $} node[above,sloped] {$\scriptstyle\langle 0,0 \rangle$}
      (HP); 
    \end{tikzpicture}
    \caption{Module $G(\phi^r)$ for $r\ge 2$}    
  \end{minipage}

\medskip
\def\figurename{Fig.}
\captionindent=0pt
\noindent\begin{minipage}{\linewidth}
\noindent
 \begin{minipage}{.3\linewidth}
    \centering
    \begin{tikzpicture}[thick,xscale=1.2]
     \tikzstyle{rond}=[draw,circle,minimum size=7mm,inner sep=1mm,fill=black!10]
     \tikzstyle{oval}=[draw,minimum width=10mm,inner sep=1mm,minimum height=7mm,rounded corners=2mm,fill=black!10]
    \path[use as bounding box] (-1,.9) -- (1,-1.7);
   \draw (0,0) node[oval] (M) {$M(\phi^1)$};
      \draw[-latex'] (-1,0) -- (M);
      \draw[-latex'] (M) -- (1,0);
    \end{tikzpicture}
    \caption{Module $H(\phi^1)$}
  \end{minipage}\hfill
  \begin{minipage}{0.32\linewidth}
    \centering
    \begin{tikzpicture}[thick,xscale=1.2]
     \tikzstyle{rond}=[draw,circle,minimum size=7mm,inner sep=1mm,fill=black!10]
     \tikzstyle{oval}=[draw,minimum width=10mm,inner sep=1mm,minimum height=7mm,rounded corners=2mm,fill=black!10]
       \path[use as bounding box] (-.5,.9) -- (1.5,-1.7);
      \draw (0,0) node[oval] (M) {$A_2/A_3$};
      \draw (0,-1.5) node[oval] (N) {$N(\phi^1)$};
      \draw[-latex'] (-1,0) -- (M);
      \draw[-latex'] (M) -- node[above] {$\scriptstyle\langle
        0,0\rangle $} node[below] {$\scriptstyle\langle 1,1 \rangle$}
      (1.5,0);
      \draw[-latex'] (M) -- node[left] {$\scriptstyle\langle
        1,0\rangle $} node[right] {$\scriptstyle\langle 0,1 \rangle$}
      (N);
    \end{tikzpicture}
    \caption{\mbox{Module $G(\phi^1)$}}\label{fig-modn}
  \end{minipage}\hfill
  \begin{minipage}{0.32\linewidth}
    \centering
    \begin{tikzpicture}[thick,xscale=1.2]
     \tikzstyle{rond}=[draw,circle,minimum size=7mm,inner sep=1mm,fill=black!10]
     \tikzstyle{oval}=[draw,minimum width=10mm,inner sep=1mm,minimum height=7mm,rounded corners=2mm,fill=black!10]
       \path[use as bounding box] (-1,.9) -- (2.5,-1.7);
      \draw (0,0) node[oval] (A1) {$\textrm{init}\vphantom{G^n}$};
      \draw (2,0) node[oval] (MP) {$G(\phi^m)$};
      \draw[-latex'] (-1,0) -- (A1);
      \draw[-latex'] (A1) -- (MP);
      \draw[-latex',rounded corners=3mm]  (MP) |-  (1,-1) -| (A1) ;
    \end{tikzpicture}
    \caption{The game $\calG$}\label{fig-g}
  \end{minipage}
\end{minipage}
\end{figure}

\begin{lemma}
  There is a Nash equilibrium in the game $\calG$ where $A_2$ and
  $A_3$ lose and $A_1$ wins if, and only if, the number of satisfiable
  formulas is even.
\end{lemma}

\begin{proof}
  Assume that there is a Nash equilibrium in~$\calG$ where $A_1$ wins
  and both~$A_2$ and~$A_3$ lose. Let~$\rho$ be its outcome. As~already
  noted, if~$\rho$~visits module~$M(\phi^r)$ infinitely often, then
  it~cannot be the case that both~$x^r_k$ and~$\neg x^r_k$ are visited
  infinitely often in~$M(\phi^r)$, as otherwise $B^r_k$ would be
  losing and have the opportunity to improve her payoff. This implies
  that $\phi^r$ is satisfiable.
  Similarly, if~$\rho$ visits infinitely often the states
  of~$H(\phi^r)$ or $G(\phi^r)$ that is controlled by~$A_2$ and~$A_3$,
  then it~must be the case that $\phi^r$ is not satisfiable, since
  from Lemma~\ref{lemma-N} this would imply that $A_2$ or~$A_3$ could
  deviate and improve her payoff by going to~$N(\phi^r)$.

  \medskip We now show by induction on~$r$ that if $\rho$ goes
  infinitely often in module $G(\phi^r)$ then $\#\{ j \le r \mid
  \phi^r\text{ is satisfiable}\}$ is even, and that (if~$n>1$) this
  number is odd if $\rho$ goes infinitely in module $H(\phi^r)$.

  When~$r=1$, since $H(\phi^1)$ is $M(\phi^1)$, $\phi^1$ is
  satisfiable, as noted above. Similarly, if $\rho$ visits $G(\phi^1)$
  infinitely often, it~also visits its $A_2/A_3$-state infinitely
  often, so that $\phi^1$ is not satisfiable. This proves the base
  case.

  Assume that the result holds up to some~$r-1$, and assume
  that~$\rho$ visits~$G(\phi^{r})$ infinitely often. Two cases may
  occur:
  \begin{itemize}
  \item it~can be the case that $M(\phi^r)$ is visited infinitely
    often, as well as $H(\phi^{r-1})$. Then $\phi^r$ is satisfiable,
    and the number of satisfiable formulas with index less than or
    equal to~$r-1$ is~odd.  Hence the number of satisfiable formulas
    with index less than or equal to~$r$ is even.
  \item it~can also be the case that the state $A_2/A_3$
    of~$G(\phi^r)$ is visited infinitely often. Then $\phi^r$ is not
    satisfiable. Moreover, since~$A_1$ wins, the play will also
    visit~$G(\phi^{r-1})$ infinitely often, so that the number of
    satisfiable formulas with index less than or equal to~$r$ is even.
  \end{itemize}
  If $\rho$ visits~$H(\phi^r)$ infinitely often, using similar
  arguments we prove that the number of satisfiable formulas with
  index less than or equal to~$r$ is odd.

  To conclude, since~$A_1$ wins, the play visits~$G(\phi^m)$
  infinitely often, so that the total number of satisfiable formulas
  is even.

  \medskip Conversely, assume that the number of satisfiable formulas
  is even. We~build a strategy profile, which we prove is a Nash
  equilibrium in which~$A_1$ wins, and~$A_2$ and~$A_3$
  lose. The~strategy for~$A_1$ in the initial states of~$H(\phi^r)$
  and~$G(\phi^r)$ is to go to~$M(\phi^r)$ when $\phi^r$ is
  satisfiable, and to state~$A_2/A_3$ otherwise. In~$M(\phi^r)$, the
  strategy is to play according to a valuation
  satisfying~$\phi^r$. In~$N(\phi^r)$, it~follows a strategy along
  which~$A_2$ and~$A_3$ lose (this exists according to
  Lemma~\ref{lemma-N}). This defines the strategy for~$A_1$. Then
  $A_2$ and~$A_3$ are required to always play the same move, so that
  the play never goes to some~$N(\phi^r)$. In~$N(\phi^r)$, they can
  play any strategy (they lose anyway, whatever they~do). Finally, the
  strategy of~$B^r_k$ never goes to~$\bot$.

  We now explain why this is the Nash equilibrium we are after. First,
  as $A_1$ plays according to fixed valuations for the
  variables~$x^r_k$, either $B^r_k$ wins or she does not have the
  opportunity to go to~$\bot$.  It~remains to prove that $A_1$ wins,
  and that $A_2$ and~$A_3$ lose and cannot improve
  (individually). To~see this, notice that between two consecutive
  visits to~$\textrm{init}$, exactly one of~$G(\phi^r)$
  and~$H(\phi^r)$ is visited. More precisely, it~can be observed that
  the strategy of~$A_1$ enforces that $G(\phi^r)$ is visited if
  $\#\{r<r' \leq m \mid \phi^{r'} \text{ is satisfiable}\}$ is even,
  and that $H(\phi^r)$ is visited otherwise. Then if~$H(\phi_1)$ is
  visited, the number of satisfiable formulas with index between~$2$
  and~$m$ is odd, so that $\phi_1$ is satisfiable and $A_1$ can return
  to~\textrm{init}. If~$G(\phi^1)$ is visited, an even number of
  formulas with index between~$2$ and~$m$ is satisfiable, and $\phi^1$
  is~not. Hence $A_1$ has a strategy in~$N(\phi^1)$ to make $A_2$
  and~$A_3$ lose, so that $A_2$ and $A_3$ cannot improve their
  payoffs.
\end{proof}

This proves hardness for the constrained NE existence problem with parity
objectives. For the NE existence problem, we use the construction of
Section~\ref{sec:link-constr-exist}, but since it can only be used to get rid
of constraint of the type ``$A_1$ is winning'', we add to the game two
players, $A_4$~and~$A_5$, whose objectives are opposite to~$A_2$ and~$A_3$
respectively, and one player~$A_6$ that will be playing matching-penny games.
The objectives for $A_4$ and~$A_5$ are definable by parity objectives, by
adding $1$ to all the priorities. Then, we consider game $\calG'=
E(E(E(\calG,A_1,A_6,\rho_1),A_4,A_6,\rho_4),A_5,A_6,\rho_5)$ where $\rho_1$,
$\rho_4$ and $\rho_5$ are winning paths for $A_1$, $A_4$ and~$A_5$
respectively. Thanks to Proposition~\ref{lem:link-constr-exist}, there is a
Nash equilibrium in~$\calG'$ if, and only~if, there is a Nash equilibrium
in~$\calG$ where $A_1$ wins and $A_2$ and $A_3$ lose. We~deduce
$\P^{\NP}_{\parallel}$-hardness for the NE existence problem with parity
objectives.

\subsection{Objectives given as deterministic Rabin automata}
\label{sec:rabin-auto}\label{subsec:rabin-auto}
In order to find Nash equilibria when objectives are given as
deterministic Rabin automata, we first define the notion of
\newdef{game simulation}, which we show has the property that when
$\calG'$~game-simulates~$\calG$, then a Nash equilibrium in the latter
game gives rise to a Nash equilibrium in the former one.

We then define the product of a game with automata (defining the
objectives of the players), and show that it game-simulates the
original game. This reduces the case of games with objectives are
defined as Rabin automata to games with Rabin objectives, which we
handled at the previous section; the resulting algorithm is in
\EXPTIME. We then show a \PSPACE lower bound for the problem in the
case of objectives given as deterministic B\"uchi automata. This
proves the following theorem:

\begin{theorem}
  For finite games with single objectives given as deterministic Rabin
  automata or deterministic B\"uchi automata, the NE existence problem and
  the constrained NE existence problem are in \EXPTIME and \PSPACE-hard.
\end{theorem}

It~must be noticed that game simulation can be used in other contexts:
in particular, in~\cite{BBM10a} (where we introduced this notion),
it~is shown that a region-based abstraction of timed games game
simulates its original timed game, which provides a way of computing
Nash equilibria in timed games.

\subsubsection{Game simulation}
\label{sec:simulation}

We~define game simulation and show how that can be used to compute
Nash equilibria.  We then apply it to objectives given as
deterministic Rabin automata.

\begin{definition}\label{def-gsim}
  Consider two games $\calG=\tuple{\Stat,\Agt,\Act, \Allow,\Tab,
    (\mathord\prefrel_A)_{A\in\Agt}}$ and
  $\calG'=\tuple{\Stat',\Agt,\Act',\Allow',\Tab',
    (\mathord\prefrel'_A)_{A\in\Agt}}$ with the same set~$\Agt$ of players.
  A~relation $\mathord{\simulrel} \subseteq \Stat \times \Stat'$ is a
  \emph{game simulation} if $s \simulrel s'$ implies that for each
  move $\indicebis {m}A\Agt$ in~$\calG$ there exists a move
  $\indicebis{m'}A\Agt$ in~$\calG'$ \st:
  \begin{enumerate}\raggedright
  \item\label{cond:sim22} $\Tab(s,\indicebis {m}A\Agt) \simulrel
    \Tab'(s',\indicebis{m'}A\Agt)$, and
  \item\label{cond:sim21} for each $t'\in\Stat'$ there exists
    $t\in\Stat$ with $t\simulrel t'$ and
    $\Susp((s',t'),\indicebis{m'}A\Agt) \subseteq
    \Susp((s,t),\indicebis {m}A\Agt)$.
  \end{enumerate}
  If $\simulrel$~is a game simulation and $(s_0,s'_0)\in {\simulrel}$,
  we~say that $\calG'$ \emph{game-simulates} (or simply
  \emph{simulates})~$\calG$.  When there are two paths $\rho$ and
  $\rho'$ such that $\rho_{=i}\simulrel\rho'_{=i}$ for all $i\in\N$,
  we will simply write $\rho \simulrel \rho'$.

  A~game simulation~$\simulrel$ is \emph{preference-preserving} from
  $(s_0,s'_0) \in \Stat \times \Stat'$ if for all $\rho_1 , \rho_2 \in
  \Play_{\calG}(s_0)$ and $\rho'_1 , \rho'_2 \in \Play_{\calG'}(s'_0)$
  with $\rho_1 \simulrel \rho'_1$ and $\rho_2 \simulrel \rho'_2$, for
  all $A\in\Agt$ it~holds that $\rho_1 \prefrel_A \rho_2$ iff
  $\rho'_1 \prefrel'_A \rho'_2$.
\end{definition}

As we show now, Nash equilibria are preserved by game simulation, in
the following sense:
\begin{proposition}
  \label{prop:sim}
  Let $\calG=\tuple{\Stat,\Agt,\Act, \Allow,\Tab,
    (\mathord\prefrel_A)_{A\in\Agt}}$ and
  $\calG'=\tuple{\Stat',\Agt,\Act',\Allow',\Tab',
    (\mathord\prefrel'_A)_{A\in\Agt}}$ be two games involving the same set of
  players. Fix two states $s_0$ and~$s_0'$ in $\calG$ and~$\calG'$
  respectively, and let $\simulrel$~be a preference-preserving game
  simulation from $(s_0,s_0')$.  If there exists a Nash equilibrium
  $\sigma_\Agt$ in~$\calG$ from~$s_0$, then there exists a Nash
  equilibrium $\sigma'_\Agt$ in~$\calG'$ from~$s_0'$ with
  $\Out_{\calG}(s_0,\sigma_\Agt) \simulrel
  \Out_{\calG'}(s'_0,\sigma'_\Agt)$.
\end{proposition}

\begin{proof}
  We fix a strategy profile~$\sigma_\Agt$ in $\calG$ and $\rho$ 
  the outcome of $\sigma_\Agt$ from $s_0$. We~derive a strategy
  profile~$\sigma'_\Agt$ in~$\calG'$ and its outcome~$\rho'$
  from $s'_0$, such that:
  \begin{enumerate}[label=(\alph*)]
  \item\label{eq-sima} for every $\overline\rho' \in
    \Play_{\calG'}(s'_0)$, there exists $\overline\rho \in
    \Play_{\calG}(s_0)$ s.t. $\overline\rho \simulrel \overline\rho'$
    and $\Susp(\overline\rho',\sigma'_\Agt) \subseteq
    \Susp(\overline\rho,\sigma_\Agt)$;
  \item\label{eq-simb} $\rho \simulrel \rho'$.
  \end{enumerate}

  Assume we have done the construction, and that $\sigma_\Agt$ is a
  Nash equilibrium in~$\calG$. We~prove that $\sigma'_\Agt$ is a Nash
  equilibrium in~$\calG'$.  Towards a contradiction, assume that some
  player~$A$ has a strategy $\overline\sigma'_A$ in $\calG'$ such that
  $\overline\rho' \not\prefrel_A \rho'$, where $\overline\rho' =
  \Out_{\calG'}(s', \replaceter{\sigma'}{A}{\overline\sigma'_A})$.
  Note that $A \in \Susp(\overline\rho',\sigma'_\Agt)$.
  Applying~\eqref{eq-sima} above, there exists $\overline\rho \in
  \Play_{\calG}(s_0)$ such that $\overline\rho \simulrel
  \overline\rho'$ and $\Susp(\overline\rho',\sigma'_\Agt) \subseteq
  \Susp(\overline\rho,\sigma_\Agt)$. In~particular, $A \in
  \Susp(\overline\rho,\sigma_\Agt)$, and there exists a strategy
  $\overline\sigma_A$ for~$A$ such that $\overline\rho =
  \Out_{\calG}(s_0,
  \replaceter{\sigma}{A}{\overline\sigma})$. As~${\rho \simulrel
    \rho'}$ (by~\eqref{eq-simb}) and $\simulrel$~is
  preference-preserving from~$(s_0,s'_0)$, $\overline\rho
  \not\prefrel_A \rho$, which contradicts the fact that
  $\sigma_\Agt$~is a Nash equilibrium. Hence, $\sigma'_\Agt$~is a Nash
  equilibrium in~$\calG'$ from~$s'_0$.

  \medskip It remains to show how we construct $\sigma'_\Agt$
  (and~$\rho'$). We~first build~$\rho'$ inductively, and define
  $\sigma'_\Agt$ along that~path.
  \begin{itemize}
  \item Initially, we let~$\rho'_{=0}=s'_0$. Since~$\simulrel$ is a
    game simulation containing~$(s_0,s'_0)$, we~have $s_0\simulrel
    s'_0$, and there is a move~$m'_\Agt$ associated with
    $\sigma_\Agt(s_0)$ satisfying the conditions of
    Definition~\ref{def-gsim}. Then $\rho_{=0}\simulrel \rho'_{=0}$,
    and $\Susp(\rho'_{=0},\sigma'_\Agt(\rho'_{=0}))\subseteq
    \Susp(\rho_{=0},\sigma_\Agt(\rho_{=0}))$.
  \item Assume we have built $\rho'_{\le i}$ and $\sigma'_\Agt$ on all
    the prefixes of~$\rho'_{\le i}$, and that they are such that
    $\rho_{\leq i}\simulrel \rho'_{\leq i}$ and $\Susp(\rho'_{\leq
      i},\sigma'_\Agt)\subseteq \Susp(\rho_{\leq i},\sigma_\Agt)$
    (notice that $\Susp(\rho'_{\leq i},\sigma'_\Agt)$ only depends on
    the value of~$\sigma'_\Agt$ on all the prefixes of~$\rho_{\leq
      i}$). In~particular, we~have $\rho_{=i}\simulrel \rho'_{=i}$, so
    that with the move~$\sigma_\Agt(\rho_{\leq i})$, we~can associate
    a move~$m'_{\Agt}$ (to~which we set~$\sigma'_{\Agt}(\rho'_{\leq
      i})$) satisfying both conditions of Definition~\ref{def-gsim}.
    This defines~$\rho'_{=i+1}$ in such a way that $\rho_{\leq i+1}
    \simulrel \rho'_{\leq i+1}$; moreover, $\Susp(\rho'_{\leq
      i+1},\sigma'_\Agt) = \Susp(\rho'_{\leq i},\sigma'_\Agt) \cap
    \Susp((\rho'_{=i}, \rho'_{=i+1}),m'_\Agt)$ is indeed a subset of
    $\Susp(\rho_{\leq i+1},\sigma_\Agt)$.
  \end{itemize}
  It~remains to define~$\sigma'_\Agt$ outside its
  outcome~$\rho'$. Notice that, for our purposes, it~suffices to
  define~$\sigma'_\Agt$ on histories starting from~$s'_0$. We~again
  proceed by induction on the length of the histories,
  defining~$\sigma'_\Agt$ in order to satisfy~\eqref{eq-sima} on
  prefixes of plays of~$\calG'$ from~$s'_0$.  At~each step, we also
  make sure that for every $h' \in \Hist_{\calG'}(s'_0)$, there exists
  $h \in \Hist_{\calG}(s)$ such that $h \simulrel h'$,
  $\Susp(h',\sigma'_\Agt) \subseteq \Susp(h,\sigma_\Agt)$, and
  $\sigma_\Agt(h)$ and $\sigma'_\Agt(h')$ satisfy the conditions of
  Definition~\ref{def-gsim} in the last states of~$h$ and~$h'$, resp.

  As we only consider histories from~$s'_0$, the case of histories of
  length~zero was already handled.  Assume we have defined
  $\sigma'_\Agt$ for histories $h'$ of length~$i$, and fix a new
  history $h'\cdot t' \in \Hist_{\calG'}(s'_0)$ of length~$i+1$ (that
  is not a prefix of~$\rho$).  By~induction hypothesis, there is $h
  \in \Hist_{\calG}(s_0)$ such that $h \simulrel h'$, and
  $\Susp(h',\sigma'_\Agt) \subseteq \Susp(h,\sigma_\Agt)$, and
  $\sigma_\Agt(h)$ and $\sigma_\Agt(h')$ satisfy the required
  properties.  In~particular, with~$t'$, we~can associate~$t$
  s.t. $t\simulrel t'$ and $\Susp((\last(h'),t'),\sigma'_\Agt(h'))
  \subseteq \Susp((\last(h),t),\sigma_\Agt(h))$. Then $(h \cdot t)
  \simulrel (h' \cdot t')$. Since~$t\simulrel t'$, there is a
  move~$m'_\Agt$ associated with $\sigma_\Agt(h\cdot t)$ and
  satisfying the conditions of
  Definition~\ref{def-gsim}. Letting~$\sigma'_\Agt(h'\cdot
  t')=m'_\Agt$, we~fulfill all the requirements of our induction
  hypothesis.
 
  We now need to lift the property from histories to infinite paths.
  Consider a play $\overline\rho'\in\Play_{\calG'}(s'_0)$, we will
  construct a corresponding play $\overline\rho$ in $\calG$.  Set
  $\overline\rho_{0} = s_0$. If $\overline\rho$ has been defined up to
  index $i$ and $\overline\rho_{i} \simulrel \overline\rho'_{i}$ (this
  is true for $i = 0$), thanks to the way $\sigma'_\Agt$ is
  constructed, $\sigma_\Agt(\overline\rho_{\le i})$ and
  $\sigma'_\Agt(\overline\rho'_{\le i})$ satisfy the conditions of
  Definition~\ref{def-gsim} in $\overline\rho_{\le i}$ and
  $\overline\rho'_{i}$, respectively.  We then pick
  $\overline\rho_{i+1}$ such that $\overline\rho_{i+1} \simulrel
  \overline\rho'_{i+1}$ and
  $\Susp((\overline\rho_{i},\overline\rho_{i+1}),
  \sigma_\Agt(\overline\rho_{i})) \subseteq
  \Susp((\overline\rho'_{i},\overline\rho'_{i+1}),
  \sigma'_\Agt(\overline\rho'_{i}))$.  This being true at each step,
  the path $\overline\rho$ that is obtained, is such that
  $\overline\rho \simulrel \overline\rho'$ and
  $\Susp(\overline\rho',\sigma'_\Agt) \subseteq
  \Susp(\overline\rho,\sigma_\Agt)$.  This is the desired property.
\end{proof}

\subsubsection{Product of a game with deterministic Rabin automata}
\label{sec:productB}
After this digression on game simulation, we come back to the game
$\calG = \tuple{\Stat, \Agt, \Act, \Allow, \Tab,
  (\mathord\prefrel_A)_{A\in\Agt}}$, where we assume that some player~$A$ has
her objective given by a deterministic Rabin automaton
$\calA=\tuple{Q,\Stat,\delta,q_0,(Q_{i},R_{i})_{i \in \lsem
    1,n\rsem}}$ (recall that this automaton reads sequences of states
of~$\calG$, and accepts the paths that are winning for player~$A$).
We show how to compute Nash equilibria in~$\calG$ by building a
product~$\calG'$ of~$\calG$ with the automaton~$\calA$ and by
computing the Nash equilibria in the resulting game, with a Rabin
winning condition for~$A$.

We define the product of the game~$\calG$ with the automaton~$\calA$
as the game $\calG \ltimes \calA = \tuple{\Stat',
  \Agt, \Act, \Allow', \Tab',(\mathord\prefrel'_A)_{A\in\Agt}}$,
where: 
\begin{itemize}
\item $\Stat' = \Stat \times Q$;
\item $\Allow'((s,q),\pl j) = \Allow(s,\pl j)$
  for every $\pl j \in \Agt$; 
\item $\Tab'((s,q), m_\Agt) = (s',q')$ where
  $\Tab(s,m_\Agt)= s'$ and $\delta (q,s) = q'$; 
\item If $B=A$ then $\prefrel'_B$ is given by the internal Rabin
  condition $Q_i' = \Stat \times Q_i$ and $R_i' = \Stat \times R'_i$.
  Otherwise $\prefrel'_B$ is derived from~$\prefrel_B$, defined by
  $\rho \prefrel'_B \overline\rho$ if, and only if, $\Rproj(\rho)
  \prefrel_B \Rproj(\overline\rho)$ (where $\Rproj$ is the projection of
  $\Stat'$ on~$\Stat$).  Notice that if $\prefrel_B$ is an internal
  Rabin condition, then so is $\prefrel'_B$.
\end{itemize}

\begin{lemma}
  \label{lemma:product1}
  $\calG \ltimes \calA$ game-simulates~$\calG$, with game simulation
  defined according to the projection: $s \simulrel (s',q)$ iff
  $s=s'$. This game simulation is preference-preserving.

  Conversely, $\calG$ game-simulates~$\calG \ltimes \calA$, with game
  simulation defined by $(s,q) \simulrel' s'$ iff $s=s'$, which is
  also preference-preserving.
\end{lemma}

\begin{proof}
  We~begin with proving that both relations are
  preference-preserving. First notice that if $((s_n,q_n))_{n\ge 0}$ is
  a play in ${\calG \ltimes \calA}$, then its $\Rproj$-projection
  $(s_n)_{n\ge 0}$ is a play in~$\calG$. Conversely, if $\rho=(s_n)_{n
    \ge 0}$ is a play in~$\calG$, then there is a unique
  path~$(q_n)_{n \ge 0}$ from initial state~$q_0$ in~$\calA$ which
  reads~it, and $((s_n,q_n))_{n\ge 0}$ is then a path in $\calG
  \ltimes \calA$ that we write $\Rproj^{-1}(\rho) = ((s_n,q_n))_{n\ge
    0}$. That way, $\Rproj$~defines a one-to-one correspondence between
  plays in $\calG$ and plays in~$\calG \ltimes \calA$ where the second
  component starts in~$q_0$. For a player~$B \ne A$, the objective is
  defined so that $\Rproj(\rho)$ has the same payoff as $\rho$. Consider
  now player $A$, she is winning in~$\calG$ for $\rho = (s_n)_{n\ge
    0}$ iff $(s_n)_{n \ge 0} \in \Lang(\calA)$ iff the unique path
  $(q_n)_{n \ge 0}$ from initial state~$q_0$ that reads $(s_n)_{n\ge
    0}$ satisfies the Rabin condition $(Q_i,R_i)_{i\in\lsem 1,n\rsem}$
  in $\calA$ iff $\Rproj^{-1}(\rho)$ satisfies the internal Rabin
  condition $(Q'_i,R'_i)_{i\in\lsem 1,n\rsem}$ in $\calG \ltimes
  \calA$. This proves that~$\simulrel$ is winning-preserving.

  \medskip It remains to show that both relations are game
  simulations.  Assume $s \simulrel (s,q)$ and pick a move $m_\Agt$ in
  $\calG$. It~is also a move in~$\calG \ltimes \calA$, and
  $\Tab'((s,q),m_\Agt) = (\Tab(s,m_\Agt),\delta(q,s))$. By definition
  of $\simulrel$ it then holds that $\Tab(s,m_\Agt) \simulrel
  \Tab'((s,q),m_\Agt)$, which proves condition~\eqref{cond:sim22} of
  the definition of a game simulation. It remains to show
  condition~\eqref{cond:sim21}. Pick a state $(s',q') \in \Stat'$. We
  distinguish two cases
  \begin{itemize}
  \item If $\delta(q,s) \ne q'$ then
    $\Susp(((s,q),(s',q')),m_\Agt)=\varnothing$, and
    condition~\eqref{cond:sim21} trivially holds.
  \item Otherwise $\delta(q,s)=q'$. In that case, 
     for any move $m'_\Agt$, we have
    that $\Tab(s,m'_\Agt)=s'$ if, and only~if,
    $\Tab'((s,q),m'_\Agt)=(s',q')$. It~follows that
    $\Susp(((s,q),(s',q')),m_\Agt) = \Susp((s,s'),m_\Agt)$, which
    implies condition~\eqref{cond:sim21}.
  \end{itemize}
  This proves that $\calG \ltimes \calA$ game-simulates $\calG$.

  We now assume $(s,q) \simulrel' s$ and pick a move $m_\Agt$ in
  $\calG \ltimes \calA$. It is also a move in~$\calG$, and as
  previously, condition~\eqref{cond:sim22} obviously holds. Pick now
  $s' \in \Stat$. We define $q' = \delta(q,s)$, and we have $(s',q')
  \simulrel s'$ by definition of $\simulrel'$.  As before, we get
  condition~\eqref{cond:sim21} because $\Susp(((s,q),(s',q')),m_\Agt)
  = \Susp((s,s'),m_\Agt)$.
\end{proof}

We will solve the case where each player's objective is given by a
deterministic Rabin automaton by applying the above result
inductively. We will obtain a game where each player has an internal
Rabin winning condition.  Applying Proposition~\ref{prop:sim} each
time, we~get the following result:

\begin{proposition}
  Let $\calG=\tuple{\Stat,\Agt,\Act, \Allow,\Tab,
    (\mathord\prefrel_A)_{A\in\Agt}}$ be a finite concurrent game, where for
  each player $A$, the preference relation $\prefrel_A$ is
  single-objective given by a deterministic Rabin
  automaton~$\calA$. Write $\Agt = \{A_1,\dots,A_n\}$.  There is a
  Nash equilibrium $\sigma_\Agt$ in $\calG$ from some state~$s$ with
  outcome~$\rho$ iff there is a Nash equilibrium $\sigma'_\Agt$
  in~$\calG' = (((\calG \ltimes \calA_1) \ltimes \calA_2) \dots \times
  \calA_n)$ from $(s,q_{01},\dots,q_{0n})$ with outcome~$\rho'$, where
  $q_{0i}$~is the initial state of~$\calA_i$ and $\rho$ is the
  projection of $\rho'$ on $\calG$.
\end{proposition}

\subsubsection{Algorithm}
Assume that the objective of player $A_i$ is given by a deterministic
Rabin automaton~$\calA_i$. The algorithm for solving the constrained
NE existence problem starts by computing the product of the game with the
automata: $\calG' = (((\calG \ltimes \calA_1) \ltimes \calA_2) \dots
\times \calA_n)$.  The resulting game has size $|\calG| \times
\prod_{j\in \lsem 1,n\rsem} | \calA_j|$, which is exponential in the
number of players. For each player~$A_j$ ($1 \le j \le n$), the number
of Rabin pairs in the product game is that of the original
specification $\calA_j$, say $k_j$. We then apply the deterministic
algorithm that we have designed for Rabin objectives (see
Subsection~\ref{rabin:det-algo} page~\pageref{rabin:det-algo}), which
yields an exponential-time algorithm in our framework.

\subsubsection{Hardness}

We prove \PSPACE-hardness in the restricted case of deterministic
B\"uchi automata, by a reduction from (the~complement~of) the problem
of the emptiness of the intersection of several language given by
deterministic finite automata. This problem is known to be
\PSPACE-complete~\cite[Lemma~3.2.3]{kozen1977lower}.

We fix finite automata~$\calA_1, \dots , \calA_n$ over alphabet~$\Sigma$.
Let~$\Sigma'=\Sigma\cup\{\init,\final\}$, where $\init$ and~$\final$ are
two special symbols not in~$\Sigma$. For
every $j\in\lsem 1,n\rsem$, we construct a B\"uchi automaton~$\calA'_j$ from
$\calA_j$ as follows. We~add a state~$F$ with a self-loop
labelled by~$\final$ and an initial state~$I$ with a transition labelled
by~$\init$ to the original initial state. We~add transitions labelled
by~$\final$ from every terminal state to~$F$. We~set the B\"uchi condition
to~$\{F\}$. If~$\calL_j$ is the language recognised by~$\calA_j$, then the
language recognised by the B\"uchi automaton~$\calA'_j$ is $\calL'_j= \init
\cdot \calL_j \cdot \final^\omega$. The~intersection of the languages
recognised by the automata~$\calA_j$ is empty if, and only~if, the
intersection of the languages recognised by the automata~$\calA'_j$ is empty.

We construct the game~$\calG$, with $\Stat = \Sigma'$. For each $j\in
\lsem 1 ,n \rsem$, there is a player~$A_j$ whose objective is given by
$\calA'_j$ and one special player~$A_0$ whose objective is
$\Stat^\omega$ (she~is always winning).  Player~$A_0$ controls all the
states and there are transitions from any state to the states of
$\Sigma \cup \{\final\}$.  Formally $\Act= \Sigma\cup\{\final\}\cup \bot$,
for all state~$s\in \Stat$, $\Mov(s, A_0) = \Act$, and if $j\ne 0$ then
$\Mov(s,A_j) = \{\bot\}$ and for all $\alpha\in \Sigma \cup \{\final\}$,
$\Tab(s,(\alpha,\bot,\dots,\bot)) = \alpha$.

\begin{lemma}
  There is a Nash equilibrium in game~$\calG$ from~$\init$ where every
  player wins if, and only~if, the intersection of the languages
  recognised by the automata~$\calA'_j$ is not empty.
\end{lemma}

\begin{proof}
  If there is such a Nash equilibrium, let~$\rho$ be its outcome.
  The~path~$\rho$ forms a word of~$\Sigma'$, it~is accepted by every
  automata~$\calA'_j$ since every player wins. Hence the intersection of the
  languages~$\calL_j$ is not empty.
  
  Conversely, if a word~$w = \init \cdot w_1 \cdot w_2 \cdots$ is accepted
  by all the automata, player~$A_0$ can play in a way such that everybody is
  winning: if at each step~$j$ she plays~$w_j$, then the outcome is~$w$
  which is accepted by all the automata. It~is a Nash equilibrium
  since $A_0$ controls everything and cannot improve her payoff.
\end{proof}

Since \PSPACE is stable by complementation, this proves that the
constrained NE existence problem is \PSPACE-hard for objectives described
by B\"uchi automata.

In order to prove hardness for the NE existence problem we use results
from Section~\ref{lem:link-constr-exist}.  Winning conditions in
$E(E(\dots(E(\calG,A_n,A_0,\rho_n),\dots,A_2,A_0,\rho_2),A_1,A_0,\rho_1)$,
where $\rho_j$ is a winning play for~$A_i$, can be defined by slightly
modifying automata $\calA'_1,\dots,\calA'_n$ to take into account the
new states.  By Proposition~\ref{lem:link-constr-exist}, there exists
a Nash equilibrium in this game if, and only,~if there is one in~$\calG$
where all the players~win.  Hence \PSPACE-hardness also holds for the
NE existence problem.

\section{Ordered B\"uchi objectives}
\label{sec:buchi}

In this Section we assume that preference relations of the players are
given by ordered B\"uchi objectives (as~defined in
Section~\ref{sec:prefrel}), and we prove the results listed in
Table~\ref{table-buchi} (page~\pageref{table-buchi}).  We~first
consider the general case of preorders given as Boolean circuits, and
then exhibit several simpler cases.

For the rest of this section, we fix a game $\calG=\tuple{\Stat,\Agt,
  \Act,\Allow, \Tab,(\mathord\prefrel_A)_{A\in\Agt}}$, and assume that
$\prefrel_A$ is given by an ordered B\"uchi objective $\omega_A =
\langle (\Omega_i^A)_{1 \le i \le n_A}, (\mathord\preorder_A)_{A \in \Agt}
\rangle$.

\subsection{General case: preorders are given as circuits}
\label{subsec:general}

\begin{theorem}\label{thm:buchi-pspace}
  For finite games with ordered B\"uchi objectives where preorders are
  given as Boolean circuits, the value problem, the NE existence problem and
  the constrained NE existence problem are \PSPACE-complete.
\end{theorem}

\begin{proof}
  We explain the algorithm for the constrained NE existence problem.  We
  assume that for each player $A$, the preorder $\preorder_A$ is given
  by a Boolean circuit $C_A$. The algorithm proceeds by trying all the
  possible payoffs for the players.

  Fix such a payoff $(v^A)_{A\in\Agt}$, with $v^A \in \{0,1\}^{n_A}$
  for every player $A$. We~build a circuit $D_A$ which 
  represents a single objective for player $A$. Inputs to circuit $D_A$
  will be states of the game. This circuit is constructed from $C_A$
  as follows: We set all input gates $w_1 \cdots w_n$ of circuit~$C_A$
  to the value given by payoff $v^A$; The former input $v_i$ receives
  the disjunction of all the states in~$\Omega_i$; We negate the
  output. It is not hard to check that the new circuit $D_A$ is such
  that for every play $\rho$, $D_A[\Inf(\rho)]$ evaluates to \true \iff
  $\payoff_A(\rho) \not\preorder_A v^A$, \ie if $\rho$ is
  an improvement for player~$A$.

  Circuit $D_A$ is now viewed as a single objective for player~$A$, we
  write $\calG'$ for the new game. We look for Nash equilibria in
  this new game, with payoff~$0$ for each player. Indeed, a Nash
  equilibrium $\sigma_\Agt$ in~$\calG$ with payoff $(v^A)_{A \in
    \Agt}$ is a Nash equilibrium in game $\calG'$ with payoff
  $(0,\dots,0)$. Conversely a Nash equilibrium $\sigma_\Agt$ in game
  $\calG'$ with payoff $(0,\dots,0)$ is a Nash equilibrium in
  $\calG$ as soon as the payoff of its outcome (in~$\calG$) is
  $(v^A)_{A \in \Agt}$.

  We use the algorithm described in Section~\ref{subsec:circuits}.
  for computing Nash equilibria with single objectives given as
  Boolean circuits, and we slightly modify it to take into account the
  constraint that it has payoff~$v^A$ for each player~$A$.  This can
  be done in polynomial space, thanks to
  Proposition~\ref{lem:play-length}: it is sufficient to look for
  plays of the form $\pi \cdot \tau^\omega$ with $|\pi| \le |\Stat|^2$
  and $|\tau|\le |\Stat|^2$.

  \PSPACE-hardness was proven for single objectives given as a Boolean
  circuit (the circuit evaluates by setting to \texttt{true} all
  states that are visited infinitely often, and to \texttt{false} all
  other states) in Section~\ref{subsec:circuits}.  This kind of
  objective can therefore be seen as an ordered B\"uchi objective with
  a preorder given as a Boolean circuit. 
\end{proof}

\subsection{When the ordered objective can be (co-)reduced to a single
  B\"uchi objective}
\label{subsec:reducible}
For some ordered objectives, the preference relation can (efficiently)
be reduced to a single objective. For instance, a~disjunction of
several B\"uchi objectives can obviously be reduced to a single
B\"uchi objective, by considering the union of the target sets.
Formally, we~say that an ordered B\"uchi objective $\omega = \langle
(\Omega_i)_{1 \le i \le n},\preorder \rangle$ is \newdef{reducible} to
a single B\"uchi objective if, given any payoff vector~$v$, we~can
construct in polynomial time a target set~$\widehat T(v)$ such that
for all paths~$\rho$, $v \preorder \payoff_\omega(\rho)$ \iff
$\Inf(\rho) \cap \widehat T(v) \neq \emptyset$.  It~means that
\emph{securing} payoff~$v$ corresponds to ensuring infinitely many
visits to the new target set.  Similarly, we~say that $\omega$ is
\newdef{co-reducible} to a single B\"uchi objective if for any
vector~$v$ we can construct in polynomial time a target set~$\widehat
T(v)$ such that $\payoff_\omega(\rho) \not\preorder v$ \iff
$\Inf(\rho) \cap \widehat T(v) \neq \emptyset$.  It~means that
\emph{improving} on payoff~$v$ corresponds to ensuring infinitely many
visits to the new target

We prove the following proposition, which exploits
(co-)reducibility for efficiently solving the various problems.
\begin{proposition}\hfill
  \label{prop:buchi-reducible} 
  \begin{itemize}
  \item For finite games with ordered B\"uchi objectives which are
    reducible to single B\"uchi objectives, and in which the preorders
    are non-trivial\footnote{That is, there is more than one class in
      the preorder.} and monotonic, the value problem is \P-complete.
  \item For finite games with ordered B\"uchi objectives which are
    co-reducible to single B\"uchi objectives, and in which the
    preorders are non-trivial and monotonic the NE existence problem and
    the constrained NE existence problem are \P-complete.
  \end{itemize}
\end{proposition}

\noindent Note that the hardness results follow from the hardness of the same
problems for single B\"uchi objectives (see
Section~\ref{subsec:buchi}).  We now prove the two upper bounds.

\subsubsection{Reducibility to single B\"uchi objectives and the value problem.}
\label{lem:value-buchi-reducible}

We transform the ordered B\"uchi objectives of the considered player
into a single B\"uchi objective, and use a polynomial-time
algorithm~\cite[Chapter~2]{GTW02} to solve the resulting zero-sum
(turn-based) B\"uchi game.

\subsubsection{Co-reducibility to single B\"uchi objectives and the
  (constrained) NE existence problem.}
\label{subsec:co-reducible}

We assume that the ordered objectives $(\omega_A)_{A \in \Agt}$ are
all co-reducible to single B\"uchi objectives.  We show that we
can use the algorithm presented in Section~\ref{subsubsec:algo2} to
solve the constrained NE existence problem in polynomial time.

We first notice that the preference relations $\prefrel_A$ satisfy the
hypotheses $(\star)$ (see page~\pageref{hyp:star}): $(\star)_a$ and
$(\star)_b$ are obvious, and $(\star)_c$ is by co-reducibility of the ordered
objectives. It means that we can apply the results of
Lemmas~\ref{lem:characterization} and~\ref{lem:character2} to the current
framework. To be able to conclude and apply Lemma~\ref{lem:compute-sol}, we
need to show that for every payoff~$v$, we~can compute in polynomial time the
set~$W(\calG,v)$ in the suspect game~$\calH(\calG,v)$.

\begin{lemma}\label{lem:w-poly}
  Fix a threshold $v$. The set $W(\calG,v)$ can be computed in
  polynomial time.
\end{lemma}
\begin{proof}
  As the ordered objectives are co-reducible to single B\"uchi
  objectives, we can construct in polynomial time target sets
  $\widehat T^A(v)$ for each player $A$. The objective of $\Eve$ in
  the suspect game $\calH(\calG,K)$ is then equivalent to a co-B\"uchi
  objective with target set $\{{(\widehat T^A(v,P)} \mid {A \in P}\}$.
  The winning region $W(\calG,v)$ can then be determined using a
  polynomial time algorithm of~\cite[Sect.~2.5.3]{GTW02}.
\end{proof}

\subsubsection{Applications.}
We will give preorders to which the above applies, allowing to infer
several \P-completeness results in Table~\ref{table-buchi} (those
written with reference ``Section~\ref{subsec:reducible}'').

We first show that reducibility and co-reducibility coincide 
when the preorder is total.

\begin{lemma}\label{lemma-total}
  Let $\omega = \langle (\Omega_i)_{1 \le i \le n},\preorder \rangle$
  be an ordered B\"uchi objective, and assume that $\preorder$ is
  total.  Then, $\omega$~is reducible to a single B\"uchi objective
  \iff $\omega$ is co-reducible to a single B\"uchi objective.
\end{lemma} 
\begin{proof}
  Let~$u \in \{0,1\}^n$ be a vector. If $u$~is a maximal element, the
  new target set is empty, and thus satisfies the property for
  co-reducibility.  Otherwise we pick a vector~$v$ among the smallest
  elements that is strictly larger than~$u$.  Since the preorder is
  reducible to a single B\"uchi objective, there is a target
  set~$\widehat T$ that is reached infinitely often whenever the
  payoff is greater than~$v$.  Since the preorder is total and by
  choice of~$v$, we~have $w \not\preorder u \Leftrightarrow v
  \preorder w$. Thus the target set~$\widehat T$ is visited infinitely
  often when~$u$ is not larger than the payoff.  Hence $\omega$ is
  co-reducible to a single B\"uchi objective.

  The proof of the other direction is similar.
\end{proof}

\begin{lemma}
  \label{lemma:examples}
  Ordered B\"uchi objectives with disjunction or maximise preorders
  are reducible to single B\"uchi objectives. Ordered B\"uchi
  objectives with disjunction, maximise or subset preorders are
  co-reducible to single B\"uchi objectives.
\end{lemma}

\proof
  Let $\omega = \langle (\Omega_i)_{1 \le i \le n},\preorder \rangle$
  be an ordered B\"uchi objective. Assume $T_i$ is the target set
  for $\Omega_i$.

  Assume $\preorder$ is the disjunction preorder. If the payoff~$v$ is
  different from $\Zero$ then we define~$\widehat T(v)$ as the union of
  all the target sets: $\widehat T (v) = \bigcup_{i=1}^n T_i$.  Then, for
  every run $\rho$, 
  \begin{eqnarray*}
    v \preorder \payoff_\omega(\rho) & \Leftrightarrow &
    \text{there is some}\ i\ \text{for which}\ \Inf(\rho) \cap T_i \ne \varnothing \\
    & \Leftrightarrow & \Inf(\rho)\cap \widehat T(v) \ne \varnothing
  \end{eqnarray*}
  If the payoff~$v$ is~$\Zero$ then we get the expected result with
  $\widehat T(v) = \Stat$.  Disjunction being a total preorder, it is
  also co-reducible (from Lemma~\ref{lemma-total}).
  
  We assume now that $\preorder$ is the maximise preorder.  Given a
  payoff~$v$, consider the index $ i_0=\max\{i\mid v_i = 1 \}$.  We
  then define $\widehat T(v)$ as the union of the target sets that are
  above~$i_0$: $\widehat T(v) = \bigcup_{i\geq i_0} T_i$.  The following
  four statements are then equivalent, if $\rho$ is a run:
  \begin{eqnarray*}
    v \preorder \payoff_\omega(\rho) & \Leftrightarrow &
    v \preorder \One_{\{i\mid \Inf(\rho)\cap T_i\ne \varnothing\}} \\
    & \Leftrightarrow &
    i_0  \leq \max\{i\mid \Inf(\rho) \cap T_i \ne \varnothing\} \\
    & \Leftrightarrow &
    \exists i\geq i_0.\ \Inf(\rho) \cap T_i \ne  \varnothing
  \end{eqnarray*}
  Hence $\omega$ is reducible, and also co-reducible as it~is total,
  to a single B\"uchi objective.

  Finally, we assume that $\preorder$ is the subset preorder, and we
  show that $\omega$ is then co-reducible to a single
  B\"uchi objective.  Given a payoff~$v$, the new target is the
  union of the target sets that are not reached infinitely often for
  that payoff: 
  $\widehat T(v) = \bigcup_{\{i \mid v_i = 0\}} T_i$.  Then the following
  statements are equivalent, if $\rho$ is a run:
  \begin{eqnarray*}
    \payoff_\omega(\rho) \not\preorder u & \Leftrightarrow & 
    \One_{\{i\mid \Inf(\rho)\cap T_i\ne \varnothing\}} \not\preorder u \\
    & \Leftrightarrow & 
    \exists i.\ \Inf(\rho)\cap T_i \ne \varnothing \text{ and }u_i = 0 \\
    &\Leftrightarrow& \Inf(\rho)\cap \widehat T(v) \ne
                      \varnothing
  \rlap{\hbox to 155 pt{\hfill\qEd}}
   \end{eqnarray*}

As a corollary, we get the following result:
\begin{corollary}
  For finite games with ordered B\"uchi objectives, with either the
  disjunction or the maximise preorder, the value problem is
  \P-complete.
  For finite games with ordered B\"uchi objectives, with either the
  disjunction, the maximise or the subset preorder, the NE existence problem and
  the constrained NE existence problem are \P-complete.
\end{corollary}

\begin{remark}
  Note that we cannot infer \P-completeness of the
  value problem for the subset preorder since the subset preorder is
  not total, and ordered objectives with subset preorder are not
  reducible to single B\"uchi objectives. Such an ordered objective
  is actually reducible to a generalised B\"uchi objective (several
  B\"uchi objectives should be satisfied).
\end{remark}

\subsection{When the ordered objective can be reduced to a
  deterministic B\"uchi automaton objective.}
\label{sec:reduc-buchi-auto}

For some ordered objectives, the preference relation can (efficiently)
be reduced to the acceptance by a deterministic B\"uchi automaton.
Formally, we~say that an ordered objective $\omega = \langle
(\Omega_i)_{1 \le i \le n},\preorder \rangle$ is \newdef{reducible} to
a deterministic B\"uchi automaton whenever, given any payoff
vector~$u$, we~can construct in polynomial time a deterministic
B\"uchi automaton over $\Stat$ which accepts exactly all plays $\rho$
with $u \preorder \payoff_\omega(\rho)$.  
For such preorders, we will see that the value problem can be solved
efficiently by constructing the product of the deterministic B\"uchi
automaton and the arena of the game. This construction does however
not help for solving the (constrained) NE existence problems since the
number of players is a parameter of the problem, and the size of the
resulting game will then be exponential.

\begin{proposition}\label{prop:buchi-value-p}
  For finite games with ordered B\"uchi objectives which are reducible
  to deterministic B\"uchi automata, the value problem is \P-complete.
\end{proposition}
\begin{proof}
  Given the payoff~$v^A$ for player~$A$, the algorithm proceeds by
  constructing the automaton that recognises the plays with payoff
  higher than~$v^A$.  By performing the product with the game as
  described in Section~\ref{sec:productB}, we obtain a new game, in
  which there is a winning strategy \iff there is a strategy in the
  original game to ensure payoff~$v^A$.  In~this new game, player~$A$
  has a single B\"uchi objective, so that the NE existence of a winning
  strategy can be decided in polynomial time.

  Hardness follows from that of games with single B\"uchi objectives.
\end{proof}

\subsubsection*{Applications}
We now give preorders to which the above result applies, that is,
which are reducible to deterministic B\"uchi automata objectives.

\begin{lemma}
  \label{lemma:conj}
  An ordered objective where the preorder is either the conjunction,
  the subset or the lexicographic preorder is reducible to a
  deterministic B\"uchi automaton objective.
\end{lemma}
\begin{proof}
  We first focus on the \textbf{conjunction preorder}.  Let $\omega =
  \langle (\Omega_i)_{1 \le i \le n},\preorder \rangle$ be an ordered
  B\"uchi objective, where $\preorder$ is the conjunction.  For every
  $1 \le i \le n$, let $T_i$ be the target set defining the B\"uchi
  condition $\Omega_i$.  There are only two possible payoffs: either
  all objectives are satisfied, or one objective is not satisfied. For
  the second payoff case, any play has a larger payoff: hence the
  trivial automaton (which accepts all plays) witnesses the
  property. For the first payoff case, we~construct a deterministic
  B\"uchi automaton~$\calB$ as follows.  There is one state for each
  target set, plus one accepting state: $Q = \{q_0, q_1, \dots, q_n
  \}$; the initial state is $q_0$, and the unique repeated state is
  $q_n$.
  For all $1 \le i \le n$, the transitions are $q_{i-1}
  \xrightarrow{s} q_{i}$ when $s\in T_i$ and $q_{i-1} \xrightarrow{s}
  q_{i-1}$ otherwise. There are also transitions $q_{n}
  \xrightarrow{s} q_0$ for every $s \in \Stat$.  Automaton $\calB$
  describes the plays that goes through each set $T_i$ infinitely
  often, hence witnesses the property. It can furthermore be computed
  in polynomial time. The construction is illustrated in
  Figure~\ref{fig:conjunction-automaton}.

  \medskip We now turn to the \textbf{subset preorder}.  Let $\omega =
  \langle (\Omega_i)_{1 \le i \le n},\preorder \rangle$ be an ordered
  B\"uchi objective, where $\preorder$ is the subset preorder.  For
  every $1 \le i \le n$, let $T_i$ be the target set defining the
  B\"uchi condition $\Omega_i$. Fix a payoff~$u$. A~play $\rho$ is
  such that $u \preorder \payoff_\omega(\rho)$ \iff $\rho$ visits
  infinitely often all sets $T_i$ with $u_i=1$. This is then
  equivalent to the conjunction of all $\Omega_i$'s with $u_i=1$. We
  therefore apply the previous construction for the conjunction and get
  the expected result.

  \medskip We finish this proof with the \textbf{lexicographic
    preorder}.  Let $\omega = \langle (\Omega_i)_{1 \le i \le
    n},\preorder \rangle$ be an ordered B\"uchi objective, where
  $\preorder$ is the lexicographic preorder.  For every $1 \le i \le
  n$, let $T_i$ be the target set defining the B\"uchi condition
  $\Omega_i$. Let $u \in \{0,1\}^n$ be a payoff vector. We construct
  the following deterministic B\"uchi automaton which recognises the
  runs whose payoff is greater than or equal to~$u$.

  In this automaton there is a state~$q_i$ for each $i$ such that
  $u_i=1$, and a state~$q_0$ that is both initial and repeated: $Q =
  \{q_0\} \cup \{q_i \mid u_i=1\}$. We write $I = \{0\} \cup \{i \mid
  u_i=1\}$. For every $i \in I$, we write $\mathsf{succ}(i) = \min (I
  \setminus \{j \mid j \le i\})$, with the convention that $\min
  \emptyset = 0$.  The transition relation is defined as follows:
  \begin{itemize}
  \item for every $s \in \Stat$, there is a transition $q_0
    \xrightarrow{s} q_{\mathsf{succ}(0)}$;
  \item for every $i \in I \setminus \{0\}$, we have the following
    transitions:
    \begin{itemize}
    \item $q_{i} \xrightarrow{T_i} q_{\mathsf{succ}(i)}$;
    \item $q_i \xrightarrow{T_k \setminus T_i} q_0$
      with $k<i$ and $u_k=0$;
    \item $q_i \xrightarrow{s} q_i$ for every $s \in \Stat \setminus
      (T_i \cup \bigcup_{k<i,u_k=0}T_k)$.
    \end{itemize}
  \end{itemize}

 \noindent An example of the construction is given in
  Figure~\ref{fig:lexico-automaton}.

  \smallskip We now prove correctness of this construction.  Consider
  a path that goes from~$q_0$ to~$q_0$: if the automaton is currently
  in state~$q_i$, then since the last occurrence of~$q_0$, at~least
  one state for each target set~$T_j$ with $j<i$ and $u_j = 1$ has
  been visited.  When $q_0$ is reached again, either it~is because we
  have seen all the $T_j$ with $u_j = 1$, or it~is because the run
  visited some target~$T_i$ with $u_i=0$ and all the $T_j$ such that
  $u_j=1$ and $j<i$; in both cases, the set of targets that have been
  visited between two visits to~$q_0$ describes a payoff greater
  than~$u$.  Assume the play~$\pi$ is accepted by the automaton; then
  there is a sequence of~$q_i$ as above that is taken infinitely
  often, therefore $\payoff_\omega(\pi)$ is greater than or equal
  to~$u$ for the lexicographic order.

  Conversely assume $v = \payoff_\omega(\pi)$ is greater than or equal
  to~$u$, that we already read a prefix $\pi_{\le k}$ for some~$k$,
  and that the current state is~$q_0$.  Reading the first symbol
  in~$\pi$ after position~$k$, the run goes to the state~$q_i$ where
  $i$ is the least integer such that $u_i=1$.  Either the path
  visits~$T_i$ at some point, or it~visits a state in a target $T_j$,
  with $j$ smaller than~$i$ and $v_j=0$, in which case the automaton
  goes back to~$q_0$.  Therefore from~$q_0$ we can again come back
  to~$q_0$ while reading the following of~$\pi$, and the automaton
  accepts.
\end{proof}

\begin{figure}[!ht]
  \begin{minipage}{.48\textwidth}
  \centering
  \begin{tikzpicture}[scale=1,thick]
    \tikzstyle{rond}=[draw,circle,minimum size=6mm,inner sep=0mm,fill=black!10]
    \draw (0,0) node[rond] (A) {$q_0$};
    \draw (1.5,0) node[rond] (B) {$q_1$};
    \draw (3,0) node[rond] (C) {$q_2$};
    \draw (1,-1) node[rond,double] (D) {$q_3$};
    \draw[-latex'] (-0.8,0) -- (A);
    \draw[-latex'] (A) -- node[above] {$T_1$} (B);
    \draw[-latex'] (B) -- node[above] {$T_2$} (C);
    \draw[-latex',rounded corners=4mm] (C) |- node[below,pos=0.7]
         {$T_3$} (D.0); 
    \draw[-latex',rounded corners=4mm] (D) -| node[left] {$\Stat$}
    (A);
    \draw[-latex'] (A) .. controls  +(.5,1) and +(-.5,1) .. (A);
    \draw[-latex'] (B) .. controls  +(.5,1) and +(-.5,1) .. (B);
    \draw[-latex'] (C) .. controls  +(.5,1) and +(-.5,1) .. (C);
  \end{tikzpicture}
  \caption{The automaton for the conjunction preorder, $n =  3$}
  \label{fig:conjunction-automaton}
  \end{minipage}\hskip-.05\textwidth
  \begin{minipage}{.57\textwidth}
  \centering
  \begin{tikzpicture}[scale=1,thick]
    \tikzstyle{rond}=[draw,circle,minimum size=6mm,inner sep=0mm,fill=black!10]
    \draw (0,0) node[rond] (A) {$q_2$};
    \draw (3,0) node[rond] (B) {$q_5$};
    \draw (5,0) node[rond] (C) {$q_6$};
    \draw (1,-1) node[rond,double] (E) {$q_0$};
    \draw[-latex'] (A) -- node[above] {$T_2$} (B);
    \draw[-latex'] (B) -- node[above] {$T_5$} (C);
    \draw[-latex'] (.2,-1.2) -- (E);
    \draw[-latex',rounded corners=4mm] (C) |- node[below,pos=0.7]
         {$T_1,T_3,T_4,T_6$} (E.0); 
    \draw[-latex',rounded corners=4mm] (B) |- node[above,pos=0.78]
         {$T_1,T_3,T_4$} (E.20); 
    \draw[-latex'] (A) -- node[above,pos=0.7] {\ $T_1$} (E);
    \draw[-latex',rounded corners=4mm] (E) -| node[left] {$\Stat$}
    (A);
    \draw[-latex'] (A) .. controls  +(.5,1) and +(-.5,1) .. (A);
    \draw[-latex'] (B) .. controls  +(.5,1) and +(-.5,1) .. (B);
    \draw[-latex'] (C) .. controls  +(.5,1) and +(-.5,1) .. (C);
  \end{tikzpicture}
  \caption{The automaton for the lexicographic order, $n = 7$ and $u=
    (0,1,0,0,1,1,0)$}\label{fig:lexico-automaton}
  \end{minipage}
\end{figure}

We conclude with the following corollary:
\begin{corollary}
  For finite games with ordered B\"uchi objectives with either of the
  conjunction, the lexicographic or the subset preorders, the value
  problem is \P-complete.
\end{corollary}

\subsection{Preference relations with monotonic preorders}
\label{subsec:monotonic}
We will see in this part that monotonic preorders lead to more
efficient algorithms. More precisely we prove the following result:

\begin{proposition}\label{prop:buchi-np}
  \begin{itemize}
  \item For finite games with ordered B\"uchi objectives where the preorders
    are given by monotonic Boolean circuits, the value problem is in \coNP,
    and the NE existence problem and the constrained NE existence problem are
    in \NP.
  \item Completeness holds in both cases for finite games with ordered
    B\"uchi objectives where the preorders are given by monotonic
    Boolean circuits or with the counting preorder. 
  \item \NP-completeness also holds for the constrained NE existence
    problem for finite games with ordered B\"uchi objectives where the
    preorders admit an element~$v$ such that for every $v'$, it~holds
    $v' \ne \One \Leftrightarrow v' \preorder v$.\footnote{To be fully
      formal, a preorder~$\preorder$ is in fact a family
      $(\mathord\preorder_n)_{n\in\N}$ (where $\preorder_n$ compares two
      vectors of size~$n$), and this condition should be stated as
      ``\textit{for all~$n$, there is an element~$v_n\in\{0,1\}^n$
        such that for all~$v'\in\{0,1\}^n$, it~holds $v' \ne \One
        \Leftrightarrow v' \preorder v_n$}''.}
  \end{itemize}
\end{proposition}

\noindent We first show that monotonicity of the preorders imply some
memorylessness property in the suspect game. We then give algorithms
witnessing the claimed upper bounds, and show the various lower bounds.

\subsubsection{When monotonicity implies memorylessness.}
We say that a strategy $\sigma$ is \emph{memoryless} (resp. memoryless
from state $s_0$) if there exists a function $f\colon \Stat \to \Act$
such that $\sigma(h\cdot s) = f (s)$ for every $h \in \Hist$
(resp. for every $h \in \Hist(s_0)$).  A strategy profile is said
memoryless whenever all strategies of single players are
memoryless. We show that when the preorders used in the ordered
B\"uchi objectives are monotonic, the three problems are also easier
than in the general case. This is because we can find memoryless
trigger profiles (recall Definition~\ref{def:trigger}).

We first show this lemma, that will then be applied to the suspect
game.

\begin{lemma}\label{lem:memoryless}
  Let $\calH$ be a turn-based two-player game. Call \Eve one player,
  and let $\sigma_\shortEve$ be a strategy for~\Eve, and $s_0$ be a
  state of~$\calH$.  There is a memoryless
  strategy~$\sigma'_\shortEve$ such that for every $\rho' \in
  \Out_{\calH}(s_0,\sigma'_\shortEve)$, there exists $\rho \in
  \Out_{\calH}(s_0,\sigma_{\shortEve})$ such that $\Inf(\rho')
  \subseteq\Inf(\rho)$.
\end{lemma}

\begin{proof}
  This proof is by induction on the size of the set 
  \[
  S(\sigma_1) =
  \{(s,m) \mid \exists h \in\Hist(\sigma_1) .\ \sigma_1(h) = m\
  \text{and}\ \last(h)=s\}.  
  \]
  If its size is the same as that of $\{ s \mid \exists h \in
  \Hist(\sigma_1).\ \last(h) = s \}$ then the strategy is memoryless.
  Otherwise, let~$s$ be a state at which $\sigma_1$ takes several
  different actions (\ie, $|(\{s\}\times\Act) \cap S(\sigma_1)| > 1$).

  We will define a new strategy~$\sigma'_1$ that takes fewer different
  actions in $s$ and such that for every outcome of~$\sigma_1'$, there
  is an outcome of~$\sigma_1$ that visits (at~least) the same states
  infinitely often.

  If $\sigma$ is a strategy and $h$~is ~a~history, we let $\sigma\circ
  h\colon h'\mapsto \sigma(h \cdot h')$ for any history~$h'$. Then for
  every~$m$ such that $(s,m) \in S(\sigma_1)$ we let ${H_m
    = \{h \in \Hist(\sigma_1) \mid \last(h)=s\ \text{and}\
    \sigma_1(h)=m\}}$, and for every~$h$, $h^{-1} \cdot H_m = \{h'
  \mid h \cdot h' \in H_m\}$.
  We pick $m$ such that $H_m$ is not empty.
  \begin{itemize}
  \item Assume that there is $h_0 \in \Hist(\sigma_1)$ with
    $\last(h_0) = s$, such that $h_0^{-1} \cdot H_m$ is empty.  We
    define a new strategy $\sigma'_1$ as follows.  If $h$ is an
    history which does not visit $s$, then $\sigma'_1(h)=\sigma_1(h)$.
    If $h$ is an history which visits $s$, then decompose $h$ as $h'
    \cdot h''$ where $\last(h') = s$ is the first visit to $s$ and
    define $\sigma'_1(h) = \sigma_1(h_0 \cdot h'')$.  Then,
    strategy~$\sigma'_1$ does not use~$m$ at state~$s$, and therefore
    at least one action has been ``removed'' from the strategy.  More
    precisely, $|(\{s\}\times\Act) \cap S(\sigma'_1)| \le
    |(\{s\}\times\Act) \cap S(\sigma_1)| - 1$.  Furthermore the
    conditions on infinite states which are visited infinitely often
    by outcomes of $\sigma'_1$ is also satisfied.
  \item Otherwise for any $h \in \Hist(\sigma_1)$ with $\last(h) = s$,
    $h^{-1}\cdot H_m$ is not empty.  We will construct a strategy
    $\sigma'_1$ which plays~$m$ at~$s$.  Let $h$ be an history, we
    first define the extension~$e(h)$ inductively in that way:
    \begin{itemize}
    \item $e(\varepsilon) = \varepsilon$, where $\varepsilon$ is the
      empty history;
    \item $e(h\cdot s) = e (h) \cdot h'$ where $h' \in (e(h))^{-1}
      \cdot H_m$;
    \item $e(h \cdot s') = e(h) \cdot s'$ if $s' \ne s$.
    \end{itemize}
    We extend the definition of $e$ to infinite outcomes in the
    natural way: $e(\rho)_{i} = e(\rho_{\le i})_{i}$.  We then define
    the strategy $\sigma'_1 \colon h \mapsto \sigma_1(e(h))$.  We show
    that if $\rho$ is an outcome of $\sigma'_1$, then $e(\rho)$ is an
    outcome of $\sigma_1$. Indeed assume $h$ is a finite outcome
    of~$\sigma'_1$, that $e(h)$ is an outcome of~$\sigma_1$ and
    $\last(h) = \last(e(h))$. If $h\cdot s$ is an outcome of
    $\sigma'_1$, by construction of $e$, $e(h\cdot s) = e(h) \cdot
    h'$, such that $\last(h') = s$, and $h'$ is an outcome of
    $\sigma_1\circ e(h)$ and as $e(h)$ is an outcome of $\sigma_1$ by
    hypothesis, that means that $e(h\cdot s)$ is an outcome of
    $\sigma_1$. If $h\cdot s'$ with $s'\ne s$ is an outcome of
    $\sigma'_1$, $e(h \cdot s') = e(h) \cdot s'$, $s' \in \Tab
    (\last(h), \sigma'_1(h))$, and $\sigma'_1(h) =
    \sigma_1(e(h))$. Using the hypothesis $\last(h) = \last(e(h))$,
    and $e(h)$ is an outcome of~$\sigma_1$, therefore $e(h\cdot s')$
    is an outcome of~$\sigma_1$. This shows that if $\rho$ is an
    outcome of~$\sigma'_1$ then $e(\rho)$ is an outcome
    of~$\sigma_1$. The property on states visited infinitely often
    follows. Several moves have been removed from the strategy at $s$
    (since the strategy is now memoryless at~$s$, playing~$m$).
  \end{itemize}
  In all cases we have $S(\sigma'_1)$ strictly included in
  $S(\sigma_1)$, and an inductive reasoning entails the result.
\end{proof}

\begin{lemma}\label{lem:suspect-based}
  If for every player $A$, $\preorder_A$ is monotonic, and if there is
  a trigger profile for some play $\pi$ from~$s$, then there is a
  memoryless winning strategy for~$\Eve$ in $\calH(\calG,\pi)$ from
  state~$(s,\Agt)$.
\end{lemma}

\begin{proof}
  Assume there is a trigger profile for~$\pi$. We have seen in
  Lemma~\ref{lem:suspect-game} that there is then a winning
  strategy~$\sigma_\shortEve$ in game $\calH(\calG,\pi)$ for $\Eve$.
  Consider the memoryless strategy $\sigma'_\shortEve$ constructed as
  in Lemma~\ref{lem:memoryless}.  Let~$\rho'$ be an outcome of
  $\sigma'_\shortEve$, there is an outcome~$\rho$
  of~$\sigma_\shortEve$ such that $\Inf(\rho') \subseteq
  \Inf(\rho)$. As $\sigma_\shortEve$ is winning in~$\calH(\calG,\pi)$,
  for every $A \in \limitpi2(\rho)$, $\Sproj_1(\rho) \prefrel_A \pi$.
  We assume the B\"uchi conditions are given by the target sets
  $(T_i^A)_{A,i}$.  For each player~$A$, $\{i \mid \Inf
  (\Sproj_1(\rho'))\cap T^A_i\} \subseteq \{i \mid \Inf
  (\Sproj_1(\rho))\cap T^A_i\}$. As the preorder is monotonic the
  payoff of $\Sproj_1(\rho')$ is smaller than that of~$\Sproj_1(\rho)$:
  $\Sproj_1(\rho') \prefrel_A \Sproj_1(\rho)$.  So~the play is winning
  for any player~$A$ and $\sigma'_\shortEve$ is a memoryless winning
  strategy in game $\calH(\calG,\pi)$ for $\Eve$.
\end{proof}

\begin{lemma}\label{lem:check-threat}
  If for every player $A$, $\preorder_A$ is given by monotonic Boolean
  circuits, then given a path~$\pi$, we can decide in polynomial time
  if a memoryless strategy for~$\Eve$ in $\calH(\calG,\pi)$ is
  winning.
\end{lemma}
\begin{proof}
  Let $\sigma_\shortEve$ be a memoryless strategy in
  $\calH(\calG,\pi)$ for $\Eve$.  By keeping only the edges that are
  taken by $\sigma_\shortEve$, we define a subgraph of the game.  We
  can compute in polynomial time the strongly connected components of
  this graph.  If one component is reachable and does not satisfy the
  objective of~$\Eve$, then the strategy is not winning.  Conversely
  if all the reachable strongly connected components satisfy the
  winning condition of $\Eve$, since the preorder is monotonic,
  $\sigma_\shortEve$ is a winning strategy.  Notice that since the
  preorder is given as a Boolean circuit, we can check in polynomial
  time whether a strongly connected component is winning or
  not. Globally the algorithm is therefore polynomial-time.
\end{proof}

We now turn to the proof of the claimed upper bounds.

\subsubsection{Proofs for the upper bounds.}

We show that the value problem is in \coNP~for finite games with
ordered B\"uchi objectives, when preorders are given by monotonic
Boolean circuits.

As already mentioned at the beginning of Section~\ref{sec:single}, for
the value problem, we can make the concurrent game turn-based: since
player $A$ must win against any strategy of the coalition $P = \Agt
\setminus \{A\}$, she must also win in the case where the opponents'
strategies can adapt to what $A$ plays. In other terms, we can make
$A$ play first, and then the coalition.  This turn-based game is
determined, so~that there is a strategy~$\sigma$ whose outcomes are
always better (for~$A$) than $v^{A}$ \iff for any strategy~$\sigma'$
of coalition $P$, there is an outcome with payoff (for~$A$) better
than~$v^{A}$.  If there is a counterexample to this fact, then thanks
to Lemma~\ref{lem:memoryless} there is one with a memoryless
strategy~$\sigma'$.  The \coNP\ algorithm proceeds by checking that
all the memoryless strategies of coalition~$P$ have an outcome better
than~$v^{A}$, which is achievable in polynomial time, with a method
similar to Lemma~\ref{lem:check-threat}.

\medskip We show now that the constrained NE existence problem is in
\NP~for finite games with ordered B\"uchi objectives, when preorders
are given by monotonic Boolean circuits.

The algorithm for the constrained NE existence problem proceeds by
guessing:
\begin{itemize}
\item the payoff for each player,  
\item a~play of the form $\pi \cdot \tau^\omega$, where
  $|\pi|\le|\Stat|^2$ and $|\tau|\le|\Stat|^2$,
\item an under-approximation $W$ of the set of winning states
  in~$\calH(\calG,\pi \cdot \tau^\omega)$
\item a memoryless strategy profile~$\sigma_\Agt$ in~$\calH(\calG,\pi
  \cdot \tau^\omega)$.
\end{itemize}
We check that $\sigma_\Agt$ is a witness for the fact that the states
in $W$ are winning; thanks to Lemma~\ref{lem:check-threat}, this can
be done in polynomial time.  We also verify that the play $\pi \cdot
\tau^\omega$ has the expected payoff, that the payoff satisfies the
constraints, and that it never gets out of $W$.  If these conditions
are fulfilled, then the play $\pi \cdot \tau^\omega$ meets the
conditions of Theorem~\ref{thm:eq-win}, and there is a Nash
equilibrium with outcome $\pi \cdot \tau^\omega$.
Lemma~\ref{lem:suspect-based} and Proposition~\ref{lem:play-length}
ensure that if there is a Nash equilibrium, we~can find it this way.

\subsubsection{Proofs for the hardness results.}

We first prove the hardness results for the counting preorder.

\begin{lemma}\label{prop:buchi-counting-nphard}
  For finite games with ordered B\"uchi objectives that use the
  counting preorder, the value problem is \co\NP-hard.
\end{lemma}

\begin{proof}
  We~reduce (the complement of) \SAT into the value problem for
  two-player turn-based games with B\"uchi objectives with the
  counting preorder.  Consider an instance
  \[\phi = C_1 \land \cdots \land C_m\]
  with $C_j=\ell_{j,1} \lor \ell_{j,2} \lor \ell_{j,3}$, over a set of
  variables~$\{x_1,\ldots,x_n\}$.  With~$\phi$, we~associate a
  two-player turn-based game~$\calG$.  Its set of states is made of
  \begin{itemize}
  \item a set containing the unique initial state~$V_0=\{s_0\}$,
  \item a set of two states~$V_k=\{x_k, \lnot x_k\}$ for each~$1\leq
    k\leq n$,
  \item and a set of three states~$V_{n+j}=\{t_{j,1}, t_{j,2},
    t_{j,3}\}$ for each~$1\leq j\leq m$.
  \end{itemize}
  Then, for each~$0\leq l\leq n+m$, there is a transition between any
  state of~$V_l$ and any state of~$V_{l+1}$ (assuming
  $V_{n+m+1}=V_0$).

  The game involves two players: player~$B$ owns all the states, but
  has no objectives (she~always loses). Player~$A$ has a set of
  B\"uchi objectives defined by $T^A_{2\cdot k} = \{x_k\} \cup
  \{t_{j,p} \mid \ell_{j,p} = x_k \}$, $T^A_{2\cdot k+1} = \{ \lnot
  x_k\} \cup \{t_{j,p} \mid \ell_{j,p} = \lnot x_k \}$, for $1\leq
  k\leq n$.  Notice that at least $n$ of these objectives will be
  visited infinitely often along any infinite play. We prove that if
  the formula is not satisfiable, then at least $n+1$ objectives will
  be fulfilled, and conversely.

  \smallskip Assume the formula is satisfiable, and pick a witnessing
  valuation~$v$.  We define a strategy~$\sigma_B$ for~$B$ that
  ``follows'' valuation~$v$: from states in~$V_{k-1}$, for any~$1\leq
  k\leq n$, the strategy plays towards~$x_k$ if $v(x_k)=\texttt{true}$
  (and to~$\lnot x_k$ otherwise).  Then, from a state in~$V_{n+l-1}$
  with~$1\leq l\leq m$, it~plays towards one of the~$t_{j,p}$ that
  evaluates to true under~$v$ (the one with least
  index~$p$,~say). This way, the number of targets of player~$A$ that
  are visited infinitely often is~$n$.

  Conversely, pick a play in~$\calG$ s.t. at most (hence exactly) $n$
  objectives of~$A$ are fulfilled. In particular, for any~$1\leq k\leq
  n$, this play never visits one of~$x_k$ and~$\lnot x_k$, so that it
  defines a valuation~$v$ over $\{x_1,\ldots, x_n\}$. Moreover, any
  state of~$V_{n+l}$, with~$1\leq l\leq p$, that is visited infinitely
  often must correspond to a literal that is made true by~$v$, as
  otherwise this would make one more objective that is fulfilled
  for~$A$. As~a consequence, each clause of~$\phi$ evaluates to true
  under~$v$, and the result follows.
\end{proof}

\begin{figure}[!ht]
  \centering
    \begin{tikzpicture}[xscale=1.4,thick]
    \tikzstyle{rond}=[draw,circle,minimum size=7mm,inner sep=0mm,fill=black!10]
      \draw (0,0) node [rond] (A1) {$s_0$}; 
      \draw (1,1) node [rond] (B1) {$x_1$};
      \draw (1,-1) node [rond] (C1) {$\lnot
        x_1$}; 
      \draw (2,1) node [rond](B2){$x_2$};
      \draw (2,-1) node [rond](C2){$\lnot x_2$};
      \draw (3,1) node [rond] (B3) {$x_3$};
      \draw (3,-1) node [rond] (C3){$\lnot
          x_3$}; 
      \draw (4.5, 1.2) node [rond] 
        (B4) {$t_{1,1}$};
      \draw (4.5, 0  ) node [rond] 
      (C4) {$t_{1,2}$};
      \draw (4.5,-1.2) node [rond]
      (D4) {$t_{1,3}$};
      
      \draw (6, 1.2) node [rond] 
      (B5) {$t_{2,1}$};
      \draw (6, 0  ) node [rond] 
      (C5) {$t_{2,2}$};
      \draw (6,-1.2) node [rond] 
      (D5){$t_{2,3}$};
      \draw (7,0) node (A6) {};

      \draw [-latex'] (A1) -- (B1);
      \draw [-latex'] (A1) -- (C1);
      \draw [-latex'] (B1) -- (B2);
      \draw [-latex'] (C1) -- (B2);
      \draw [-latex'] (B1) -- (C2);
      \draw [-latex'] (C1) -- (C2);

      \draw [-latex'] (B2) -- (B3);
      \draw [-latex'] (C2) -- (B3);
      \draw [-latex'] (B2) -- (C3);
      \draw [-latex'] (C2) -- (C3);

      \draw [-latex'] (B3) -- (B4);
      \draw [-latex'] (C3) -- (B4);
      \draw [-latex'] (B3) -- (C4);
      \draw [-latex'] (C3) -- (C4);
      \draw [-latex'] (B3) -- (D4);
      \draw [-latex'] (C3) -- (D4);
      \draw [-latex'] (B4) -- (B5);
      \draw [-latex'] (C4) -- (B5);
      \draw [-latex'] (D4) -- (B5);
      \draw [-latex'] (B4) -- (C5);
      \draw [-latex'] (C4) -- (C5);
      \draw [-latex'] (D4) -- (C5);
      \draw [-latex'] (B4) -- (D5);
      \draw [-latex'] (C4) -- (D5);
      \draw [-latex'] (D4) -- (D5);
      \draw[rounded corners=4mm] (B5) -| (7,-1.6);
      \draw[rounded corners=4mm] (C5) -| (7,-1.6);
      \draw[rounded corners=4mm] (D5) -| (7,-1.6);
      \draw [rounded corners=4mm,-latex'] (7,-1.6) |- (5,-2) -| (A1);
    \end{tikzpicture}
  \caption{The game~$\calG$ associated with formula~$\phi$ of~\ref{eq:counting}}
  \label{fig:counting}
\end{figure}
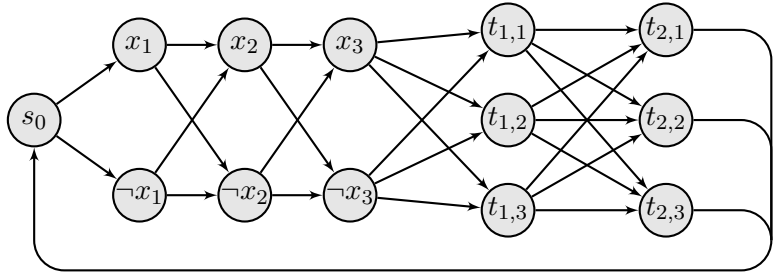
\begin{example}
  We illustrate the construction of the previous proof in
  Figure~\ref{fig:counting} for the formula
  \begin{equation}\label{eq:counting}
    \varphi = (x_1 \lor x_2 \lor \lnot x_3) \land (\lnot x_1 \lor
    x_2 \lor \lnot x_3)\,.
  \end{equation}
  The targets for player $A$ are $T_1 = \{ x_1, t_{1,1}\}$, $T_2 = \{
  \lnot x_1, t_{2,1}\}$, $T_3 = \{ x_2, t_{1,2}, t_{2,2}\}$, $T_4 =
  \{ \lnot x_2\}$, $T_5 = \{ x_3\}$, $T_6 = \{ \lnot x_3, t_{1,3},
  t_{2,3}\}$.  Player~$A$ cannot ensure visiting infinitely often four
  target sets, therefore the formula is satisfiable.
\end{example}

\begin{lemma}
  For finite games with ordered B\"uchi objectives that use the
  counting preorder, the NE existence problem is \NP-hard.
\end{lemma}
\begin{proof}
  Let $\calG$ be the game we constructed for
  Lemma~\ref{prop:buchi-counting-nphard}.  We construct the game
  $\calG''$ from $\calG$ as described in
  Section~\ref{sec:link-value-exist}.  The preference in $\calG'$ can
  still be described with ordered B\"uchi objectives and the counting
  preorder: the only target set of $B$ is $\{s_1\}$ and we add $s_1$
  to $n$ different targets of $A$, where $n$ is the number of
  variables as in Lemma~\ref{prop:buchi-counting-nphard}.  From
  Proposition~\ref{lem:link-value-exist} there is a Nash equilibrium
  in~$\calG''$ from~$s_0$ \iff $A$ cannot ensure visiting at least
  $n+1$ targets infinitely often. Hence the NE existence problem is
  \NP-hard.
\end{proof}

This proves also \NP-hardness for the constrained NE existence problem
for ordered B\"uchi objectives with the counting preorder. Hardness
results for preorders given by monotonic Boolean circuits follow from
the above since the counting preorder is a special case of preorder
given as a monotonic Boolean circuit (and the counting preorder can be
expressed as a polynomial-size monotonic Boolean circuit).

We now show hardness in the special case of preorders with (roughly)
at most one maximal element below $\One$.

\begin{lemma}\label{prop:buchi-nphard}
  For finite turn-based games with ordered B\"uchi objectives with a
  monotonic preorder for which there is an element~$v$ such that for
  every $v'$, $v' \ne \One \Leftrightarrow v' \preorder v$, the
  constrained NE existence problem is \NP-hard.
\end{lemma}

\begin{proof}
  Let us consider a formula $\phi = C_1 \land \cdots \land C_m$
  For~each variable~$x_i$, our~game has one player~$B_i$ and three
  states~$s_i$, $x_i$ and~$\lnot x_i$. The objectives of~$B_i$ are the
  sets~$\{x_i\}$ and~$\{\lnot x_i\}$.  Transitions go from each~$s_i$
  to~$x_i$ and~$\lnot x_i$, and from~$x_i$ and~$\lnot x_i$
  to~$s_{i+1}$ (with $s_{n+1}=s_0$). Finally, an~extra player~$A$ has
  full control of the game (\ie, she~owns all the states) and has $n$
  objectives, defined by $T^A_i = \{\ell_{i,1}, \ell_{i,2},
  \ell_{i,3}\}$ for $1\leq i\leq n$.  The construction is illustrated
  in Figure~\ref{fig:buchi-nphard}.
\begin{figure}[!ht]
  \centering
    \begin{tikzpicture}[xscale=1.4,thick]
    \tikzstyle{rond}=[draw,circle,minimum size=7mm,inner sep=0mm,fill=black!10]
      \draw (0,0) node [rond] (A1) {$s_1$}; 
      \draw (1,.8) node [rond] (B1) {$x_1$};
      \draw (1,-.8) node [rond] (C1) {$\lnot
        x_1$}; 
      \draw (2,0) node [rond] (A2) {$s_2$};
        \draw (3,.8) node [rond](B2){$x_2$};
        \draw (3,-.8) node [rond](C2){$\lnot x_2$};
        \draw (4,0) node [rond] (A3) {$s_3$}; 
        \draw (5,.8) node [rond] (B3) {$x_3$};
        \draw (5,-.8) node [rond] (C3){$\lnot
          x_3$}; 
        \draw (6,0) node [rond] (A4) {$s_4$};
        \draw (7,.8) node [rond] (B4) {$x_4$};
        \draw (7,-.8) node [rond] (C4){$\lnot x_4$};

      \draw [-latex'] (A1) -- (B1);
      \draw [-latex'] (A1) -- (C1);
      \draw [-latex'] (B1) -- (A2);
      \draw [-latex'] (C1) -- (A2);
      \draw [-latex'] (A2) -- (B2);
      \draw [-latex'] (A2) -- (C2);

      \draw [-latex'] (B2) -- (A3);
      \draw [-latex'] (C2) -- (A3);

      \draw [-latex'] (A3) -- (B3);
      \draw [-latex'] (A3) -- (C3);
      \draw [-latex'] (B3) -- (A4);
      \draw [-latex'] (C3) -- (A4);
      \draw [-latex'] (A4) -- (B4);
      \draw [-latex'] (A4) -- (C4);
      \draw[rounded corners=4mm] (B4) -| (8,-1.2);
      \draw[rounded corners=4mm] (C4) -| (8,-1.2);
      \draw [rounded corners=4mm,-latex'] (8,-1.2) |- (5,-1.6) -| (A1);
    \end{tikzpicture}
  \caption{The B\"uchi game for a formula with $4$ variables}
  \label{fig:buchi-nphard}
\end{figure}
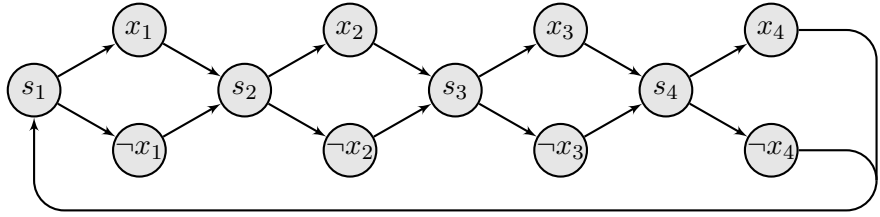

\smallskip We show that formula~$\phi$ is satisfiable \iff there is a
Nash equilibrium where each player~$B_i$ gets payoff~$\beta_i$
satisfying~$\beta_i\preorder v$ (hence $\beta_i\not=(1,1)$), and
player~$A$ gets payoff~$\One$.

First assume that the formula is satisfiable, and pick a witnessing
valuation~$u$. By~playing according to~$u$, player~$A$ can satisfy all
of her objectives (hence she cannot improve her payoff, since the
preorder is monotonic). Since she alone controls all the game, the
other players cannot improve their payoff, so that this is a Nash
equilibrium. Moreover, since~$A$ plays memoryless, only one of~$x_i$
and~$\lnot x_i$ is visited for each~$i$, so~that the payoff~$\beta_i$
for~$B_i$ satisfies~$\beta_i\preorder v$.  Conversely, if there is a
Nash equilibrium with the desired payoff, then by hypothesis, exactly
one of each~$x_i$ and~$\lnot x_i$ is visited infinitely often (so that
the payoff for~$B_i$ is not~$(1,1)$), which defines a
valuation~$u$. Since in this Nash equilibrium, player~$A$ satisfies
all its objectives, one state of each target is visited, which means
that under valuation~$u$, formula~$\phi$ evaluates to true.
\end{proof}

\subsubsection{Applications.}
We now describe examples of preorders which satisfy the conditions on
the existence of an element $v$ such that $v' \ne \One \Leftrightarrow
v' \preorder v$.

\begin{lemma}\label{lem:second-maximal}
  Conjunction, counting and lexicographic preorders have an element
  $v$ such that $v' \ne \One \Leftrightarrow v' \preorder v$.
\end{lemma}
\begin{proof}
  Consider $v = (1,\dots,1,0)$, and $v' \ne \One$.  For conjunction,
  there is~$i$ such that $v'_i=0$, so $v'\preorder v$.  For counting,
  $|\{i\mid v'_i = 1\}| < n$, so $v'\preorder v$.  For the
  lexicographic preorder, let $i$ be the smallest index such that
  $v'_i =0$, and either $v_i=1$ and $v_j = v'_j$ for all $j<i$, or for
  all $j \in \{1,\ldots,n\}$, $v_j = v'_j$.  In both cases $v'
  \preorder v$.
\end{proof}

As a consequence, the result of Lemma~\ref{prop:buchi-nphard} applies
in particular to the conjunction and lexicographic preorders, for
which the constrained NE existence problem is thus \NP-complete. Hence we
get:

\begin{corollary}
  For finite games with ordered B\"uchi objectives with either of the
  conjunction or the lexicographic preorders, the constrained
  NE existence problem is \NP-complete. 
\end{corollary}\section{Ordered reachability objectives}\label{sec:reach}
In this Section we assume that preference relations of the players are
given by ordered reachability objectives (as defined in
Section~\ref{sec:prefrel}), and we prove the results listed in
Table~\ref{table-reach} (page~\pageref{table-reach}).  We will first
consider the general case when preorders are given by Boolean circuits
and we will show that the various decision problems are
\PSPACE-complete. We will even notice that the hardness result holds
for several simpler preorders.  We will finally improve this result in
a number of cases.

For the rest of this section, we fix a game $\calG=\tuple{\Stat,\Agt,
  \Act,\Allow, \Tab,(\mathord\prefrel_A)_{A\in\Agt}}$, and we assume that
$\prefrel_A$ is given by an ordered reachability objective $\omega_A =
\langle (\Omega_i^A)_{1 \le i \le n_A},\allowbreak (\mathord\preorder_A)_{A \in \Agt}
\rangle$.

\subsection{General case: preorders are given as circuits}
\label{ssec-generalcase}

We prove the following result:

\begin{proposition}\label{prop:reach-pspace}\hfill
  \begin{itemize}
  \item For finite games with ordered reachability objectives where
    preorders are given by Boolean circuits, the value problem, the NE
    existence problem and
  the    constrained NE existence problem are in \PSPACE.
  \item For finite two-player turn-based games with ordered
    reachability objectives where preorders have~$\One$ as a unique
    maximal element, the value problem is \PSPACE-hard.
  \item For finite two-player games with ordered reachability
    objectives where preorders have~$\One$ as a unique maximal
    element, and have an element $v$ such that for every $v'$, $v' \ne
    \One \Leftrightarrow v' \preorder v$, then the NE existence problem and the
    constrained NE existence problem are \PSPACE-hard.
  \end{itemize}
\end{proposition}

\noindent The upper bound will be proven by reduction to games with ordered
B\"uchi objectives using game-simulation.

\subsubsection{Reduction to a game with ordered B\"uchi objectives.}
We show how to transform a game $\calG$ with preferences given by
Boolean circuits over reachability objectives into a new
game~$\calG'$, with preferences given by Boolean circuits over B\"uchi
objectives.  Although the size of~$\calG'$ will be exponential,
circuit order with B\"uchi objectives define prefix-independent
preference relations and thus checking condition~\ref{cond:win} of
Theorem~\ref{thm:eq-win} can be made more efficient.

States of $\calG'$ store the set of states of $\calG$ that have
already been visited. The set of states of~$\calG'$ is $\Stat' = \Stat
\times 2^\Stat$. The transitions are as follows: $(s,S) \rightarrow
(s',S')$ when there is a transition $s\rightarrow s'$ in~$\calG$ and
$S' = S \cup \{s'\}$. We keep the same circuits to define the
preference relations, but the reachability objectives are transformed
into B\"uchi objectives: a~target set~$T$ is transformed into $T' = \{
(s,S) \mid S \cap T \ne \varnothing\}$. Although the game has
exponential size, the preference relations only depend on the strongly
connected components the path ends~in, so~that we will be able to use
a special algorithm, which we describe after this lemma.

We define the relation $s \simulrel s'$ over states of $\calG$
and~$\calG'$ \iff $s' = (s,S)$ with $S\subseteq \Stat$, and prove that
it is a game simulation (see Definition~\ref{def-gsim}).

\begin{lemma}\label{lem:construction-reach-general}
  The relation~$\simulrel$ (resp. $\simulrel^{-1}$) is a game
  simulation between $\calG$ and~$\calG'$, and it is
  preference-preserving from $(s_0,(s_0,\{s_0\}))$
  (resp. $((s_0,\{s_0\}),s_0)$).
\end{lemma}
\begin{proof}
  Let $m_\Agt$ be a move; writing $t = \Tab(s,m_\Agt)$, we~have
  $\Tab'((s,S),m_\Agt) = (t,S\cup \{t\})$. Therefore $\Tab(s,m_\Agt)
  \simulrel \Tab'(s',m_\Agt)$.  Let $(t,S')$ be a state of~$\calG'$;
  then we also have ${t\simulrel (t,S')}$.  If ${S' = S\cup \{t\}}$
  then $\Susp((s,t),m_\Agt) = \Susp(((s,S),(t,S')),m_\Agt)$; otherwise
  ${\Susp(((s,S),(t,S')),m_\Agt) = \varnothing}$.  In both cases,
  condition~(2) in the definition of a game simulation is obviously
  satisfied.

  In the other direction, let $(s',S\cup\{s'\}) = \Tab((s,S),m_\Agt)$;
  we have that $s'\simulrel (s',S\cup\{s'\})$.  Let $t \in
  \Stat$. Then $t \simulrel (t,S\cup \{t\})$, and $\Susp((s,t),m_\Agt)
  = \Susp(((s,S),(t,S\cup\{t\})),m_\Agt)$. Hence $\simulrel^{-1}$ is a
  game simulation.
  
  \smallskip Let $\rho$ and $\rho'$ be two paths, from $s_0$ and
  $(s_0,\{s_0\})$ respectively, and such that $\rho \simulrel \rho'$.
  We show preference preservation, by showing that $\rho$ reaches
  target set~$T$ \iff $\rho'$~visits~$T'$ infinitely often.
  If~$\rho$~visits some state $s\in T$, then from that point, states
  visited by $\rho'$ are of the form $(s',S')$ with $s \in S'$; all
  these states are in $T'$, therefore $\rho'$ visits~$T'$ infinitely
  often.  Conversely, if $\rho'$ visits~$T'$ infinitely often, then
  some state of~$T'$ have been visited by~$\rho$. From this, we~easily
  obtain preference preservation.
\end{proof}

\noindent As a corollary (Proposition~\ref{prop:sim}) we get that there is a
correspondence between Nash equilibria in $\calG$ and Nash equilibria
in $\calG'$.

\begin{lemma}
  If there is a Nash equilibrium $\sigma_\Agt$ in $\calG$ from $s_0$,
  then there is a Nash equilibrium $\sigma'_\Agt$ in $\calG'$ from
  $(s_0,\{s_0\})$ such that $\Out_{\calG}(s_0,\sigma_\Agt) \simulrel
  \Out_{\calG'}((s_0,\{s_0\}),\sigma'_\Agt)$.  And vice-versa: if
  there is a Nash equilibrium $\sigma'_\Agt$ in $\calG'$ from
  $(s_0,\{s_0\})$, then there is a Nash equilibrium $\sigma_\Agt$ in
  $\calG$ from $s_0$ such that
  $\Out_{\calG'}((s_0,\{s_0\}),\sigma'_\Agt) \simulrel^{-1}
  \Out_{\calG}(s_0,\sigma_\Agt)$.
\end{lemma}
Note that, if $\Out_{\calG}(s_0,\sigma_\Agt) \simulrel
\Out_{\calG'}((s_0,\{s_0\}),\sigma'_\Agt)$, then
$\Out_{\calG}(s_0,\sigma_\Agt)$ satisfies the reachability objective
with target set $T$ \iff $\Out_{\calG'}((s_0,\{s_0\}),\sigma'_\Agt)$
satisfies the B\"uchi objective with target set $T' = \{(s,S) \mid S
\cap T \ne \emptyset\}$.  From this strong correspondence between
$\calG$ and $\calG'$, we get that it is sufficient to look for Nash
equilibria in game $\calG'$.

\subsubsection{How to efficiently solve the suspect game of $\calG'$}
In game $\calG'$, preference relations are
prefix-independent. Applying Remark~\ref{rem:prefix-independant} the
preference relation in the suspect game is then also
prefix-independent, and the payoff of a play only depends on which
strongly-connected component the path ends in.  We now give an
alternating algorithm which runs in polynomial time and solves the
game~$\calH(\calG',\pi')$, where $\pi'$ is an infinite path in
$\calG'$.

\begin{lemma}\label{lem:alternating-algo}
  The winner of $\calH(\calG',\pi')$ can be decided by an alternating
  algorithm which runs in time polynomial in the size of
  $\calG$.
\end{lemma}
\begin{proof}
  Let $C^A$ be the circuit defining the preference relation of player
  $A$.  Let $\rho = (s_i,S_i)_{i\ge 0}$ be a path in $\calG'$, the
  sequence $(S_i)_{i \ge 0}$ is non-decreasing and converges to a
  limit $S(\rho)$.  We have $\payoff_A(\rho) = \One_{\{i \mid T_A^i
    \cap S(\rho) = \varnothing\}}$.  Therefore the winning condition
  of $\Eve$ in $\calH(\calG',\pi')$ for a play $\rho$ only depends on
  the limits~$\limitpi2(\rho)$ and~$S(\Sproj_1(\rho))$.  It can be
  described as a single B\"uchi condition with target set $T
  = \{ ((s,S),P) \mid \forall A\in P.\ C^A[ v^A(S) , w^A ] \
  \text{evaluates to \true} \}$ where $v^A(S) = \One_{\{i \mid T_A^i
    \cap S = \varnothing\}}$ and $w^A = \payoff_A(\pi')$.  We now
  describe the algorithm.

  Initially the current state is set to $((s_0,\{s_0\}),\Agt)$. We
  also keep a list of the states which have been visited, and we
  initialise it with $\Occ \leftarrow \{ (s_0,\{s_0\}),\Agt \}$.
  Then,
  \begin{itemize}
  \item if the current state is $((s,S),P)$, the algorithm
    existentially guesses a move~$m_\Agt$ of \Eve and we set $t =
    ((s,S),P,m_\Agt)$;
  \item otherwise if the current state is of the form
    $((s,S),P,m_\Agt)$, it universally guesses a state $s'$ which
    corresponds to a move of \Adam and we set $t = ((s',S\cup
    \{s'\}),P\cap \Susp((s,s'),m_\Agt))$.
  \end{itemize}
  If $t$ was already seen (that is, if $t \in \Occ$), the algorithm
  returns $\true$ when $t\in T$ and $\false$ when $t \notin T$,
  otherwise the current state is set to $t$, and we add $t$ to the
  list of visited states: $\Occ \leftarrow \Occ \cup \{t\}$, and we
  repeat this step.  Because we stop when the same state is seen, the
  algorithm stops after at most $\ell+1$ steps, where $\ell$ is the
  length of the longest acyclic path.  Since the size of~$S$ can only
  increase and the size of~$P$ only decrease, we~bound~$\ell$ with
  $|\Stat|^2 \cdot |\Agt|$.
 
  We now prove the correctness of the algorithm.  First,
  $\calH(\calG',\pi')$ is a turn-based B\"uchi game, which is a
  special case of parity game.  Parity games are known to be
  determined with memoryless
  strategies~\cite{mostowski1991games,emerson1991tree}, hence
  $\calH(\calG',\pi')$ is determined with memoryless strategies.
 
  If the algorithm returns \true, then there exist a strategy
  $\sigma_\shortEve$ of \Eve such that for all the strategies
  $\sigma_\shortAdam$ of \Adam, any outcome $\rho$ of
  $\Out(\sigma_\shortEve,\sigma_\shortAdam)$ is such that there exist
  $i < j \leq \ell + 1$ with $\rho_i = \rho_j \in T$ and all $\rho_k$
  with $k < j$ are different.  We extend this strategy
  $\sigma_\shortEve$ to a winning strategy $\sigma'_\shortEve$ for
  \Eve.  We do so by ignoring the loops we see in the history,
  formally we inductively define a reduction $r$ of histories by:
  \begin{itemize}
  \item $r(\varepsilon) = \varepsilon$;
  \item if $((s,S),P)$ does not appear in $r(h)$ then $r(h \cdot
    ((s,S),P)) = r(h) \cdot ((s,S),P)$; 
  \item otherwise $r(h \cdot ((s,S),P)) = r(h)_{\le i}$  where $i$ is
    the smallest index such that $r(h)_i = ((s,S),P)$.
  \end{itemize}
  We then define $\sigma'_\shortEve$ for any history $h$ by
  $\sigma'_\shortEve(h) = \sigma_\shortEve(r(h))$.
  
  We show by induction that if $h$ is a history compatible with
  $\sigma'_\shortEve$ from $((s_0,\{s_0\}),\Agt)$ then $r(h)$ is
  compatible with $\sigma_\shortEve$ from $((s_0,\{s_0\}),\Agt)$ .  It
  is true when $h= ((s_0,\{s_0\}),\Agt)$, now assuming it holds for
  all history of length $\le k$, we show it for history of length
  $k+1$.  Let $h\cdot s$ be a history of length $k+1$ compatible with
  $\sigma'_\shortEve$.  By hypothesis $r(h)$ is compatible with $h$
  and since $\sigma'_\shortEve (h)= \sigma_\shortEve(r(h))$,
  $r(h)\cdot s$ is compatible with $\sigma_\shortEve$.  If $r(h\cdot
  s) = r(h)\cdot s$ then $r(h\cdot s)$ is compatible with
  $\sigma_\shortEve$.  Otherwise $r(h\cdot s)$ is a prefix of $r(h)$
  and therefore of length $\le k$, we can apply the induction
  hypothesis to conclude that $r(h\cdot s)$ is compatible with
  $\sigma_\shortEve$.
  
  We now show that the strategy $\sigma'_\shortEve$ that we defined,
  is winning.  Let $\rho$ be a possible outcome of
  $\sigma'_\shortEve$, let $i<j$ be the first indexes such that
  $\rho_i,\rho_j \in (\Stat\times S(\rho)) \times \limitpi2(\rho)$ and
  $\rho_i = \rho_j$.  Because there is no repetition between $i$ and
  $j-1$: $r(\rho_{\le j-1}) = r(\rho_{\le i -1}) \rho_{i} \cdots
  \rho_{j-1}$.  We have that $\sigma_\shortEve(r(\rho_{\le i -1})
  \rho_{i} \cdots \rho_{j-1}) = \sigma'_\shortEve(\rho_{j-1})$.  From
  this move, $\rho_j$ is a possible next state, so $r(\rho_{\le i -
    1}) \rho_{i}\cdots \rho_{j}$ is a possible outcome of
  $\sigma_\shortEve$.  As $\rho_{i} = \rho_{j}$ and all other states
  are different, by the hypothesis on $\sigma_\shortEve$ we have that
  $\rho_j \in T$.  This shows that $\rho$ ultimately loops in states
  of $T$ and therefore $\rho$ is a winning run for \Eve.
 
  Reciprocally, if \Eve has a winning strategy, she has a memoryless
  one~$\sigma_\shortEve$ since this is a B\"uchi game. 
  We can see this strategy as an oracle for the various
  existential choices in the algorithm.
  Consider some universal choices in the algorithm, it corresponds to
  a strategy $\sigma_\shortAdam$ for \Adam. The branch corresponding to
  $(\sigma_\shortEve,\sigma_\shortAdam)$ ends the first time we
  encounter a loop, we write this history $h\cdot h'$ with $\last(h')
  = \last(h)$.  Since the strategy $\sigma_\shortEve$ is memoryless,
  $h\cdot h'^\omega$ is a possible outcome.  Since it is winning,
  $\last(h')$ is in $T$ and therefore the branch is accepting.  This
  being true for all the branches given by the choices of
  $\sigma_\shortEve$, the algorithm answers \true.
\end{proof}

\subsubsection{Proof of the \PSPACE\ upper bounds in
Proposition~\ref{prop:reach-pspace}.}
We describe a \PSPACE~algorithm for solving the constrained NE existence
problem.  The algorithm proceeds by trying all plays~$\pi$ in $\calG$
of the form described in Proposition~\ref{lem:play-length}.  This
corresponds to a (unique) play $\pi'$ in $\calG'$.  We check that
$\pi'$ has a payoff satisfying the constraints, and that there is a
path $\rho$ in $\calH(\calG',\pi')$, whose projection is $\pi'$, along
which \Adam obeys \Eve, and which stays in the winning region of \Eve.
This last step is done by using the algorithm of
Lemma~\ref{lem:alternating-algo} on each state $\rho$ goes through.
All these conditions are satisfied exactly when the conditions of
Theorem~\ref{thm:eq-win} are satisfied, in which case there is a Nash
equilibrium within the given bounds.

The \PSPACE\ upper bound for the value problem can be inferred from
Proposition~\ref{lem:link-value-constr}.

\subsubsection{Proof of \PSPACE-hardness for the value problem.}

We show \PSPACE-hardness of the value problem when the preorder has
$\One$ as a unique maximal element.

We reduce \QSAT to the value problem, where \QSAT is the
satisfiability problem for quantified Boolean formulas. For an instance
of \QSAT, we assume without loss of generality that the Boolean
formula is a conjunction of disjunctive clauses\footnote{With the
  convention that an empty disjunction is equivalent to~$\bot$.}.

Let $\phi = Q_1 x_1 \dots Q_p x_p.\ \phi'$, where $Q_i \in \{ \forall,
\exists \}$ and $\phi'= c_1 \land \dots \land c_n$ with $c_i =
\bigvee_{1\leq j\leq 3} \ell_{i,j}$ and $\ell_{i,j} \in \{ x_k ,
\lnot x_k \mid 1 \le k \le p\}\cup\{\top,\bot\}$.  We~define a
turn-based game~$\calG(\phi)$ in the following way (illustrated in
Example~\ref{ex:reduction} below).  There~is one state for each
quantifier, one for each literal, and two additional states~$\top$
and~$\bot$:
\[
\Stat = \{ Q_k \mid 1 \le k \le p\} \cup \{ x_k , \lnot x_k \mid 
1 \le k\le p \} \cup \{ \top,\bot\}.
\]

The game involves two players, $A$~and~$B$. The~states~$\top$, 
and~$\bot$,
the existential-quantifier states and the literal states are all controlled
by~$A$, while the universal-quantifier states belong to player~$B$. For
all~$1\leq k\leq p$, the state corresponding to quantifier~$Q_k$ has two
outgoing transitions, going to $x_k$ and~$\lnot x_k$ respectively. Those two
literal states only have one transition to the next quantifier
state~$Q_{k+1}$, or to the final state~$\top$ if $k=p$. Finally,
states~$\top$ and~$\bot$ carries a self-loop
(notice that~$\bot$ is not reachable, while $\top$ will always be
visited).

Player~$A$ has one target set for each clause: if $c_i =
\bigwedge_{1\leq j\leq 3} \ell_{i,j}$ then $T^A_i = \{\ell_{i,j} \mid
1\leq j\leq 3\}$.  The $i$-th objective~$\Omega_i^A$ is to reach
target set~$T^A_i$.  The following result is then straightforward:
\begin{lemma}\label{lem:formula-valid}
  Formula $\phi$ is valid \iff player~$A$ has a strategy whose
  outcomes from state~$Q_1$ all visit each target set~$T^A_i$.
\end{lemma}
\begin{proof}  
  We begin with the direct implication, by induction on~$p$.  For the
  base case, $\phi=Q_1 x_1.\ \bigwedge_i c_i$ where $c_i$ only
  involves~$x_1$ and~$\non x_1$. We~consider two cases:
  \begin{itemize}
  \item $Q_1=\exists$: since we assume~$\phi$ be true, there must
    exist a value for~$x_1$ which makes all clauses true. If this
    value is~$\top$, consider the strategy~$\sigma_\top$ of Player~$A$
    such that $\sigma_\top(Q_1)=x_1$.  Then each clause~$c_i$ must
    have~$x_1$ as one of its literals, so that the
    objective~$\Omega^A_i$ is satisfied with this strategy.  The same
    argument applies if the value for $x_1$ were~$\bot$.
  \item $Q_1=\forall$: in that case, Player~$A$ has only one strategy.
    For both $x_1$ and $\non x_1$ all the clauses are
    satisfied. It~follows that each clause~$c_i$ must contain $x_1$
    and $\non x_1$, so that objective~$\Omega^A_i$ is satisfied for
    any strategy of player~$B$.
  \end{itemize}
  
\noindent   Now, assume that the result holds for all \QSAT instances with at
  most $p-1$ quantifiers.
  \begin{itemize}
  \item if $Q_1 = \exists$, then one of $Q_2x_2\ldots Q_p x_p
    \phi'[x_1 \leftarrow \top]$ and $Q_2x_2\ldots Q_p x_p\phi'[x_1
    \leftarrow \bot]$ is valid. We~handle the first case, the second
    one being symmetric. For a literal~$\lambda_k\in\{x_k,\non x_k\}$,
    we~write 
    $T_{\lambda_k}$ for the set of target sets~$T_i^A$ such that the
    clause~$c_i$ contains the literal~$\lambda_k$.

    Assume $Q_2x_2\ldots Q_p x_p \phi'[x_1 \leftarrow \top]$ is valid;
    by~induction we know that there exists a strategy~$\sigma^{x_1}$
    such that all the targets in~$T_{\lambda_k}$ are visited along any
    outcome from state~$Q_2$ (because $\calG(Q_2x_2\ldots Q_p x_p
    \phi'[x_1 \leftarrow \top])$ is the same game as~$\calG(\phi)$,
    but with~$Q_2$ as the initial state, and with the targets
    in~$T_{x_1}$ containing~$\{\top\}$ in place of~$x_1$).  We~define
    the strategy~$\sigma$ by $\sigma(Q_1)=x_1$ and $\sigma(Q_1 \cdot
    x_1 \cdot \pathg) = \sigma^{x_1}(\pathg)$.  An~outcome of~$\sigma$
    will necessarily visit~$x_1$, hence visiting all the targets
    in~$T_{x_1}$; because $\sigma$ follows $\sigma^{x_1}$, all the
    objectives not in~$T_{x_1}$ are~met as~well.
  \item if $Q_1 = \forall$, then $Q_2x_2\ldots Q_p x_p \phi'[x_1
    \leftarrow \top]$ is valid.  Using the induction hypothesis we
    know that from~$Q_2$ there is a strategy $\sigma^{x_1}$ that
    enforces a visit to all the targets in~$T_{x_1}$.  Similarly,
    $Q_2x_2\ldots Q_p x_p \phi'[x_1 \leftarrow \bot]$ is valid, and
    there is a strategy~$\sigma^{\non x_1}$ that visits all the
    objectives not in~$T_{\non x_1}$.  We~define a new
    strategy~$\sigma$~as follows: $\sigma(Q_1 \cdot x_1 \cdot \pathg)
    = \sigma^{x_1}(\pathg)$ and $\sigma(Q_1 \cdot \non x_1 \cdot
    \pathg) = \sigma^{\non x_1}(\pathg)$. Consider an outcome
    of~$\sigma$: if~it visits~$x_1$, then all the objectives
    in~$T_{x_1}$ are visited, and because the path
    follows~$\sigma^{x_1}$, the objectives not in~$T_{x_1}$ are also
    visited. The other case is similar.
  \end{itemize}

  \medskip\noindent  We now turn to the converse implication.  Assume the
  formula is not valid. We prove that for any strategy~$\sigma$ of
  player~$A$, there is an outcome~$\pathg$ of this strategy such that
  some objective~$\Omega^A_i$ is not satisfied.  We~again proceed by
  induction, beginning with the case where~$n=1$.
  \begin{itemize}
  \item if $Q_1 = \exists$, then both $\phi'[x_1 \leftarrow \top]$ and
    $\phi'[x_1 \leftarrow \bot]$ are false.
    This entails that one of the clauses only
    involves~$\bot$ (no~other disjunction involving~$x_1$ and\slash
    or~$\non x_1$ is always false), and the corresponding
    reachability condition is~$\bot$, which is not reachable.
  \item if $Q_1=\forall$, then one of $\phi'[x_1 \leftarrow \top]$
    and $\phi'[x_1 \leftarrow \bot]$ is false.
    In~the former case, one of the clauses~$c_i$
    contains~$\non x_1$, or only contains~$\bot$.  Then along the run
    $Q_1 \cdot x_1 \cdot \top^\omega$, 
    the objective~$T^A_i$ is not
    visited.  The other case is similar.
  \end{itemize}

  Now, assuming that the result holds for formulas with~$n-1$
  quantifiers, we prove the result with~$n$ quantifiers.
  \begin{itemize}
  \item if $Q_1 = \exists$, then both $Q_2x_2\ldots Q_p x_p \phi'[x_1
    \leftarrow \top]$ and $Q_2x_2\ldots Q_p x_p \phi'[x_1 \leftarrow
    \bot]$ are false.  Ising the induction
    hypothesis, any run from~$Q_2$ fails to visit some objective not
    in $T_{x_1} \cup T_{\non x_1}$. Hence no strategy from~$Q_1$ can
    enforce a visit to all the objectives.
  \item if $Q_1 = \forall$, then one of $Q_2x_2\ldots Q_p x_p \phi'[x_1
    \leftarrow \top]$ and $Q_2x_2\ldots Q_p x_p \phi'[x_1 \leftarrow
    \bot]$ is false.  
   We handle the first case, the second one being symmetric.
    By induction hypothesis, for any strategy~$\sigma$ of player~$A$
    in the game $\calG(\phi'[x_1 \leftarrow \top])$, one of the
    outcome fails to visit all the objective not in~$T_{x_1}$.  Then
    along the path $\pathg = Q_1 \cdot x_1 \cdot \pathg'$, some
    objectives not in $T_{x_1}$ are not visited.\forceqed
  \end{itemize} 
\end{proof}

\noindent We can directly conclude from this lemma that the value of the game for $A$ is
$\One$ (the~unique maximal payoff for our preorder) \iff the formula~$\phi$ is
valid, which proves that the former problem is \PSPACE-hard.

\begin{example}
  \label{ex:reduction}
  As an example of the construction, let us consider the formula
  \begin{equation}\label{eq:formula1}
    \phi = \forall x_1.\ \exists
    x_2.\ \forall x_3.\ \exists x_4.\ (x_1 \lor \lnot x_2 \lor \lnot x_3)
    \land (x_1 \lor x_2 \lor x_4) \land (\lnot x_4 \lor\bot \lor\bot)
  \end{equation}
  The target sets for player~$A$ are given by $T^A_1 = \{x_1; \lnot
  x_2; \lnot x_3\}$, $T^A_2 = \{ x_1; x_2; x_4\}$, and $T^A_3 =
  \{\lnot x_4; \bot\}$.  The structure of the game is represented in
  Figure~\ref{fig:reach-game}.  $B$ has a strategy that falsifies one
  of the clauses whatever $A$ does, which means that the formula is
  not valid.
\end{example}

  \begin{figure}[!ht]
    \centering{
      \begin{tikzpicture}[thick]
        \tikzstyle{carre}=[draw,minimum width=6mm,minimum height=6mm,inner sep=0pt,fill=black!10]
        \tikzstyle{rond}=[draw,minimum width=7mm,circle,inner sep=0pt,fill=black!10]
        \draw (-3.5,1.2) node [rond] {} node [right=.5cm] {player $A$};
        \draw (-3.5,.3) node [carre] {} node [right=.5cm] {player $B$};
        \draw (0,0) node [carre] (A1) {$\forall_1$}; 
        \draw (1,1) node [rond] (B1) {$x_1$};
        \draw (1,-1) node [rond] (C1) {$\lnot x_1$}; 
        \draw (2,0) node [rond] (A2) {$\exists_2$};
        \draw (3,1) node [rond](B2){$x_2$};
        \draw (3,-1) node [rond](C2){$\lnot x_2$};
      \draw (4,0) node [carre] (A3) {$\forall_3$}; 
      \draw (5,1) node [rond] (B3) {$x_3$};
      \draw (5,-1) node [rond] (C3){$\lnot x_3$};
      \draw (6,0) node [rond] (A4) {$\exists_4$};
      \draw (7,1) node [rond] (B4) {$x_4$};
      \draw (7,-1) node [rond] (C4){$\lnot x_4$};
      \draw (8,0) node [rond] (A5){$\top$};

      \draw [-latex'] (A1) -- (B1);
      \draw [-latex'] (A1) -- (C1);
      \draw [-latex'] (B1) -- (A2);
      \draw [-latex'] (C1) -- (A2);
      \draw [-latex'] (A2) -- (B2);
      \draw [-latex'] (A2) -- (C2);

      \draw [-latex'] (B2) -- (A3);
      \draw [-latex'] (C2) -- (A3);

      \draw [-latex'] (A3) -- (B3);
      \draw [-latex'] (A3) -- (C3);
      \draw [-latex'] (B3) -- (A4);
      \draw [-latex'] (C3) -- (A4);
      \draw [-latex'] (A4) -- (B4);
      \draw [-latex'] (A4) -- (C4);
      \draw [-latex'] (B4) -- (A5);
      \draw [-latex'] (C4) -- (A5);

      \draw[-latex'] (A5) .. controls +(0.8,-0.5) and +(0.8,0.5) .. (A5);
      \end{tikzpicture}
    }
    \caption{Reachability game associated with the
      formula~\eqref{eq:formula1}}
    \label{fig:reach-game}
\end{figure}
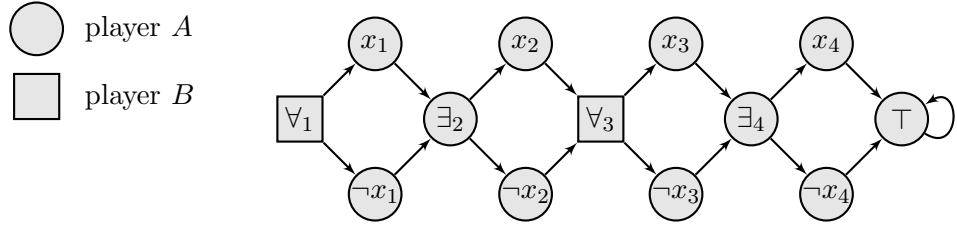

\subsubsection{Proof of \PSPACE-hardness for the (constrained) NE existence problem.}

We will now prove \PSPACE-hardness for the NE existence problem, under
the conditions specified in the statement of
Proposition~\ref{prop:reach-pspace}, using
Proposition~\ref{lem:link-value-exist}.  We specify the new preference
relation for the construction of Section~\ref{sec:link-value-exist}.
We give $B$ one objective, which is to reach $s_1$ ($s_1$ is the sink
state introduced by the construction).  In terms of preferences for
$A$, going $s_1$ should be just below visiting all targets.  For this
we use the statement in Proposition~\ref{prop:reach-pspace}, that
there is $v$ such that for every $v'$, $v' \ne \One \Leftrightarrow v'
\preorder v$, and add $s_1$ as a target to each $T^A_i$ such that $v_i
= 1$.  This defines a preference relation equivalent to the one in the
game constructed in Section~\ref{sec:link-value-exist}, therefore we
deduce with Proposition~\ref{lem:link-value-exist} that the NE existence
problem is \PSPACE-hard.

\subsubsection{Applications}
\label{subsec:boo-reach}
We should now notice that conjunction, counting and lexicographic
preorders (thanks to the fact that $\One$ is the unique maximal
element for theses orders and to Lemma~\ref{lem:second-maximal}).
As conjunction (for instance) can easily be encoded using a
(monotonic) Boolean circuit in polynomial time, the hardness results
are also valid if the preorder is given by a (monotonic) Boolean
circuit. Finally the subset preorder can be expressed as a
polynomial-size Boolean circuit and has a maximal element. We
therefore get the following summary of results:
\begin{corollary}\hfill
  \begin{itemize}
  \item For finite games with ordered reachability objectives, with either the
    conjunction, the counting or the lexicographic preorder, the value
    problem, the NE existence problem and the constrained NE existence problem
    are \PSPACE-complete.
  \item For finite games with ordered reachability objectives, where the
    preorders are given by (monotonic) Boolean circuits, the value problem,
    the NE existence problem and the constrained NE existence problem are
    \PSPACE-complete.
  \item For finite games with ordered reachability objectives, with
    the subset preorder, the value problem is \PSPACE-complete.
  \end{itemize}
\end{corollary}

\noindent On~the other hand, the disjunction and maximise preorders do not have
a unique maximal element, so the hardness result does not carry over
to these preorders.  In~the same way, for the subset preorder, there
is no~$v$ such that $v' \ne \One \Leftrightarrow v' \preorder v$, so
the hardness result does not apply.  We prove later (in
Section~\ref{reach-simple}) that in these special cases, the
complexity is actually lower.

\subsection{Simple cases}
\label{reach-simple}

As for ordered B\"uchi objectives, for some ordered reachability
objectives, the preference relation can be (efficiently) (co-)reduced
to a single reachability objective. We do not give the formal
definitions, they can easily be inferred from that for B\"uchi
objectives on page~\pageref{subsec:reducible}.

\begin{proposition}\label{prop:simple-reach-p}\label{prop:reach-max-np}\hfill
  \begin{itemize}
  \item For finite games with ordered reachability objectives which
    are reducible to single reachability objectives and in which the
    preorders are non-trivial, the value problem is \P-complete.
  \item For finite games with ordered reachability objectives which
    are co-reducible to single reachability objectives, and and in
    which the preorders are non-trivial, the NE existence problem and the
    constrained NE
    existence problem are \NP-complete.
  \end{itemize}
\end{proposition}

\begin{proof}
  Since \P-hardness (resp. \NP-hardness) already holds for the value
  (resp. NE existence) problem with a single reachability objective
  (see~\cite[Sect.~2.5.1]{GTW02}), we only focus on the upper bounds.

  We begin with the value problem: given a payoff vector~$u$ for
  player~$A$, we~build the new target set~$\widehat T$ in polynomial
  time, and then use a classical algorithm for deciding whether
  $A$~has a winning strategy (see \cite[Sect.~2.5.1]{GTW02}).  If
  she~does, then she can secure payoff~$u$.

  Consider now the constrained NE existence problem, and assume that the
  preference relation for each player~$A$ is given by target
  sets~$(T_i^A)_{1\leq i\leq n_A}$.  The \NP-algorithm consists in
  guessing the payoff vector~$(v_A)_{A\in\Agt}$ and an ultimately
  periodic play~$\rho = \pi \cdot \tau^\omega$ with $|\pi|,|\tau| \le
  |\Stat|^2$, which, for each~$A$, visits $T_i^A$ \iff
  $v^A_i = 1$.  We then co-reduce the payoff to a new target
  set~$\widehat T^A(v^A)$ for each player~$A$.

  The run~$\rho$ is the outcome of a Nash equilibrium with
  payoff~$(v_A)_{A \in \Agt}$ for the original preference relation
  \iff $\rho$~is the outcome of a Nash equilibrium with payoff~$0$
  with the single reachability objective~$\widehat T^A(v^A)$ for
  each~$A\in\Agt$.  Indeed, in~both cases, this is equivalent to the
  property that no player~$A$ can enforce a payoff greater than~$v^A$.
  Applying the algorithm presented in
  Section~\ref{subsec:reachability}.  this condition can be checked in
  polynomial time.
\end{proof}

We now see to which ordered objectives this result applies. It is
not difficult to realise that the same transformations as those made
in the proof of Lemma~\ref{lemma:examples} can be made as well for
reachability objectives. We therefore get the following lemma, from
which we get the remaining results in Table~\ref{table-reach}.

\begin{lemma}
  Ordered reachability objectives with disjunction or maximise
  preorders are reducible to single reachability objectives. Ordered
  reachability objectives with disjunction, maximise or subset
  preorders are co-reducible to single reachability objectives.
\end{lemma}

We conclude with stating the following corollary:
\begin{corollary}\hfill
  \begin{itemize}
  \item For finite games with ordered reachability objectives, with
    either the disjunction or the maximise preorder, the value problem
    is \P-complete.
  \item For finite games with ordered reachability objectives, with either the
    disjunction, the maximise or the subset preorder, the NE existence problem
    and the constrained NE existence problem are \NP-complete.
  \end{itemize}
\end{corollary}

\section{Conclusion}

\paragraph{Summary and impact of the results}
Concurrent games are a natural class of games, extending classical
turn-based games with more complex interactions.  We have developed a
complete methodology, involving new techniques, for computing pure
Nash equilibria in this class of games. We were able to characterise
the complexity of finding Nash equilibria (possibly with constraints
on the payoff) for simple qualitative objectives first
(Section~\ref{sec:single}), and then for semi-quantitative objectives
(Section~\ref{sec:buchi} and \ref{sec:reach}). We would like to point
out that the algorithm for B\"uchi objectives with maximise preorder
(see Section~\ref{subsec:reducible}) has been implemented in tool
Praline\footnote{Available on
  \url{http://www.lsv.ens-cachan.fr/Software/praline/}}~\cite{brenguier13b}

We believe the methodology we have developed in this paper can be used
in many other contexts, and the suspect game is a very powerful tool
that will allow to analyze various properties of multi-agent
systems. Indeed, the correspondence between pure Nash equilibria in
the original game and winning strategies in the suspect game holds
with no assumption on the structure of the game. In particular it can
be applied to games given as pushdown systems, counter systems, \etc
Also it does not assume anything on the preference relations, only the
resulting winning condition in the suspect game can become very
complex if the preference relations are complex. Now the matter is
just algorithmics, in that we have to solve a two-player turn-based
game in a potentially complex arena (if the original game structure is
complex) with a potentially complex winning condition (if the
preference relations are complex).

The suspect game construction can also be adapted to compute many
other kinds of equilibria; this is for instance applied to robust
equilibria in~\cite{brenguier13}. We believe this can be used in many
other contexts.

We have also developed in this paper another tool that might have its
own interest and be useful in some other contexts: the game-simulation
(see Section~\ref{sec:simulation}). We used this tool several times
(for handling objectives given by deterministic Rabin automata, but
also for handling ordered reachability objectives). This tool can also
be used to handle more complex game structures, like we did
in~\cite{BBMU11} for timed games, when we originally introduced this
notion. In particular, the construction done in~\cite{BBMU11} shows
that we can compute Nash equilibria for timed games with all kinds of
objectives studied in the current paper.

Our future researches will include extending the use of the
suspect game abstraction for other families of games, and to push it
further to also handle truly quantitative objectives.

\paragraph{Discussion on the various hypotheses made in this paper}
We have assumed strategies are pure, and game structures are
deterministic. This is indeed a restriction, and allowing for
randomised strategies would be of great interest. Note however that
pure Nash equilibria are resistant to malicious randomised players
(that is, to deviations by randomised strategies).  There is no
obvious way to modify the suspect game construction to handle either
stochastic game structures or randomised strategies. Indeed, given a
history, it is hard to detect strategy deviations if they can be
randomised, and therefore the set of suspects is hard to compute (and
actually even define).  This difficulty is non-surprising, since the
existence of a Nash equilibrium in pure or randomised strategies is
undecidable for stochastic games with reachability or B\"uchi
objectives~\cite{UW11}, and the existence of a Nash equilibrium in
randomised strategies is undecidable for deterministic
games~\cite{UW11a}. However we would like to exhibit subclasses of
stochastic games for which we can synthesize randomised Nash
equilibria, this is part of our research programme.

We have assumed that strategies are based on histories that only
record states which have been visited, and not actions which have been
played. We believe this is more relevant in the context of distributed
systems, where only the effect of an action might be visible to other
players. Furthermore, this framework is more general than the one
where every player could see the actions of the other players, since
the latter can easily be encoded in the former. In the context of
complete information (precise view of the actions), computing Nash
equilibria is rather easy since, once a player has deviated from the
equilibrium, all the other players know it and can make a coalition
against that player. To illustrate that simplification, we only
mention that the constrained NE existence problem falls in \NP for finite
games with single parity objectives (we can obtain this bound based on
the suspect game construction), if we assume that strategies can
observe actions, whereas the problem is $\P^\NP_\parallel$-hard if
strategies do not observe the actions.

Finally we have chosen the framework of concurrent games, and not that of
turn-based games as is often the case in the literature. Concurrent games
naturally appear when studying timed games~\cite{BBMU11} (the~semantics of a
timed game is that of a concurrent game, and the abstraction based on regions
that is correct for timed games is concurrent), and in the context of
distributed systems, concurrent moves are also very natural. In~fact
turn-based games are even a simpler case of concurrent games when we assume
strategies can see the actions. Of~course, the~suspect game construction
applies to turn-based games, but becomes quite simple (as~is the case if
strategies do see actions), since the set of suspect players is either the
set~$\Agt$ of all players (this is the case as long as no player has deviated from the
equilibrium), or reduces to a singleton, as~soon as a player has deviated.
To~illustrate this simplification, we notice that in the turn-based finite
games, the constrained NE existence problem is \NP-complete for single parity
objectives~\cite{ummels08} (it is $\P^\NP_\parallel$-complete in finite
concurrent games).

\bigskip\noindent \textbf{Acknowledgment.} 
We would like to thank the reviewers for their numerous comments and remarks,
which helped us improve the presentation of this paper.

\bibliographystyle{myplain}
\bibliography{newbib}

\clearpage
\appendix

\section*{Appendix: Proof of Proposition~\ref{proposition:explosion}}
\label{app}

We show \PSPACE-hardness of the constrained existence of a Nash
equilibrium for various kinds of qualitative objectives, using an
encoding of the satisfiability of a \QSAT formula $\psi = \forall x_1.\
\exists x_2.\ \dots \forall x_{p-1}.\ \exists x_p.\ \bigwedge_{1 \le i
  \le n} c_i$, where each $c_i$ is of the form $\ell_{i,1} \lor
\ell_{i,2} \lor \ell_{i,3}$ and $\ell_{i,j} \in \{ x_k , \lnot x_k
\mid 1 \le k \le p \}$.

We construct a game
$\calG_\psi=\tuple{\Stat,\Agt,\Act,\Allow,\Tab,(\mathord\prefrel_A)_{A\in\Agt}
}$ as follows:
$\Stat = \{u,w\} \cup \bigcup_{k\in \lsem 1 , p \rsem} \{ s_k , t_k,
f_k, d_k, e_k\} \cup \bigcup_{i\in \lsem 1 ,n \rsem} \{ b_i , c_i \}$;
$\Agt = \{ \Eve \} \cup \bigcup_{k\in \lsem 1 , p \rsem} \{A_k , B_k
\}$; $\Act = \{ 0 , 1 , 2\}$. We now define the transition table 
(the~structure of the game is represented in
\figurename~\ref{fig:hardness-succinct}).

\begin{figure}[h]
  \centering
  \begin{tikzpicture}
    \tikzstyle{rond}=[draw,circle,minimum size=7mm,inner sep=0mm,fill=black!10]
    \tikzstyle{oval}=[draw,minimum height=7mm,inner sep=1mm,rounded corners=2mm,fill=black!10]
    \draw (0,0) node[rond] (S1) {$s_1$};
    \draw (1,1) node[rond] (T1) {$t_1$};
    \draw (1,-1) node[rond] (F1) {$f_1$};
    \draw (2,0) node[rond] (S2) {$s_2$};
    \draw (3,1) node[rond] (T2) {$t_2$};
    \draw (3,-1) node[rond] (F2) {$f_2$};
    \draw (4,0) node[inner sep=4mm] (S3) {\dots};
    \draw (5,0) node[rond] (SN) {$s_p$};
    \draw (6,1) node[rond] (TN) {$t_p$};
    \draw (6,-1) node[rond] (FN) {$f_p$};
    \draw (8,1) node[rond] (B1) {$b_1$};
    \draw (8,0) node[rond] (B2) {$b_2$};
    \draw (8,-1) node (B3) {\dots};
    \draw (8,-2) node[rond] (BK) {$b_n$};
    \draw (9.5,1) node[rond] (C1) {$c_1$};
    \draw (9.5,0) node[rond] (C2) {$c_2$};
    \draw (9.5,-1) node (C3) {\dots};
    \draw (9.5,-2) node[rond] (CK) {$c_n$};
    \draw (11,1.8) node[oval] (D1) {$d(\ell_{1,1})$};
    \draw (11,1) node[oval] (D2) {$d(\ell_{1,2})$};
    \draw (11,0.2) node[oval] (D3) {$d(\ell_{1,3})$};
    \draw (11,-1.5) node {\dots};
    \draw (13,0) node[rond] (W) {$w$};
    
    \draw (5,2) node[rond] (U) {$u$};
    
    \draw[-latex'] (S1) -- (T1);
    \draw[-latex',dotted] (S1) -- (U);
    \draw[-latex'] (S1) -- (F1);
    \draw[-latex'] (T1) -- (S2);
    \draw[-latex'] (F1) -- (S2);
    \draw[-latex'] (S2) -- (T2);
    \draw[-latex'] (S2) -- (F2);
    \draw[dashed] (T2) -- (S3);
    \draw[dashed] (F2) -- (S3);
    \draw[-latex'] (SN) -- (TN);
    \draw[-latex'] (SN) -- (FN);
    \draw[-latex'] (TN) -- (B1);
    \draw[-latex'] (FN) -- (B1);
    \draw[-latex'] (B1) -- (C1);
    \draw[-latex'] (B1) -- (B2);
    \draw[-latex'] (B2) -- (C2);
    \draw[-latex'] (B2) -- (B3);
    \draw[-latex'] (B3) -- (C3);
    \draw[-latex'] (BK) -- (CK);
    
    \draw[-latex',dotted] (T2) -- (U);
    \draw[-latex',dotted] (F2) -- (U);
    \draw[-latex',dotted] (TN) -- (U);
    \draw[-latex',dotted] (T2) -- (U);
    \draw[-latex',dotted] (D1) -- (U);
    \draw[-latex',dotted] (D2) -- (U);
    \draw[-latex',dotted] (D3) -- (U);
    \draw[-latex'] (D1) -- (W);
    \draw[-latex'] (D2) -- (W);
    \draw[-latex'] (D3) -- (W);
    
    \draw[-latex'] (C1) -- (D1);
    \draw[-latex'] (C1) -- (D2);
    \draw[-latex'] (C1) -- (D3);
  \end{tikzpicture}
  \caption{Encoding of a \QSAT formula into a game with succinct
    representation of the transition formula.  Dotted edges correspond
    to the strategy profile that in each states selects action $0$ for
    every player.}
  \label{fig:hardness-succinct}
\end{figure}
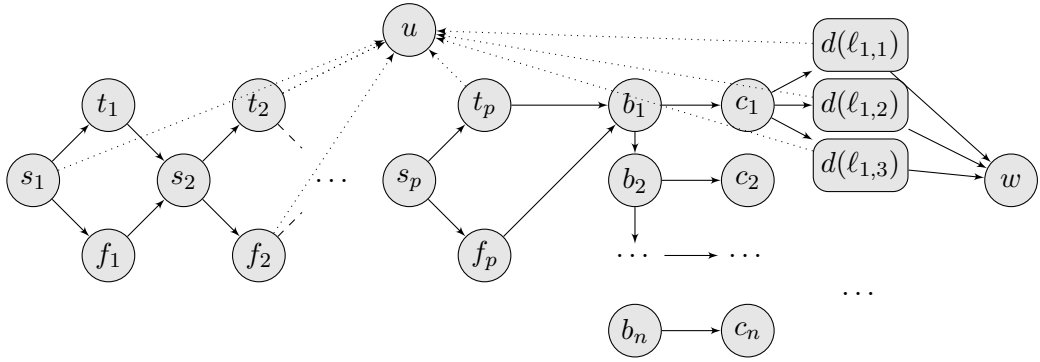

\begin{itemize}
\item If $k \le p$ is odd, then in state $s_k$, the transition
  function is given by:\footnote{The operator $\bigotimes$ evaluates
    the parity of the number of subformulas that are true:
    $\bigotimes_{h=1}^g \alpha_h$ is \true iff $|\{\alpha_h \mid
    \alpha_h\ \text{evaluates to}\ \true\}|$ is odd.}
  \begin{multline*}
  \Bigl(\bigl(\bigotimes_{1\le k' \le p} (A_{k'} = 1) \otimes 
   \bigotimes_{1\le k' \le
    p, k' \ne k} (B_{k'} = 1),t_k\bigr),\\ 
  \bigl(\bigotimes_{1\le k' \le p, k' \ne k}
  (A_{k'} = 0) \otimes \bigotimes_{1\le k' \le p} (B_{k'} = 0),f_k\bigr),\
  \bigl(\top,u\bigr)\Bigr)
  \end{multline*}
  In the first part, the coalition of all the players except \Eve and
  $B_k$ takes the decision to go to~$t_k$, and any of those players
  can switch her action and enforce state~$t_k$ (meaning that $x_k$ is
  set to \true); if the move to state $t_k$ is not chosen, then the
  coalition of all players except \Eve and $A_k$ takes the decision to
  go to~$f_k$, and any of those players can switch her action and
  enforce state $f_k$ (meaning that $x_k$ is set to \false); otherwise
  the game goes to state~$u$.

  In states~$t_k$ and $f_k$, the transition function is~$(\top,s_{k+1})$.
\item If $k \le p$ is even, then in state $s_k$ the transition
  function is given by $((\Eve = 1, t_k), (\top,f_k))$: \Eve~decides the value
  of variable~$x_k$ (state~$t_k$ corresponds to setting~$x_k$ to \true, and
  state~$f_k$ corresponds to setting variable~$x_k$ to \false).

  In state $t_k$, the transition function is given by 
  \[
  \Bigl(\bigl(\bigvee_{1 \le k' \le p} (A_{k'} = 1) \lor \bigvee_{1\le k' \le p, k'
    \ne k} (B_{k'} = 1), s_{k+1}\bigr),\ \bigl(\top, u\bigr)
  \Bigr)
  \] 
  with $s_{p+1} = b_1$: any player except \Eve and $B_k$ can decide to
  go to state $s_{k+1}$ by playing action~$1$; otherwise the game
  proceeds to state~$u$. Intuitively, any of the above players can
  ``validate'' the previous choice of \Eve having set $x_k$ to \true.

  In state $f_k$, the transition function is given by 
  \[
  \Bigl(\bigl(\bigvee_{1 \le k' \le p, k' \ne k} (A_{k'} = 0) \lor 
   \bigvee_{1\le k'  \le p} (B_{k'} = 0), s_{k+1}\bigr),\bigl(\top, u\bigr)
  \Bigr)  
  \] 
  with $s_{p+1} = b_1$: any player except \Eve and $A_k$ can decide to
  go to state $s_{k+1}$ by playing action~$1$; otherwise the game
  proceeds to state~$u$. Intuitively, any of the above players can
  ``validate'' the previous choice of \Eve having set~$x_k$ to \false.
\item If $i \le n$, in $b_i$,
  the transition function is given by
  \[
  \Bigl(\bigl(\bigotimes_{1 \le k \le p} ((A_k = 1) \otimes (B_k= 1)),
  c_i\bigr),\bigl(\top,b_{i+1}\bigr)\Bigr)
  \]
  with $b_{n+1} = u$. Intuitively the coalition of all players except
  \Eve can decide to go to state~$c_i$, which will mean that they want
  to check the truth of clause~$c_i$. Moreover, any of those players
  can switch her action and decide by her own to check this clause.
\item If $i \le n$, in~$c_i$, the transition function is given by
  \[ 
  \Bigl(\bigl(\Eve = 1, d(\ell_{i,1})\bigr) , \bigl(\Eve = 2,
  d(\ell_{i,2})\bigr) , \bigl(\top, d(\ell_{i,3})\bigr)\Bigr) 
  \]
  where for all $1 \le k \le p$, $d(x_k) = d_k$ and $d(\non x_k) =
  e_k$. Intuitively \Eve proves the current clause is satisfied by
  pointing to the literal which is set to \true.

  In state $d_k$ ($1 \le k \le p$), the transition function is given by
  $\bigl((B_k = 1,w), (\top,u)\bigr)$. Player~$B_k$ decides to go to~$u$
  or~$w$.

  In state $e_k$ ($1 \le k \le p$), the transition function is given by
  $\bigl((A_k = 1,w), (\top,u)\bigr)$.
\end{itemize}\medskip

\noindent Intuitively, in the game we have just defined, \Eve will be in charge of
properly choosing the value of the existentially quantified variables
in~$\psi$. The value of the variables will be given by the history
(visiting~$t_k$ means variable $x_k$ is set to \true, whereas visiting $f_k$
means variable~$x_k$ is set to \false). Then, player~$A_k$ will be in charge
of witnessing that variable~$x_k$ is set to
\true, whereas player $B_k$ will be in charge of witnessing that
variable~$x_k$ is set to \false. Their role will be clearer in the
proof.
  
The objective for each player $A_k$, $B_k$ is to reach state~$w$, and
for \Eve to reach state~$u$.  This is naturally a reachability
objectives but can also be encoded as a B\"uchi objective or a safety
objective where the goal is to avoid state~$u$ for $A_k$ and~$B_k$,
and avoid~$v$ for~\Eve.

\bigskip We will show that
there is a Nash equilibrium in $\calG_\psi$
where \Eve wins if, and only if, $\psi$ is valid.

Before switching to the proof of this equivalence, we define a
correspondence between (partial) valuations and histories in the
game. with~a partial valuation $v\colon\{x_1,\dots,x_k\} \to
\{\true,\false\}$, we~associate the history $\mathsf{h}(v) = s_1 w_1
s_2 w_2 \dots w_k s_{k+1}$ where for all~$1\leq k'\leq k$, $w_{k'} = t_{k'}$
(resp. $w_{k'}=f_{k'}$) if $v(x_{k'}) =\true$ (resp. $v(x_{k'})=\false$).
Conversely with every history~$h$ in~$\calG_\psi$, we~associate the partial
valuation $\mathsf{v}_h\colon\{x_1,\dots,x_k\} \to \{\true,\false\}$ such that
state $s_{k+1}$ (with $s_{p+1}=b_1$) is the latest such state appearing
along~$h$, and $\mathsf{v}_h(x_{k'}) = \true$ (resp.~$\false$) if $h$ visits
$t_{k'}$ (resp.~$f_{k'}$), for all~$1\leq k'\leq k$.

\medskip 
Assume formula $\psi$ is valid.
For all players $A_k$ and $B_k$ we set strategies $\sigma_{A_k}$ and
$\sigma_{B_k}$ to always play action~$2$. We~now turn to the strategy
for~\Eve.  Consider a history $h = s_1 \cdots w_{k-1} \cdot s_{k}$ where
$k<p$ is odd.  Let $v'$ be the valuation where $v'(x_{k'}) =
\mathsf{v}_h(x_{k'})$ for all $k'<k$, and $v'(x_k) = 1$.  We set
$\sigma_\Eve(h)$ to be~$1$ if $v'$ makes the formula $\forall
x_{k+1}.\ \dots \exists x_p.\ \bigwedge_{1\le i \le n} c_i$ valid, and~$0$
otherwise. Since $\psi$ is valid, one of the two choices makes the rest of the
formula true. This ensures that a history that reaches~$b_1$ and that is
compatible with $\sigma_\Eve$ will define a valuation that makes
$\bigwedge_{1\le i \le n} c_i$ true. Fix a history~$h$ that is compatible
with~$\sigma_\Eve$ and ends up in some state~$c_i$: the~strategy of~$\Eve$ is
to go to~$d(\ell_{i,j})$ where the literal~$\ell_{i,j}$ makes the clause~$c_i$
true under valuation $\mathsf{v}_h$. For all other histories, we set the
strategy of~\Eve to be~$2$.

We show that the strategy profile $\sigma_\Agt = (\sigma_\Eve,
(\sigma_{A_k},\sigma_{B_k})_{1 \le k \le p})$ is a Nash
equilibrium. First notice that the outcome of $\sigma_\Agt$ is $s_1
\cdot u$ (since all players $A_k$ and $B_k$ play action~$2$): 
\Eve~wins, and all the other players lose.
We now describe interesting deviating strategies for the players~$A_k$
or~$B_k$:
\begin{itemize}
\item Consider a deviating strategy $\sigma'_{A_k}$ for player $A_k$:
  let $h \in \FOut(\replaceter{\sigma}{A_k}{\sigma'_{A_k}})$; if
  $\sigma'_{A_k}(h)=2$, then $\Out(\replaceter{\sigma}{A_k}{\sigma'_{A_k}})$
  ends up in state~$u$; therefore an interesting deviating strategy should
  choose some value~$0$ or~$1$ after any history. Now if $k'$ is odd with
  $k' \ne k$, then from $s_{k'}$, player~$A_k$ can choose to go to $t_{k'}$ 
  (action~$1$) or $f_{k'}$ (action~$0$). 
  If $k$ is odd, then the only way not to end up
  in~$u$ from~$s_k$ is to choose action~$1$ which leads to state $t_k$.
  Now at state $t_{k'}$ with $k'$ even, $\sigma'_{A_k}$ should validate
  the choice of \Eve (that~is, play action~$1$ in $t_{k'}$~--~meaning
  that variable $x_{k'}$ has value \true). At~state~$f_{k'}$ with $k'$ even,
  if $k' \ne k$, $\sigma'_{A_k}$ should validate the choice of~\Eve
  (that~is, play action~$0$ in $f_{k'}$~--~meaning that variable $x_{k'}$
  has value \false). At state $f_k$ if $k$ is even, nothing can be
  done which could be profitable to player~$A_k$: state~$u$ will be
  reached.
\item A similar reasoning can be done for player $B_k$: the only
  difference is at state $s_k$ when $k$ is odd, where player $B_k$ can
  only choose action~$0$ and go through~$f_k$.
\item In the part of the game after $b_1$, each player can deviate and
  choose to go to some state~$c_i$; this choice will be made for checking
  the truth of clause $c_i$ under the valuation that has been fixed by
  the history so~far.
\end{itemize}

\noindent Consider a deviation of some player that moves to $c_i$, and write~$h$
for the corresponding history up to state~$c_i$. The strategy of~\Eve
after~$h$ is to go to~$d(\ell_{i,j})$ where $\ell_{i,j}$ sets $c_i$ to
true under valuation~$\mathsf{v}_h$.
If $d(\ell_{i,j}) = x_k$, then $(a)$~this means that $\mathsf{v}_h(x_k)
= \true$, and $(b)$~the~next state is controlled by player~$B_k$. Using the
characterization of interesting deviating strategies above, it~cannot be the
case that player $B_k$ is the deviating player since from~$t_k$ (which~is
visited by~$h$), if~only $B_k$~deviates, the game unavoidably goes to
state~$u$. Hence, for every strategy $\sigma'_{B_k}$ for player $B_k$,
history~$h$ cannot be an outcome of $\replaceter{\sigma}{B_k}{\sigma'_{B_k}}$.
In~particular, no~deviation of player~$B_k$ can lead to state~$w$. Similarly,
if $d(\ell_{i,j}) = \lnot x_m$, the~outcome ends up in~$u$. In~other words,
each time a player other than \Eve changes her strategy, the outcome ends up
in~$u$, yielding no improvement for the player.

Hence no player can improve her outcome by changing unilaterally her
strategy, which shows that the strategy profile $\sigma_\Agt$ is a
Nash equilibrium where \Eve wins.

\bigskip 
Now assume there is a Nash equilibrium~$\sigma_\Agt$ in which \Eve wins.
Let~$v$ be a valuation such that for every $2\leq k' \leq p$ even,
\begin{equation}
  \label{eq} \tag{\#}
  (\sigma_\Eve(\mathsf{h}(v_{|\{x_1,\dots,x_{k'-1}\}})) = 1)
  \Leftrightarrow (v(x_{k'}) = \true)
\end{equation}
where $v_{|\{x_1,\dots,x_{k'-1}\}}$ is the valuation$v$ restricted to
$\{x_1,\dots,x_{k'-1}\}$. We show the following two properties:
\begin{itemize}
\item if $v(x_k) = \true$ then there is a strategy $\sigma'_{A_k}$ 
  for~$A_k$ s.t. $\mathsf{h}(v) \in
  \FOut(\replaceter{\sigma}{A_k}{\sigma'_{A_k}})$;
\item if $v(x_k) = \false$ then there is a strategy $\sigma'_{B_k}$
  for~$B_k$ s.t. $\mathsf{h}(v) \in
  \FOut(\replaceter{\sigma}{B_k}{\sigma'_{B_k}})$.
\end{itemize}

We show the result by induction on the number of atomic
propositions. For zero atomic propositions, the result obviously
holds. Assume the result holds for atomic propositions
$\{x_1,\dots,x_{h-1}\}$ ($h \le p$). Let $v$ be a valuation over $\{
x_1, \dots ,x_{h}\}$, and $k$ such that $v(x_k)$ is true. Define~$v'$ as the
restriction of~$v$ to atomic propositions $\{ x_1, \dots ,x_{h-1}\}$.
By~induction hypothesis, $\mathsf{h}(v') = s_1 \cdot w_1 \cdots w_{h-1}
\cdot s_{h}$ is an outcome of some strategy $\sigma'_{A_k}$ for player~$A_k$.
\begin{itemize}
\item If $h$ is odd.  Let $m_\Agt = \sigma_\Agt(\mathsf{h}(v'))$.  We
  set $\sigma'_{A_k}(\mathsf{h}(v'))$ to be~$1$ if $\bigotimes_{1\le k'
    \le p, k'\le k} (m_{A_{k'}} = 1) \otimes \bigotimes_{1\le k' \le p, k'
    \ne k} (m_{B_{k'}} = 1)$ is different from $v(x_h)$, and to be~$0$
  otherwise.  Then we have that the next state is $t_h$ if, and only~if,
  $v(x_h)$ is true.
\item If $h$ is even, then the state after~$s_h$ (actually after
  $\mathsf{h}(v')$) is~$t_h$ if $v(x_h)$ is true, and $f_h$ otherwise,
  and this cannot be changed by player~$A_k$.  Then in $t_h$ and $f_h$
  we set $\sigma'_{A_k}(\mathsf{h}(v') t_h)$
  (resp. $\sigma'_{A_k}(\mathsf{h}(v') f_h)$) to be~$1$.  Note that
  since $v(x_k)$ is true we cannot reach~$f_k$, hence setting the
  action of $A_k$ in those states to~$1$ always ensures that the next
  state is~$s_{h+1}$.
\end{itemize}
This shows that $\mathsf{h}(v) \in
\FOut(\replaceter{\sigma}{A_k}{\sigma'_{A_k}})$ for some strategy
$\sigma'_{A_k}$.

The second property can be proven similarly for player~$B_k$.

\medskip 

Let $v$ be a valuation satisfying condition~\eqref{eq}. We show that
$\psi$ evaluates to \true under that valuation.  Let $c_i$ be a clause
of~$\psi$, and let $j= \sigma_\Eve(\mathsf{h}(v) \cdot b_1 \cdots b_l
\cdot c_l)$. We show that $v(\ell_{i,j}) = \true$, which means that
$c_i$ evaluates to \true under~$v$. This will show that formula~$\psi$
is valid since condition~\eqref{eq} defines sufficiently many witness
valuations. Assume w.l.o.g. that $\ell_{i,j}=x_k$. Assume towards a
contradiction that $v(x_k)=\false$. We~have proven that there is a
strategy $\sigma'_{B_k}$ for player~$B_k$ such that $\mathsf{h}(v)
\cdot b_1 \cdots b_i \cdot c_i \in
\FOut(\replaceter{\sigma}{B_k}{\sigma'_{B_k}})$. Now, the state~$x_k$ is
controlled by player~$B_k$, so $B_k$ can enforce a visit to~$w$ from $x_k$,
so there is a strategy $\sigma''_{B_k}$ for player~$B_k$ such that
$\mathsf{h}(v) \cdot b_1 \cdots b_l \cdot c_l \cdot x_k \in
\FOut(\replaceter{\sigma}{B_k}{\sigma''_{B_k}})$. This contradicts the fact
that $\sigma_\Agt$ is a Nash equilibrium. We conclude that $v(x_k) =
\true$, and we conclude that $\psi$ is valid (as explained above).
\qed

\end{document}